\documentclass[11pt]{article}

\usepackage{geometry}
\usepackage[T1]{fontenc}
\usepackage{graphicx}
\usepackage[linesnumbered,boxed,ruled,vlined,algo2e,resetcount,algosection]{algorithm2e}

\usepackage{xspace}
\usepackage[utf8]{inputenc}
\usepackage{amsmath, amssymb, amsfonts, amsthm,mathtools}
\usepackage{booktabs}
\usepackage{enumerate}
\usepackage{xcolor}
\usepackage[colorlinks,linkcolor=magenta,citecolor=blue,bookmarks,bookmarksopen,bookmarksnumbered]{hyperref}
\usepackage{bbm}
\usepackage{ebproof}
\usepackage{paralist}
\usepackage{thmtools, thm-restate}
\usepackage{rotating}
\usepackage{mdframed}
\usepackage{multirow}
\usepackage{fancyvrb}
\usepackage{mathrsfs}
\usepackage{indentfirst}
\usepackage{mleftright}
\usepackage[nottoc]{tocbibind}%
\usepackage[langfont=roman]{complexity}
\usepackage{tablefootnote}
\usepackage{relsize}
\usepackage{exscale}
\usepackage{tikz}
\usetikzlibrary{positioning}
\usepackage{afterpage}

\usepackage{enumitem}

\usepackage{soul}

\usepackage[backend=biber, style=alphabetic, backref=true, maxbibnames=99, maxalphanames=4, minalphanames=3]{biblatex}

\newbibmacro{string+doiurl}[1]{%
  \iffieldundef{doi}{%
    \iffieldundef{url}{%
      #1%
    }{%
      \href{\thefield{url}}{#1}%
    }%
  }{%
    \href{https://doi.org/\thefield{doi}}{#1}%
  }%
}

\DeclareFieldFormat{title}{\usebibmacro{string+doiurl}{\mkbibemph{#1}}}
\DeclareFieldFormat[article,incollection,inproceedings,thesis,book]{title}{\usebibmacro{string+doiurl}{\mkbibquote{#1}}}

\DeclareFieldFormat
  [article,inbook,incollection,inproceedings,patent,thesis,unpublished]
  {titlecase}{%
    \ifcurrentfield{title}{\MakeSentenceCase*{#1}}{%
      \ifcurrentfield{subtitle}{\MakeSentenceCase*{#1}}{#1}%
    }%
  }

\usepackage[capitalize]{cleveref}%

\usepackage{centernot}

\newtheorem{theorem}{Theorem}[section]
\newtheorem{lemma}[theorem]{Lemma}
\newtheorem{corollary}[theorem]{Corollary}
\newtheorem{claim}[theorem]{Claim}

\newtheorem{observation}[theorem]{Observation}
\newtheorem{fact}[theorem]{Fact}
\newtheorem{question}[theorem]{Question}

\newtheorem{proposition}[theorem]{Proposition}

\declaretheoremstyle[
spaceabove=\topsep, spacebelow=\topsep,
headfont=\normalfont\bfseries,
notefont=\bfseries, notebraces={}{},
bodyfont=\normalfont\itshape,
postheadspace=0.5em,
name={\ignorespaces},
numbered=no,
headpunct=.]
{mystyle}

\theoremstyle{definition}
\newtheorem{definition}[theorem]{Definition}

\newtheorem{example}[theorem]{Example}

\newtheorem{assumption}[theorem]{Assumption}

\theoremstyle{remark}
\newtheorem{remark}[theorem]{Remark}

\makeatletter
\def\moverlay{\mathpalette\mov@rlay}
\def\mov@rlay#1#2{\leavevmode\vtop{%
		\baselineskip\z@skip \lineskiplimit-\maxdimen
		\ialign{\hfil$\m@th#1##$\hfil\cr#2\crcr}}}
\newcommand{\charfusion}[3][\mathord]{
	#1{\ifx#1\mathop\vphantom{#2}\fi
		\mathpalette\mov@rlay{#2\cr#3}
	}
	\ifx#1\mathop\expandafter\displaylimits\fi}
\makeatother



\DeclareMathOperator{\GF}{GF}

\newcommand{\Ex}{\operatornamewithlimits{\mathbb{E}}}
\renewcommand{\poly}{\mathrm{poly}}

\renewcommand{\polylog}{\mathrm{polylog}}

\newcommand{\eps}{\varepsilon}

\newlang{\MCSP}{MCSP}
\newlang{\MOCSP}{MOCSP}
\newlang{\MFSP}{MFSP}
\newlang{\MKtP}{MKtP}
\newlang{\MKTP}{MKTP}
\newlang{\itrMCSP}{itrMCSP}
\newlang{\itrMKTP}{itrMKTP}
\newlang{\itrMINKT}{itrMINKT}
\newlang{\MINKT}{MINKT}
\newlang{\MINK}{MINK}
\newlang{\MINcKT}{MINcKT}
\newlang{\CMD}{CMD}
\newlang{\DCMD}{DCMD}
\newlang{\CGL}{CGL}
\newlang{\PARITY}{PARITY}
\renewlang{\Gap}{Gap}
\newlang{\Avoid}{\textnormal{\textsc{Avoid}}}
\newlang{\RemotePoint}{\textnormal{\textsc{Remote-Point}}}
\newlang{\LossyCode}{\textsc{Lossy-Code}}
\newlang{\MissingString}{\textsc{Missing-String}}
\newlang{\SinkOfDAG}{\textsc{Sink-Of-DAG}}
\newlang{\Iter}{\textsc{Iter}}
\newlang{\Palindromes}{\textsc{Palindromes}}
\newlang{\Sparsification}{\textsc{Sparsification}}
\newlang{\HamEst}{\mathsf{HammingEst}}
\newlang{\HamHit}{\mathsf{HammingHit}}
\newlang{\CktEval}{\textsc{Circuit-Eval}}
\newlang{\Hard}{\textsc{Hard}}
\newlang{\cHard}{\textsc{cHard}}
\newlang{\Lin}{\textsc{Lin}}
\newlang{\CAPP}{CAPP}
\newlang{\GapUNSAT}{GapUNSAT}
\newlang{\OV}{OV}
\newlang{\PRIMES}{PRIMES}
\renewlang{\PCP}{PCP}
\newlang{\PCPP}{PCPP}
\newclass{\FMA}{FMA}
\newclass{\Avg}{Avg}
\newclass{\ZPEXP}{ZPEXP}
\newclass{\DLOGTIME}{DLOGTIME}
\newclass{\ALOGTIME}{ALOGTIME}
\newclass{\ATIME}{ATIME}%
\newclass{\SZKA}{SZKA}
\newclass{\Laconic}{Laconic\text{-}}
\newclass{\APEPP}{APEPP}
\newclass{\SAPEPP}{SAPEPP}
\newclass{\TFSigma}{TF\Sigma}
\newclass{\NTIMEGUESS}{NTIMEGUESS}
\newclass{\FZPP}{FZPP}
\newclass{\SearchNP}{SearchNP}
\newclass{\SearchAM}{SearchAM}
\newclass{\UEoPL}{UEoPL}
\newclass{\EoPL}{EoPL}
\newclass{\SoPL}{SoPL}
\newclass{\CLS}{CLS}
\newclass{\PWPP}{PWPP}

\newlang{\Formula}{Formula}
\newlang{\THR}{THR}
\newcommand{\DNF}{\mathsf{DNF}}

\newlang{\ETHR}{ETHR}
\newlang{\Midbit}{Midbit}
\newlang{\LCS}{LCS}
\newlang{\TAUT}{TAUT}
\newcommand{\PartialHard}{\textnormal{\textsc{Partial-Hard}}}
\newcommand{\PartialAvgHard}{\textnormal{\textsc{Partial-AvgHard}}}

\newcommand{\Proj}{\mathrm{Proj}}

\newcommand{\AND}{\mathsf{AND}}

\newcommand{\XOR}{\mathsf{XOR}}

\newcommand{\Supp}{\mathrm{Supp}}

\newcommand{\calA}{\mathcal{A}}
\newcommand{\calB}{\mathcal{B}}

\newcommand{\calD}{\mathcal{D}}

\newcommand{\calF}{\mathcal{F}}
\newcommand{\calG}{\mathcal{G}}
\newcommand{\calH}{\mathcal{H}}

\newcommand{\calP}{\mathcal{P}}

\newcommand{\calR}{\mathcal{R}}

\newcommand{\calX}{\mathcal{X}}

\newcommand{\calZ}{\mathcal{Z}}

\newcommand{\N}{\mathbb{N}}

\newcommand{\F}{\mathbb{F}}

\newcommand{\ckt}{\mathscr{C}}

\newcommand{\Eval}{\mathsf{Eval}}%

\newcommand{\Range}{\mathrm{Range}}

\newcommand{\Adv}{\mathsf{Adv}}

\newcommand{\TT}{\mathsf{TT}}

\newcommand{\Krajicek}{Kraj\'{\i}\v{c}ek\xspace}

\newcommand{\Prover}{\mathsf{Prover}}
\newcommand{\Verifier}{\mathsf{Verifier}}

\newcommand{\val}{\mathrm{val}}

\newcommand{\cbra}[1]{\mleft\{ #1 \mright\}}
\newcommand{\sbra}[1]{\mleft[ #1 \mright]}
\newcommand{\pbra}[1]{\mleft( #1 \mright)}
\newcommand{\abra}[1]{\mleft< #1 \mright>}

\usepackage[inline,marginclue,index]{fixme}  %
\fxsetup{theme=color,mode=multiuser}

\definecolor{color1}{RGB}{46,134,193}
\FXRegisterAuthor{hanlin}{ahanlin}{\colorbox{color1}{\color{white}Hanlin}}
\newcommand{\hanlin}[1]{{\hanlinnote{#1}}}

\FXRegisterAuthor{yan}{ayan}{\colorbox{cyan}{\color{white}Yan}}
\newcommand{\yan}[1]{{\yannote{#1}}}

\FXRegisterAuthor{xin}{axin}{\colorbox{magenta}{\color{white}Xin}}

\FXRegisterAuthor{afterddl}{aafterddl}{\colorbox{blue}{\color{white}After Deadline}}

\usepackage{aliascnt} %

\bibliography{ref}
\def\Anonymity{0}%

\begin{document}

\newgeometry{margin=0.85in}
\title{Many Proof Complexity Generators Inside One Demi-Bits Generator}
\ifnum\Anonymity=0
\author{
    Xin Li\\ \small{Johns Hopkins University}\\ \small{\texttt{\href{mailto:lixints@cs.jhu.edu}{lixints@cs.jhu.edu}}}
    \and
	Hanlin Ren\footnote{Hanlin Ren is supported by the Massive Dynamics Member Fund at the Institute for Advanced Study.}\\ \small{Institute for Advanced Study} \\ \small{\texttt{\href{mailto:h4n1in.r3n@gmail.com}{h4n1in.r3n@gmail.com}}}
	\and
    Yan Zhong\\ \small{Johns Hopkins University} \\ \small{\texttt{\href{mailto:yanzhong.cs@gmail.com}{yanzhong.cs@gmail.com}}}
}
\fi

\pagenumbering{gobble}

\maketitle
\begin{abstract}
    For a propositional proof system $\mathcal{P}$ and a polynomial-time function $G: \{0, 1\}^n \to \{0, 1\}^N$ ($N > 10n$), we say that $G$ is a \emph{proof complexity generator} against $\mathcal{P}$ if $\mathcal{P}$ cannot efficiently prove the (suitably encoded) statement ``$y\not\in\mathrm{Range}(G)$'' for every $y \in \{0, 1\}^N$, and $G$ is a \emph{demi-bits generator} against $\mathcal{P}$ if $\mathcal{P}$ cannot efficiently prove the statement ``$y \not\in\mathrm{Range}(G)$'' for a noticeable fraction of $y \in \{0, 1\}^N$. As can be seen from the definitions, proof complexity generators are qualitatively stronger objects than demi-bits generators.

    Our main result is that, perhaps counter-intuitively, every demi-bits generator ``contains'' exponentially many proof complexity generators. In fact, a random subset of output bits of a demi-bits generator forms a proof complexity generator with constant probability. This result is an extremely simple corollary of Pajor's Lemma (a strengthening of the well-known Sauer--Shelah Lemma).

    This result allows us to exhibit the hardness of the Range Avoidance problem ($\Avoid$) in several new, restricted settings of interest:
    \begin{itemize}
        \item We show that demi-bits generators computable in $\NC^0$ imply the hardness of $\NC^0$-$\Avoid$ up to constant factors in the stretch, demonstrating a barrier to further improvements on the recent progress on this problem (Korten--Pitassi--Impagliazzo, FOCS'25; Guruswami--Lyu--Yuan, SODA'26).
        \item Given a linear space $V\subseteq \GF(2)^n$ of dimension $k$, the $\XOR$-$\RemotePoint$ problem asks to find a vector far from $V$ (Alon--Panigrahy--Yekhanin, RANDOM'09). Assuming a demi-hardness version of LPN (Learning Parity with Noise), we show that $\XOR$-$\RemotePoint$ cannot be solved by efficient nondeterministic algorithms.
        \item An intriguing challenge in circuit complexity is to build a partial Boolean function on a given domain that has high circuit complexity (Arvind--Srinivasan, ICS'10; Chen--Huang--Li--Ren, STOC'23). Even for hardness against \emph{polynomial-size DNFs}, no efficient algorithm is known for this task. We show that under a version of the random $k$-SAT Hypothesis against $\AM$ algorithms, such hard functions cannot be constructed in nondeterministic polynomial time.
    \end{itemize}
    
    Along the way, we introduce the notion of \emph{zero-error disperser families} that can transform demi-bits generators into proof complexity generators, and show that this family can be computed by projections (i.e., without any circuit complexity overhead). Using a different instantiation of this family, we construct proof complexity generators from \emph{barely non-trivial} demi-bits generators $G: \{0, 1\}^n \to \{0, 1\}^N$, where the number of hard-to-prove statements of the form ``$y \not \in \Range(G)$'' just slightly exceeds $2^n$ (which is the number of \emph{false} statements of this form).
\end{abstract}

\hanlin{I will go through the paper in the next weeks. I'll \emph{italicize} everything that I think are good: \emph{Abstract}, \emph{Section 1}, \emph{Section 2}, \emph{Section 3}, \emph{Section 4}, \emph{Section 5} (up to the comments in it), \emph{Section 6} (I left a comment at the end of sec 6.1; corollary 6.16 is incorrect; everything else seems good), \emph{Appendix A} (since we may want to add more proofs inside appendix A, let me write down what I verified precisely: the proofs of \autoref{fact: AM adversaries to NPpoly adversaries} and \autoref{thm: reduction in DS16}), \emph{appendix D} (I plan to move it to appendix A somehow) other appendices}

\newpage
\tableofcontents
\newpage

\pagenumbering{arabic}

\section{Introduction}

The \emph{Range Avoidance} problem ($\Avoid$) is the following computational problem: Given the description of a circuit $C: \{0, 1\}^n \to \{0, 1\}^\ell$ as input where $\ell > n$, the task is to find some $\ell$-bit string $y$ that is not in the range of $C$. Introduced in~\cite{KKMP21}, this problem has received significant attention recently due to its connections to explicit construction problems~\cite{Korten21, RenSW22, GuruswamiLW22, GajulapalliGNS23, ChenHLR23, LZ25} and circuit lower bounds~\cite{Korten21, CHR24, Li24, KortenP24}. We refer to~\cite{Korten-EATCS} for a comprehensive survey on this problem.%

It was observed in~\cite{Korten21} that the explicit construction problems for many important objects in mathematics and theoretical computer science reduce to $\Avoid$; hence, a deterministic algorithm for $\Avoid$ would immediately resolve a plethora of them in one shot. These include constructions of Ramsey graphs~\cite{Erdos59}, rigid matrices~\cite{Valiant77}, two-source extractors~\cite{ChorG88, CZ19, Li23}, good binary codes~\cite{Gilbert52, Varshamov57, Ta-Shma17}, and much more, all with near-optimal parameters. %

\begin{mdframed}[innertopmargin = -0.3em, skipabove=-0.3em]
\small
\begin{example}\label{ex: circuit lower bounds}
    The most notorious example of explicit construction problems is, perhaps, \emph{circuit lower bounds}. Already in 1949, Shannon~\cite{Shannon49} proved that almost all Boolean functions $f: \{0, 1\}^n \to \{0, 1\}$ require circuits of $\Omega(2^n / n)$ size to compute. However, despite decades of efforts, the best lower bounds for explicit functions against general circuits are still only slightly above $3n$~\cite{FindGHK16, LY22}.

    Fix a size bound $s \ll 2^n / n$ and consider the following \emph{truth table generator} $\TT: \{0, 1\}^{O(s \log s)} \to \{0, 1\}^{2^n}$ that takes the description of a size-$s$ circuit as input and outputs the length-$2^n$ truth table of this circuit. It is easy to see that $\Range(\TT)$ consists of exactly the truth tables with circuit complexity at most $s$. Hence, proving a circuit lower bound is equivalent to solving $\Avoid$ on the particular input instance $\TT$. For example, a deterministic algorithm for solving $\Avoid$ on $\TT$ is equivalent to a circuit lower bound for $\E := \DTIME[2^{O(n)}]$~\cite{Korten21, RenSW22}.
\end{example}
\end{mdframed}

\begin{mdframed}[innertopmargin=-0.3em, skipabove=-0.3em]
    \small
    \begin{example}\label{ex: rigid matrix}
        Fix a finite field $\F$. A matrix $A \in \F^{n\times n}$ is \emph{$(r, \Delta)$-rigid} if it has Hamming distance at least $\Delta$ to every matrix of rank at most $r$. Rigid matrices were introduced by Valiant~\cite{Valiant77} where it was shown that a random matrix is rigid with high probability and that an explicit construction of rigid matrices would imply several new circuit lower bounds. Since then, matrix rigidity has become an important topic in complexity theory (for a recent survey, see~\cite{RamyaSurvey}). However, it is still open to construct a rigid matrix with good parameters, say a $(0.1n, 0.1n^2)$-rigid matrix over $\F = \GF(2)$.

        Constructing a rigid matrix in the above parameter regime reduces to the following instance $C_{\rm Rigid}$ of $\Avoid$~\cite{Korten21, GuruswamiLW22}. The input of $C_{\rm Rigid}$ consists of two matrices $X, Y \in \F^{n\times 0.1n}$ and a sparse matrix $A$ containing at most $0.1n^2$ nonzeros, and the output is the $n^2$-bit string describing $XY^T + A$. One can encode $A$ efficiently so that the total input length is less than $n^2$, hence $C_{\rm Rigid}$ is a valid $\Avoid$ instance.
    \end{example}
\end{mdframed}

Ren, Santhanam, and Wang~\cite{RenSW22} initiated the study of $\Avoid$ for \emph{restricted} circuit classes. One motivation for such a study is that reductions from concrete explicit construction problems to $\Avoid$ tend to produce restricted circuits: For example, using tools from succinct data structures~\cite{Patrascu08, Yu22}, Guruswami, Lyu, and Wang~\cite{GuruswamiLW22} showed that the circuit $C_{\rm Rigid}$ in \autoref{ex: rigid matrix} can be computed in $\NC^1$. Furthermore, applying a result in~\cite{RenSW22} (proven by the \emph{randomized encoding} technique~\cite{IshaiK00, IshaiK02, ApplebaumIK06}), they further reduced the circuit complexity to $\NC^0_4$ (i.e., each output bit in the $\Avoid$ instance only depends on \emph{four} input bits). The circuit complexity of $C_{\rm Rigid}$ can be reduced even further to $\NC^0_3$~\cite{GajulapalliGNS23}. Thus, a seemingly promising research program towards explicit constructions is:
\begin{quote}
    For each explicit construction problem $\Pi$ of interest, reduce $\Pi$ to an $\Avoid$ instance $C_\Pi: \{0, 1\}^n \to \{0, 1\}^\ell$ with the \emph{smallest possible circuit complexity} and the \emph{largest possible stretch function} $\ell = \ell(n)$. Then, design efficient deterministic algorithms for $\Avoid$ that can handle larger and larger circuit classes and smaller and smaller stretch functions.
\end{quote}

This program turns out to be highly non-trivial to realize even for the simplest circuit class one may consider: functions that only depend on $k = O(1)$ many input bits ($\NC^0_k$). Guruswami, Lyu, and Wang~\cite{GuruswamiLW22} made the first progress on the $\NC^0_k$-$\Avoid$ problem by providing an $\FP^\NP$ algorithm when the stretch is $\ell(n) \ge \Omega_k(n^{k-1})$. The complexity of $\NC^0_k$-$\Avoid$ has received much attention recently~\cite{GajulapalliGNS23, KuntewarS25, LZ25}, and the state-of-the-art algorithm runs in deterministic polynomial time when the stretch is $\ell(n) \ge \Omega_k(n^{(k-1)/2}\log n)$~\cite{KPI25, GLY26}.

\subsection{Special Cases of Range Avoidance}
Prior to the introduction of the Range Avoidance problem~\cite{KKMP21}, there were already several special cases studied in the literature. We discuss two of them studied in this paper. In some sense, both were motivated by (our inability to prove) circuit lower bounds and explicit constructions.

\paragraph{The remote point problem.} Given a linear space $L\subseteq \GF(2)^\ell$ of dimension $n < \ell$, the \emph{Remote Point} problem ($\RemotePoint$) asks to find a vector $v \in \GF(2)^\ell$ that is far from $L$. This problem was introduced by Alon, Panigrahy, and Yekhanin~\cite{AlonPY09}, motivated by the task of explicitly constructing rigid matrices~\cite{Valiant77}. They proposed the remote point problem as ``an intermediate challenge that is less daunting than'' that of constructing rigid matrices and provided a deterministic algorithm achieving certain parameters. Unfortunately, to the best of our knowledge, algorithmic progress on $\RemotePoint$ so far has not resulted in better constructions of rigid matrices.

For any linear space $L \subseteq \GF(2)^\ell$ of dimension $n$, there is a circuit $C: \{0, 1\}^n \to \{0, 1\}^\ell$ consisting of only $\XOR$ gates such that $\Range(C) = L$. This motivates a more general version of the remote point problem~\cite{KKMP21}: Given the description of a circuit $C: \{0, 1\}^n \to \{0, 1\}^\ell$ (from some circuit class $\ckt$) and a parameter $\delta \in (0, 1/2)$ as input, find a string $y\in\{0, 1\}^\ell$ that is $\delta$-far from $\Range(C)$. Under suitable parameter regimes, the existence of such $y$ can be proven by a union bound, and the proof gives a reduction from $\RemotePoint$ to $\Avoid$. Following~\cite{ChenHLR23}, in the rest of this paper, we will use $\RemotePoint$ to denote the general version with a circuit as input, and refer to the special case studied in~\cite{AlonPY09} as $\XOR$-$\RemotePoint$.

\paragraph{Finding hard partial truth tables.} Fix a circuit class $\ckt$ and a size bound $s$. In the \emph{Hard Partial Truth Table} problem, given a ``domain'' $x_1, x_2, \dots, x_L \in \{0, 1\}^n$ ($L \gg s$), the goal is to output a ``partial truth table'' $b_1, b_2, \dots, b_L \in \{0, 1\}$ such that the partial function $\{x_i \mapsto b_i\}$ is hard against size-$s$ $\ckt$ circuits; that is, for any such circuit $C: \{0, 1\}^n \to \{0, 1\}$, there exists an index $i\in [L]$ such that $C(x_i) \ne b_i$. It is easy to see that under suitable choice of parameters, this problem also reduces to $\Avoid$ (see, e.g.,~\cite{ChenHLR23}). We denote this problem by $\ckt$-$\PartialHard$.

This problem was introduced in~\cite{ArvindS10} under the name ``circuit lower bounds with help functions,'' where it was observed that for $\ckt = \AC^0$, this problem reduces to the remote point problem in certain parameter regimes. This problem turns out to be stubbornly hard even for simple circuit classes $\ckt$ such as \emph{polynomial-size DNFs}. Even though it is trivial to prove a $2^{\Omega(n)}$ size lower bound for DNFs, we are not aware of any polynomial-time algorithm for $\DNF$-$\PartialHard$ in the regime of $s \ll L \le \poly(n)$.

Besides being interesting in its own right, this problem was also motivated by closures of non-uniform circuit classes under uniform reductions: a polynomial-time algorithm for $\ckt$-$\PartialHard$ implies a problem $L \in \E$ that cannot be (polynomial-time) mapping reduced to any problem in (the non-uniform complexity class) $\ckt$ (see \cite[Theorem 5]{ArvindS10}). The reduction is allowed to run in (uniform) polynomial time (which can potentially be more powerful than $\ckt$), hence such statements seem strictly stronger than circuit lower bounds for $\ckt$. By providing an $\FP^\NP$ algorithm for $\ACC^0$-$\PartialHard$ in certain parameter regimes, \cite{ChenHLR23} proved that there is a language in $\E^\NP$ that does not mapping reduce to (non-uniform) $\ACC^0$. However, due to our inability to solve $\DNF$-$\PartialHard$ in polynomial time, even the following innocent-looking question is open:
\begin{question}
    Is it possible that for every language $L \in \E$, there is a language $L'$ computable by a (non-uniform family of) polynomial-size DNF, such that $L \le_m L'$?
\end{question}

\subsection{Is \texorpdfstring{$\Avoid$}{AVOID} Hard?}\label{sec: intro hardness of avoid}

Our inability to solve restricted versions of $\Avoid$ leads to the possibility that these versions might be \emph{hard} to solve. For example, could it be possible that there exists no deterministic polynomial-time algorithm for $\NC^0_{100}$-$\Avoid$ with stretch $n^{1.01}$?

At first glance, $\Avoid$ seems to be an easy problem since it can be solved by the following trivial randomized algorithm: output a random string $y \gets \{0, 1\}^\ell$ and it is not in $\Range(C)$ with high probability. However, since it is unclear how to \emph{verify} whether a string $y$ is in $\Range(C)$ or not (indeed, this is $\NP$-hard), it is also unclear how to turn this trivial randomized algorithm into a deterministic algorithm, even under standard derandomization assumptions~\cite{NisanW94, ImpagliazzoW97, Umans03}.

Somewhat surprisingly, we now have cryptographic evidence that this trivial randomized algorithm for $\Avoid$ \emph{cannot} be made deterministic. The first such evidence was given by Ilango, Li, and Williams~\cite{DBLP:conf/stoc/IlangoLW23}. They showed that if subexponentially-secure indistinguishability obfuscation exists~\cite{BarakGIRSVY12, GGHRSW16, JainLS21} and $\NP\ne\coNP$, then there is no deterministic polynomial-time algorithm for $\Avoid$. Based on this work, Chen and Li~\cite{DBLP:conf/stoc/ChenL24} showed that $\Avoid$ does not have \emph{nondeterministic polynomial-time} ($\SearchNP$) algorithms under certain hardness assumptions about LPN (Learning Parity with Noise) or LWE (Learning with Errors) against nondeterministic algorithms. More recently, a pair of concurrent works~\cite{Ilango25, RenWZ26} strengthened the results in~\cite{DBLP:conf/stoc/ChenL24}, showing that $\Avoid\not\in\SearchNP$ if there exist \emph{demi-bits generators} with suitable stretch~\cite{Rudich97} (an assumption weaker than that used in~\cite{DBLP:conf/stoc/ChenL24}). We will discuss this assumption soon in \autoref{sec: intro demi-bits generators vs proof complexity generators}.

\paragraph{Hardness for restricted circuits?} In light of the aforementioned research program, it is important to understand for which circuit classes $\ckt$ and which stretch functions $\ell(n)$ the $\ckt$-$\Avoid$ problem is hard. Except for~\cite{DBLP:conf/stoc/IlangoLW23}, all other three works imply hardness of restricted $\Avoid$. %

\begin{itemize}
    \item The construction in \cite{DBLP:conf/stoc/ChenL24} crucially relies on the structure of LWE and LPN. In particular, under a nondeterministic variant of the LPN assumption, \cite{DBLP:conf/stoc/ChenL24} showed that the $\RemotePoint$ problem for $\XOR\circ\AND_{O(\log n)}$ circuits (i.e., $O(\log n)$-degree polynomials over $\GF(2)$) is hard for $\SearchNP$ algorithms. However, \cite{DBLP:conf/stoc/ChenL24} mentioned that the hardness of $\XOR$-$\RemotePoint$, i.e., the original ``Remote Point Problem'' considered in~\cite{AlonPY09}, remains open.
    \item Given a secure demi-bits generator $G: \{0, 1\}^n \to \{0, 1\}^N$ and a pairwise independent hash family $\calH = \{h: \{0, 1\}^N \to \{0, 1\}^m\}$ (where $m > n$), \cite{RenWZ26} shows that any $\SearchNP$ algorithm fails to solve $\Avoid$ on some instance of the form $h\circ G$ where $h \in \calH$. Pairwise independent hash families can be computed by $\GF(2)$-linear circuits (circuits consisting of only $\XOR$ gates), therefore under suitable assumptions, \cite{RenWZ26} obtained hardness of $\Avoid$ where each output bit is computable in $\XOR\circ\AND_{O(1)}$ (i.e., a constant-degree $\GF(2)$-polynomial).
    \item Given a secure demi-bits generator $G: \{0, 1\}^n \to\{0, 1\}^N$, it is implicit in \cite{Ilango25} that any $\SearchNP$ algorithm fails to solve $\Avoid$ on some instance of the form
    \[C_{x_1, \dots, x_t}(s, i) := G(s) \oplus x_i,\]
    where $t = O(N)$ and $x_1, \dots, x_t\in \{0, 1\}^N$. This allows \cite{Ilango25} to refute the \emph{Oracle Derandomization Hypothesis} of Fortnow and Santhanam~\cite{FortnowS11} under suitable assumptions.
\end{itemize}

Unfortunately, one can see that due to the overhead in the constructions, none of the above works have the potential of implying the hardness of $\NC^0$-$\Avoid$. Ilango's result is the closest one: assuming demi-bits generators computable in $\NC^0$, \cite{Ilango25} implies that there is no $\SearchNP$ algorithm for $\Avoid$ on circuits with $O(\log N)$ locality. Furthermore, the techniques of~\cite{RenWZ26, Ilango25} do not seem applicable to the Remote Point Problem.

\subsection{Demi-Bits Generators and Proof Complexity Generators}\label{sec: intro demi-bits generators vs proof complexity generators}

The above line of work on the hardness of $\Avoid$ can be equivalently stated as connections between demi-bits generators and proof complexity generators, both of which are important objects in their own right. In this sub-section, a \emph{generator} is an efficiently computable function $G: \{0, 1\}^n \to \{0, 1\}^N$ with $N > n$, and an \emph{adversary} trying to break the generator $G$ is a \emph{nondeterministic} polynomial-time algorithm $\calA$ trying to \emph{certify} a certain string $y \in \{0, 1\}^N$ is not in $\Range(G)$. Equivalently, one can think of the adversary as a (Cook--Reckhow~\cite{CookR79}) \emph{proof system} by treating nondeterministic bits as ``proofs'' that $y\not\in\Range(G)$ (see~\autoref{def: Cook-Reckhow proof systems}). We also assume that nondeterministic adversaries discussed here correspond to \emph{sound} proof systems, i.e., for every $y \in \Range(G)$, they never prove the incorrect statement ``$y\not\in\Range(G)$''.

\paragraph{Demi-bits generators} are generators $G$ such that there is no efficient nondeterministic adversary $\calA$ that accepts a large fraction of strings $y \in \{0, 1\}^N$ but rejects every string $y \in \Range(G)$. That is, there is no sound proof system that proves (a suitably encoded version of) ``$y\not\in\Range(G)$'' efficiently for a large fraction of $y \in \{0, 1\}^N$. Introduced by Rudich~\cite{Rudich97}, demi-bits generators can be seen as the analogue of cryptographic \emph{hitting set generators} against efficient nondeterministic adversaries.\footnote{We note that Rudich~\cite{Rudich97} also introduced the notion of \emph{super-bits generators}, which is the analogue of cryptographic pseudorandom generators against efficient nondeterministic adversaries.} 

Demi-bits generators are also related to the \emph{natural proof} barrier of Razborov and Rudich~\cite{RazborovR97}. To prove a circuit lower bound, we need to design a property $\calP$ over truth tables of Boolean functions such that every Boolean function satisfying $\calP$ requires large circuits. We say that $\calP$ is a \emph{natural property}~\cite{RazborovR97} if it satisfies the following two requirements:
\begin{itemize}
    \item {\bf Constructivity:} Given the length-$2^n$ truth table of a Boolean function $f$, it is possible to decide whether $f$ satisfies $\calP$ in $\poly(2^n)$ time.
    \item {\bf Largeness:} A constant fraction of length-$2^n$ truth tables satisfy $\calP$.
\end{itemize}

\begin{mdframed}[innertopmargin=-0.3em, skipabove=-0.3em]
    \small
    \begin{example}\label{ex: TT as HSG}
        The main result of~\cite{RazborovR97} is that if secure one-way functions exist, then there is no natural property $\calP$ for proving lower bounds against general circuits ($\P/_\poly$). Roughly speaking, this is because the existence of one-way functions implies that $\TT$, the truth table generator, is a cryptographic hitting set generator~\cite{DBLP:journals/jacm/GoldreichGM86}, while every natural property against $\P/_\poly$ can be used to break $\TT$ as a hitting set generator.
    \end{example}
\end{mdframed}

\begin{mdframed}[innertopmargin=-0.3em, skipabove=-0.3em]
    \small
    \begin{example}\label{ex: rigid matrix as HSG}
        Under the LPN assumption, the circuit $C_{\rm Rigid}$ constructed in~\autoref{ex: rigid matrix} is a cryptographic pseudorandom generator. (Proof sketch: if one can distinguish $XY^T + A$ from random where $X, Y \in \F_2^{N\times N^{0.1}}$ and $A$ is some sparse random noise, then one can apply a hybrid argument and, given $X$, distinguish $Xs + e$ from random, where $s \in \F_2^{N^{0.1}}$ and $e$ is some sparse random noise.)

        Therefore, there is a ``natural proof'' barrier to constructing a rigid matrix: if there is a property $\calP$ of $N\times N$ matrices that is constructive (runs in $\poly(N)$ time), large (accepts a constant fraction of random matrices), and only accepts matrices that are rigid, then one can use $\calP$ to break the LPN assumption.
    \end{example}
\end{mdframed}

Rudich~\cite{Rudich97} considered a notion of $\NP$-constructivity, where the property $\calP$ is only required to be computable in \emph{nondeterministic} $\poly(2^n)$ time. Given the correspondence between proof systems and nondeterministic algorithms~\cite{CookR79}, $\NP$-constructivity is arguably a more suitable notion than (polynomial-time) constructivity in the study of the (metamathematical) difficulty of circuit lower bounds. By definition, if $\TT$ is a secure demi-bits generator, then there is a natural proof barrier to proving circuit lower bounds, even if the property is allowed to use nondeterminism. Similarly, if $C_{\rm Rigid}$ is a secure demi-bits generator, then there is a natural proof barrier to constructing rigid matrices.

\paragraph{Proof complexity generators} against a proof system $\calP$ are generators $G$ satisfying the stronger requirement that $\calP$ cannot efficiently prove ``$y\not\in\Range(G)$'', for \emph{any} $y \in \{0, 1\}^N$.\footnote{It is more natural to talk about proof complexity generators against a \emph{fixed} proof system instead of those against \emph{all} proof systems. One reason is that if the proof system is allowed to be non-uniform, then it can simply hardwire a string $y\not\in\Range(G)$ and add this fact as an axiom. There is also evidence that any uniform proof complexity generator cannot be secure against all uniform proof systems simultaneously~\cite{KortenS25}.} Introduced independently in~\cite{ABRW04} and~\cite{krajivcek2001weak}, this notion was heavily inspired by the notion of \emph{pseudorandomness}~\cite{Yao82}. Indeed, the motivation of~\cite{ABRW04} was to pin down a precise definition of ``pseudorandom generators against proof systems,'' and~\cite[Section 7]{krajivcek2001weak} conjectured that the existence of pseudorandom generators implies the existence of proof complexity generators secure against Extended Frege. The influence of pseudorandomness goes far beyond these two papers, as can be seen, for example, in the literature studying the Nisan--Wigderson generator (\cite{NisanW94}) as a proof complexity generator~\cite{ABRW04, Krajicek04, razborov2015pseudorandom, Pich11, Krajicek11, Kra12, Pich15, Khaniki22}. A comprehensive survey on proof complexity generators can be found in~\cite{Krajicek-generator-book}.

Motivated by the apparent relationship between pseudorandomness and proof complexity generators, Razborov~\cite{razborov2015pseudorandom} asked the following question: For what family of generators $G: \{0, 1\}^n \to \{0, 1\}^N$ can we base their hardness as proof complexity generators on their ``\emph{computational or combinatorial}'' hardness? Razborov coined this approach towards proof complexity lower bounds as the \emph{generator approach}.\footnote{We remark that our paper, as well as the recent line of works~\cite{DBLP:conf/stoc/ChenL24, RenWZ26, Ilango25}, does not contribute to the exact questions asked in~\cite{razborov2015pseudorandom} since we still need hardness assumptions against \emph{nondeterministic} adversaries.} Razborov also conjectured that plugging suitable hard functions (in the \emph{computational} sense) into the Nisan--Wigderson generator would result in proof complexity generators against proof systems such as Frege and Extended Frege~\cite[Conjectures 1 and 2]{razborov2015pseudorandom}. Unfortunately, progress on this conjecture has only been made for restricted proof systems such as $\AC^0$-Frege~\cite{Pich11, Khaniki22}.

\def\lb{\mathsf{lb}}

Apparently, many generators arising from explicit construction problems possess both cryptographic hardness and proof complexity hardness. Our inability to prove circuit lower bounds \emph{for any explicit function} and construct \emph{any} rigid matrices suggests that the generators $\TT$ and $C_{\rm Rigid}$ described in~\autoref{ex: circuit lower bounds} and \autoref{ex: rigid matrix} are plausible proof complexity generators. Indeed, great effort has been devoted to the unprovability of circuit lower bounds~\cite{Razborov_feasible_mathematics_II, DBLP:journals/jsyml/Krajicek97, Razborov98, Raz04, Razborov04, Krajicek04, razborov2015pseudorandom, Pich15, PichS19, SanthanamT21, PichS21, PichS22, GarlikGRT26}, which is equivalent to $\TT$ being a proof complexity generator. On the other hand, as discussed in \autoref{ex: TT as HSG} and \autoref{ex: rigid matrix as HSG}, under standard cryptographic assumptions, $\TT$ and $C_{\rm Rigid}$ are cryptographic hitting set generators as well! Understanding the role of (computational) pseudorandomness in the generator approach is profoundly related to the natural proof barrier, see~\autoref{sec: discussion natural proof largeness} for more discussions. %

Finally, another motivation for \cite{krajivcek2001weak} to define proof complexity generators is the \emph{dual weak pigeonhole principle}: For any function $f: [N] \to [2N]$, there exists at least one element $y \in [2N]$ that is not in the range of $f$. If $G:\{0, 1\}^n \to \{0, 1\}^{n+1}$ is a proof complexity generator against a certain proof system $\calP$, then the dual weak pigeonhole principle is ``effectively false'' for $\calP$! The same principle also underlies the Range Avoidance problem; indeed, $\Avoid$ on $f$ asks precisely to find such an element $y \in [2N]\setminus\Range(f)$. It is easy to see (in fact, \emph{by definition}) that
\begin{observation}[{\cite[Section 6]{RenSW22}}]
    $\Avoid\not\in\SearchNP$ if and only if for every propositional proof system $\calP$, there exists a proof complexity generator $G$ secure against $\calP$.
\end{observation}

This connection allows us to import the hardness results for $\Avoid$ described in \autoref{sec: intro hardness of avoid}~\cite{DBLP:conf/stoc/ChenL24, RenWZ26, Ilango25} to the realm of proof complexity generators. These results need to assume the existence of demi-bits generators, hence they are in fact transformations from demi-bits generators to proof complexity generators. For example:
\begin{theorem}[Informal]\label{thm: previous results}
    Let $G: \{0, 1\}^n \to \{0, 1\}^N$ be a demi-bits generator, and $\calP$ be a ``nice'' proof system. Then:
    \begin{itemize}
        \item \textnormal{(\cite{RenWZ26})} For any suitable family of pairwise independent hash functions $\calH = \{h: \{0, 1\}^N \to \{0, 1\}^m\}$, there exists $h \in \calH$ such that $h\circ G$ is a proof complexity generator against $\calP$.
        \item \textnormal{(\cite{Ilango25})} There exists a family of ``shifts'' $z_1, z_2, \dots, z_{O(n)}$ such that the function
        \[C(s, i) := G(s) \oplus z_i\]
        is a proof complexity generator against $\calP$.
    \end{itemize}
\end{theorem}

However, the same issue regarding circuit complexity overhead mentioned in \autoref{sec: intro hardness of avoid} arises here as well: What is the smallest complexity overhead of constructing proof complexity generators from demi-bits generators? In particular, are there proof complexity generators computable in $\NC^0$?

\subsection{Our Results}
\subsubsection{Proof Complexity Generators from Demi-Bits Generators with No Overhead}

Our main result is the following seemingly surprising statement: 
\begin{quote}
    \it For every demi-bits generator $G$, there are lots of proof complexity generators ``hidden'' in $G$.
\end{quote}

More precisely, let $m := 10n$, $G: \{0, 1\}^n \to \{0, 1\}^m$, and let $I \subseteq [m]$. Define $G|_I: \{0, 1\}^n \to \{0, 1\}^{|I|}$ to be the function consisting of the output bits of $G$ indexed by $I$. We show:
\begin{theorem}[Main; Informal]\label{thm: main informal}
    Let $\calP$ be a ``well-behaved''\footnote{Thanks to the simplicity of our arguments, the definition of ``well-behaved'' here is much more relaxed than the definition of ``nice'' in~\autoref{thm: previous results}; see \autoref{sec: proof complexity generators from demi-bits generators} for more details.} proof system and $G$ be a demi-bits generator against $\calP$. Then with constant probability over a uniformly random subset $I \subseteq [m]$, $G|_I$ is a proof complexity generator against $\calP$.
\end{theorem}

The proof of \autoref{thm: main informal} is remarkably simple: It follows directly from the well-known \emph{Sauer--Shelah} lemma~\cite{Sauer72, Shelah72, VC68}. Recall that a set system $\calF$ over $[m]$ \emph{shatters} a subset of indices $I \subseteq [m]$, if for every subset $J \subseteq I$, there exists a set $S \in \calF$ such that $S\cap I = J$. The Sauer--Shelah lemma (in fact, its strengthening by Pajor~\cite{pajor1985sous}) states that every set system $\calF$ shatters at least $|\calF|$ many subsets.%

\begin{proof}[Proof Idea of \autoref{thm: main informal}]
    Let $H \subseteq \{0, 1\}^m$ denote the set of ``hard'' strings for $\calP$, i.e., the strings $z \in \{0, 1\}^m$ such that $\calP$ does not admit short proofs of ``$z \not\in\Range(G)$''. If $G$ is a demi-bits generator against $\calP$, then $|H| \ge \Omega(2^m)$. Let $I$ be a random subset of $[m]$, then by Pajor's lemma, $H$ shatters $I$ with constant probability. (Also, $|I| > n$ with high probability.)
    
    It is easy to verify that if $H$ shatters $I$, then $G|_I$ is a proof complexity generator against $\calP$: for every $y \in \{0, 1\}^I$, there exists a string $z \in H$ such that $z|_I = y$; hence, if $\calP$ can efficiently prove ``$y \not\in\Range(G|_I)$'', then $\calP$ can efficiently prove that ``$z\not\in\Range(G)$'' as well, a contradiction to $z \in H$.
\end{proof}

We remark that the main results in~\cite{RenWZ26, Ilango25} are also obtained by very simple proofs. In our opinion, the simplicity of our proofs indicates that there are potentially much more excitement in the direction of \emph{cryptography against nondeterministic adversaries} yet to be discovered. Moreover, the versatility of nondeterministic adversaries suggests that many results there are going to be vastly different from our existing intuition in the setting of deterministic adversaries.

\autoref{thm: main informal} gives a transformation from demi-bits generators to proof complexity generators with \emph{zero} circuit complexity overhead, which also allows us to make progress on the hardness of special cases of the Range Avoidance problem.

\paragraph{Range avoidance for $\NC^0$ circuits.} Let $k\ge 3$ be a constant. Assuming there are demi-bits generators $G: \{0, 1\}^n \to \{0, 1\}^{\ell(n)}$ computable in $\NC^0_k$ for $\ell(n) \ge 10n$, we immediately obtain hardness of $\NC^0_k$-$\Avoid$ against $\SearchNP$ algorithms. Moreover, the hard $\Avoid$ instances we obtain have the same stretch ($\Theta(\ell(n))$) as the demi-bits generator $G$ up to a constant factor.

\begin{corollary}[$\NC^0$-$\Avoid$ is Hard]\label{cor: NC0-Avoid}
    Let $k\in\N$, and $\ell(n) > n$ be a stretch function. The following holds for some universal constant $c\ge 1$: If there exists a demi-bits generator $G: \{0, 1\}^n \to \{0, 1\}^{c\ell(n)}$ computable in $\NC^0_k$, then $\NC^0_k$-$\Avoid[n, \ell(n)] \not\in\SearchNP$.
\end{corollary}

Guruswami, Lyu, and Yuan~\cite{GLY26} explicitly asked the open problem of determining the ``computational threshold'' for $\NC^0_k$-$\Avoid$. \autoref{cor: NC0-Avoid} contributes to this problem by reducing it to determining the best stretch of demi-bits generators computable in $\NC^0_k$.

Every pseudorandom generator computable in $\NC^0_k$ with stretch $\Omega_k(n^{\lceil k/2\rceil})$ can be distinguished by a linear test with constant advantage~\cite{DBLP:journals/rsa/MosselST06}. To the best of our knowledge, when $k$ is a large enough constant (say $k\ge 10$), we are not aware of attacks for \emph{general, arbitrary} pseudorandom generators computable in $\NC^0_k$ with a smaller stretch. (In contrast, there is a long line of work on attacking specific $\NC^0_k$ PRGs such as Goldreich's PRG~\cite{CM01,BQ12,OW14,ABR16,AL18,CDMRR18,OST22,unal23}.) Plugging recent algorithmic advances for $\NC^0_k$-$\Avoid$~\cite{KPI25, GLY26} into \autoref{cor: NC0-Avoid} immediately implies a nondeterministic adversary that achieves a better stretch: For every $k\ge 3$, \cite{GLY26} presented a deterministic algorithm for $\NC^0_k$-$\Avoid$ with $\Omega_k(n^{(k-1)/2}\log n)$ stretch, hence we obtain nondeterministic adversaries breaking every demi-bits generator of stretch $\Omega_k(n^{(k-1)/2}\log n)$ as well.

\paragraph{Finding hard partial truth tables is hard.} \autoref{thm: main informal} also allows us to explain our lack of progress in the $\ckt$-$\PartialHard$ problem even for very simple circuit classes such as $\ckt = \DNF$. In fact, we show that this problem is at least as hard as \emph{PAC-learning} $\ckt$ by nondeterministic algorithms.\footnote{We actually define nondeterministic PAC-learning algorithms in terms of ``random right-hand-side refutation'' algorithms, as these two problems are equivalent in the deterministic setting~\cite{Vadhan17}, and the latter can be defined straightforwardly in the nondeterministic setting. See \autoref{def: nondet PAC learning} for details.} It is a fundamental and long-standing open problem in learning theory to design an efficient PAC-learning algorithm for DNFs, and the best algorithm to date requires $2^{\tilde{O}(n^{1/3})}$ time~\cite{KlivansS04}. It turns out that this long-standing question explains our inability to solve $\DNF$-$\PartialHard$ as well!

\begin{theorem}[Informal]
    For nondeterministic algorithms, generating hard partial truth tables for a circuit class $\ckt$ is at least as hard as learning $\ckt$. More precisely:
    \begin{itemize}
        \item A worst-case nondeterministic algorithm for 
        $\ckt$-$\PartialHard$ implies a distribution-free 
        nondeterministic PAC-learner for $\ckt$.
        \item An errorless average-case nondeterministic heuristic for 
        $\ckt$-$\PartialHard$ implies a distribution-specific 
        nondeterministic PAC-learner for $\ckt$.
    \end{itemize}
\end{theorem}

Moreover, Daniely and Shalev-Shwartz~\cite{DanielyS16} proved the hardness of learning DNFs under a natural assumption about the complexity of random $k$-SAT. Applying their reduction, we establish the hardness of $\DNF$-$\PartialHard$ under a variant of the random $k$-SAT assumption against nondeterministic algorithms. (Since their reduction is randomized, we actually need the hardness of random $k$-SAT against $\AM$ algorithms.)

\begin{corollary}[$\DNF$-$\PartialHard$ Is Hard; Informal]
    If random $k$-SAT with suitable parameters is hard against $\AM$ algorithms, then there is no errorless nondeterministic heuristic that efficiently solves $\DNF$-$\PartialHard$, even on average over a certain samplable distribution.
\end{corollary}

\paragraph{Hardness of the Remote Point problem.} \autoref{thm: main informal} also helps explain the hardness of $\XOR$-$\RemotePoint$, the original Remote Point problem introduced in~\cite{AlonPY09}.

To prove the hardness of the Remote Point problem, we
introduce a slight generalization of demi-bits generators
called remote-point demi-bits generators: these are generators $G:\{0, 1\}^n \to \{0, 1\}^N$ such that no proof system can prove ``$y$ is \emph{far from} $\Range(G)$'' for a noticeable fraction of $y$. It is easy to see that the nondeterministic hardness of LPN implies a remote-point demi-bits generator computable using only $\XOR$ gates. 
Combining this with a disperser family satisfying a
certain \emph{Lipschitz property} (which is satisfied by our main construction in~\autoref{thm: intro proj family}), we have:%

\begin{theorem}[Informal]\label{thm: intro remote-point demi-bits generator to hardness of remote point}
    Composing a remote-point demi-bits generator with a Lipschitz 
    disperser family yields hard instances for the Remote Point 
    Problem: no efficient errorless nondeterministic heuristic can find a point far from the range of the resulting circuit, 
    even on average over the choice of disperser.
\end{theorem}

Combining \autoref{thm: intro remote-point demi-bits generator to hardness of remote point} with an assumption on the demi-hardness of LPN, we immediately obtain:

\begin{corollary}[Hardness of $\RemotePoint$ for Linear Circuits, Informal]
    Under the~\hyperref[assumption: demi-hardness of LPN]{Demi-hardness of LPN}, 
    no efficient nondeterministic heuristic can find a point far 
    from the range of a $\XOR$ circuit, even on average. The same 
    holds for $k$-$\XOR$ circuits under 
    the~\hyperref[assumption: sparse-demi-hardness of LPN]{Demi-hardness of Sparse LPN}.
\end{corollary}

Finally, we consider the hardness of finding \emph{average-case hard partial truth tables} against a circuit class $\ckt$. This is the natural ``hybrid'' of $\PartialHard$ and $\RemotePoint$~\cite{ChenHLR23}: Given the domain of a partial function $x_1, \dots, x_L \in \{0, 1\}^n$ and parameters $s\ge 1$, $\delta > 0$, our goal is to output a sequence of bits $b_1, \dots, b_L \in \{0, 1\}$ such that the partial function $\{x_i \mapsto b_i\}_{i \in [L]}$ is \emph{$\delta$-average-case hard} against $\ckt$ circuits of size $s$. Namely, there is no $\ckt$ circuit $C$ of size $s$ such that
\[\Pr_{i\gets [L]}[C(x_i) = b_i] \ge 1-\delta.\]
We denote this problem as $\ckt$-$\PartialAvgHard$.

It turns out that average-case hard partial truth tables are hard to construct if \emph{agnostic PAC-learning}~\cite{KSS92} is hard for nondeterministic algorithms:

\begin{theorem}[Informal]\label{thm: intro hardness of agnostic PAC-learning implies hardness of avgHPTT}
    For nondeterministic algorithms, generating hard partial truth tables for a circuit class $\ckt$ 
    is at least as hard as agnostic PAC-learning $\ckt$. More precisely:
    \begin{itemize}
        \item A worst-case nondeterministic algorithm for 
        $\ckt$-$\PartialAvgHard$ implies a nondeterministic distribution-free 
        agnostic PAC-learner for $\ckt$.
        \item An errorless average-case nondeterministic heuristic for 
        $\ckt$-$\PartialAvgHard$ implies a nondeterministic distribution-specific 
        agnostic PAC-learner for $\ckt$.
    \end{itemize}
\end{theorem}

\autoref{thm: intro hardness of agnostic PAC-learning implies hardness of avgHPTT} allows us to base the hardness of $\PartialAvgHard$ for \emph{extremely simple circuit classes}, such as $k$-$\XOR$, on variants of nondeterministic hardness of LPN:

\begin{corollary}[Hardness of $k$-$\XOR$-$\PartialAvgHard$ (Informal)]
    Under the~\hyperref[assumption: demi-hardness of random-Sparse-LPN]{Sparse-Secret LPN Assumption for Random Matrices}, no efficient nondeterministic 
    heuristic can find a partial truth table that is far from 
    every $k$-$\XOR$ function, even on average.
\end{corollary}

\subsubsection{Abstraction: Zero-Error Disperser Families}

Intuitively, the results in~\cite{RenWZ26, Ilango25} and \autoref{thm: main informal} suggest that proof complexity generators can be obtained by composing a certain function with a demi-bits generator. We make this intuition formal by introducing the notion of \emph{zero-error disperser families}.

Let $\calX\sim \{0, 1\}^n$ be a random variable. A deterministic function $f: \{0, 1\}^n \to \{0, 1\}^m$ is a \emph{zero-error disperser} for $\calX$ if the support of the random variable $f(\calX)$ contains every string in $\{0, 1\}^m$. We are interested in constructing zero-error dispersers for a class $\mathcal{C}$ of random variables (weak sources) such that the function $f$ is a zero-error disperser for every $\calX \in \mathcal{C}$. It is well known that such dispersers cannot exist for the class of general weak random sources even if the min-entropy\footnote{The min-entropy of a random variable $\calX$ is denoted as $H_\infty(\calX)$. $H_{\infty}(\calX) \geq k$ iff $\forall x \in \Supp(\calX), \Pr[\calX=x] \leq 2^{-k}$.} of $\calX$ is as large as $n-1$. Therefore, historically such zero-error dispersers have been studied for restricted classes of sources such as bit-fixing sources \cite{GabizonS12, GabizonS12b}, affine sources \cite{Li11, GabizonS12b,LZ24}, and independent sources \cite{GabizonS12, Cohen21}. For our applications here, however, we study a \emph{family} of zero-error dispersers (or equivalently, seeded zero-error dispersers) for \emph{general weak random sources} with relatively large (e.g., $n-o(n)$) min-entropy. 

Specifically, we say that a distribution $\calF$ over functions $f: \{0, 1\}^n \to \{0, 1\}^m$ is a family of $(k, \eps)$ zero-error dispersers if for every source $\calX \sim \{0, 1\}^n$ with min-entropy at least $k$, with probability at least $1-\eps$ over a random function $f \sim \calF$, $f$ is a zero-error disperser for $\calX$. Note that this is stronger than standard $(k, \eps)$ dispersers, which only require that the output has support size at least $(1-\eps)2^m$. We also remark that such a family of zero-error dispersers follows from a \emph{seeded randomness extractor} with sufficiently small error (e.g., when the error is smaller than $2^{-(m+1)}$ where $m$ is the output length). However, for our applications, it is important that the disperser has the smallest possible circuit complexity. %

\paragraph{The projection family.} Let $n \ge 10m$ and denote by $\Proj_{n \to m}$ the family of random projections mapping $n$ bits to $m$ bits. That is, to sample a random function $f: \{0, 1\}^n \to \{0, 1\}^m$ in $\Proj_{n \to m}$, sample a random subset $I \subseteq [n]$ of size $|I| = m$. Suppose $I = \{i_1, i_2, \dots, i_m\}$ with $1\le i_1 < i_2 < \dots < i_m \le n$, then define
\[f(x) := x_{i_1}x_{i_2} \dots x_{i_m}.\]

\begin{theorem}[Corollary of~\autoref{thm: disperser family by a projection}]\label{thm: intro proj family}
    Let $n\ge 10m$, then $\Proj_{n \to m}$ is an $(n-1, 1/2)$ zero-error disperser family.
\end{theorem}

To showcase how zero-error disperser families turn a demi-bits generator into proof complexity generators, we prove the following strengthening of our main theorem~\autoref{thm: main informal}. Roughly speaking, this theorem states that under suitable choice of parameters, even if we fix in advance the output length $m$ of our desired proof complexity generator, we can still draw $m$ output bits uniformly at random from our demi-bits generator and obtain a secure proof complexity generator.
\begin{theorem}[Informal]\label{thm: informal main 2}
    Let $\calP$ be a ``well-behaved'' proof system and $G: \{0, 1\}^n \to \{0, 1\}^{20n}$ be a demi-bits generator against $\calP$. Then with constant probability over a uniformly random subset $I \subseteq [20n]$ of size $|I| = 2n$, $G|_I$ is a proof complexity generator against $\calP$.
\end{theorem}
\begin{proof}[Proof Idea of \autoref{thm: informal main 2}]
    This follows easily from definition and is nearly identical to~\autoref{thm: main informal}.
    
    Let $H \subseteq \{0, 1\}^{20n}$ denote the set of ``hard'' strings for $\calP$, i.e., the strings $z \in \{0, 1\}^{20n}$ such that $\calP$ does not admit short proofs that $z\not\in\Range(G)$. If $G$ is a demi-bits generator against $\calP$, then $|H| \ge 2^{20n} / 2$, hence the uniform distribution over $H$ has min-entropy $20n-1$. It follows from \autoref{thm: intro proj family} that a random function $f\gets \Proj_{20n \to 2n}$ is a zero-error disperser for $H$ with probability at least $1/2$. 

    It is easy to verify that if $f$ is a zero-error disperser for $H$, then $f\circ G$ is a proof complexity generator against $\calP$: for every $y \in \{0, 1\}^{2n}$, there exists a string $z \in H$ such that $f(z) = y$; hence, if $\calP$ can efficiently prove $y\not\in\Range(f\circ G)$, then $\calP$ can efficiently prove $z\not\in\Range(G)$ as well, a contradiction to $z \in H$.
\end{proof}

\paragraph{Proof complexity generators from barely non-trivial demi-bits generators.} In~\cite{RenWZ26}, the authors composed pairwise independent hash functions with demi-bits generators to obtain proof complexity generators. Inspired by their work, we show that $t$-wise independent hash functions are zero-error disperser families with extremely good parameters:

\begin{theorem}[Corollary of~\autoref{thm: disperser family from twise independence}]
    Any family of $m$-wise independent hash functions $\calH = \{h: \{0, 1\}^n \to \{0, 1\}^m\}$ is an $(m + \log m + O(1), 2^{-m})$ zero-error disperser family.
\end{theorem}

This family requires large circuit complexity overhead. Nevertheless, in the regime where we do not need to optimize the circuit complexity of our generators, the extremely nice parameters of this family enable us to obtain proof complexity generators from demi-bits generators that are \emph{barely non-trivial}.

Let $G: \{0, 1\}^n \to \{0, 1\}^m$ be any function and $\calP$ be a sound proof system. If $y \in \Range(G)$, then the (false) statement ``$y \not\in\Range(G)$'' is trivially unprovable in $\calP$, since $\calP$ is sound. Therefore, if $G$ is injective, then there are at least $2^n$ statements of the form ``$y \not \in \Range(G)$'' that are infeasible to prove in $\calP$. We say that $G$ is a \emph{barely non-trivial} demi-bits generator if there are at least $100n2^n$ many statements of the form ``$y\not\in\Range(G)$'' that are infeasible to prove in $\calP$. This bound only exceeds the trivial bound of $2^n$ by an $O(n)$ factor, hence, as a demi-bits generator, $G$ looks \emph{barely non-trivial}. However, composing $G$ with a hash function in an $m$-wise independent hash family, we obtain that:
\begin{theorem}[Informal]
    If there is a barely non-trivial demi-bits generator against $\calP$, then there is a proof complexity generator against $\calP$.
\end{theorem}

Finally, we note that Ilango's proof~\cite{Ilango25} can be regarded as a construction of \emph{seeded} zero-error disperser families as well; see \autoref{sec: ilango seeded}.

\subsection{Interlude: Are Lower Bound Properties Inherently Large?}\label{sec: discussion natural proof largeness}
A major objection to the natural proof barrier is its \emph{largeness} requirement~\cite{Chow11, Williams16, CWY23}. For example, Williams~\cite{Williams16} showed that circuit lower bounds for $\NEXP$ are equivalent to properties that are constructive (but not necessarily large). Suppose a complexity theorist discovers a property $\calP$ of hard Boolean functions, why does $\calP$ have to accept many \emph{random} functions?

Moreover, the largeness issue was exposed in the recent line of work on \emph{hardness magnification}~\cite{OliveiraS18, OliveiraPS21, McKayMW19, ChenMMW19, ChenJW19, ChenJW20, Hirahara23, ChenHOPRS22, LiuP21, CLY22}: One can use ``simple tricks'' (such as kernelization) to reduce the task of proving strong lower bounds (which may be subject to the natural proof barrier) to that of proving weaker lower bounds. Moreover, in many scenarios, the weak lower bound is already known (just not for the specific hard function required for magnification) and can be proved by a natural property, hence a ``simple trick'' seems to bypass the barrier completely!

We suspect that the relationship between demi-bits generators and proof complexity generators is the key towards understanding the largeness condition. If it were the case that every demi-bits generator is also a proof complexity generator, then largeness would be inherent in any circuit lower bound proofs---when we prove a circuit lower bound, we break $\TT$ as a proof complexity generator, hence we also break it as a demi-bits generator, and this is equivalent to constructing an $\NP$-natural property (with the largeness requirement). However, if demi-bits generators exist at all, then one can construct contrived demi-bits generators that are not proof complexity generators.\footnote{Let $G:\{0, 1\}^n \to \{0, 1\}^N$ be a demi-bits generator, $y \in \{0, 1\}^N$, and $G_y$ be the same generator as $G$ except that when $G$ returns $y$, $G_y$ returns some string different from $y$. It is trivial to show that $y\not\in\Range(G_y)$, i.e., to break $G_y$ as a proof complexity generator.} Still, \autoref{thm: main informal} along with the recent works connecting demi-bits generators to proof complexity generators~\cite{RenWZ26, Ilango25} suggest that largeness might be somewhat inherent in circuit lower bound proofs.

Our results imply a version of the natural proof barrier that does not refer to the largeness of \emph{functions} (proven to be hard) but refers to the largeness of \emph{generators} (that the lower bound method applies to). Suppose that the truth table generator $\TT$ is a demi-bits generator. Then with constant probability over a random subset $S\subseteq \{0, 1\}^n$, $\TT|_S$ is a proof complexity generator, which implies hardness of a large family of explicit construction problems that look very similar to circuit lower bounds (namely, hard partial truth tables over a random domain $S\subseteq \{0, 1\}^n$). To summarize, the demi-hardness of $\TT$ implies infeasibility of not only proving circuit lower bounds for \emph{random} functions, but also proving circuit lower bounds for arbitrary functions over a \emph{random} domain.

In his influential work on ``the hardest explicit constructions,'' Korten~\cite{Korten21} concluded\footnote{Modulo the conjecture that $\Avoid\not\in\FP$, which is now confirmed under cryptographic assumptions~\cite{DBLP:conf/stoc/IlangoLW23}.} that to prove a circuit lower bound for $\E$, one must use specific properties of $\TT$ that distinguish it from an arbitrary generator. Going one step further, we conclude that one must use specific properties of $\TT$ that distinguish it from \emph{most} generators of the form $\TT|_S$! In any case, we believe that our work will lead to a better understanding of the role of largeness in circuit lower bound proofs.

\section{Preliminaries}

In this paper, a \emph{circuit class} $\ckt$ is just any class of Boolean functions closed under projections. In contrast to standard circuit classes considered in the literature such as $\AC^0$ and $\P/_\poly$, we will consider \emph{heavily restricted} circuit classes such as $\NC^0_k$ or the class of $k$-bit parities.

We use $\{x_i \mapsto b_i\}_{i \in [L]}$ to denote a partial function $f$ such that $f(x_i) = b_i$ for every $i \in [L]$.

\subsection{Explicit Construction Problems}
Let $\mathscr{C}$ be a (multi-output) circuit class, we consider the following problems.
    \begin{itemize}
        \item \underline{$\mathscr{C}$-$\Avoid[n,m]$} is the class of  $\Avoid$ problems where the circuits are in $\mathscr{C}$, with input length $n$ and output length $m$;
        
        \item \underline{$\mathscr{C}$-$\RemotePoint[n,m,c(n)]$} is the class of $\RemotePoint$ problems where the underlying circuits belong to $\mathscr{C}$, with input length~$n$ and output length~$m$, and where the desired output has relative Hamming distance~$> c(n)$ from every string in the range of circuits in~$\mathscr{C}$.

        \underline{$\mathscr{C}$-$\PartialHard[n,s,L]$}: Given $L$ input strings $x_1, x_2, \dots, x_L \in \{0,1\}^n$, output $L$ bits $b_1, b_2, \dots, b_L$ such that no circuit $C \in \ckt$ of size $\le s$ satisfies $C(x_i) = b_i$ for all $i \in [L]$. %
        \item          \underline{$\mathscr{C}$-$\PartialAvgHard[n,s,L,\delta]$}: Given $L$ input strings $x_1, x_2, \dots, x_L \in \{0,1\}^n$, output $L$ bits $b_1, b_2, \dots, b_L$ such that for every circuit $C \in \ckt$ of size $\le s$,
        \[
            \frac{1}{L} \left| \{ i \in [L] : C(x_i) \neq b_i \} \right| \ge \delta.
        \]
    \end{itemize}

\subsection{Demi-Bits Generators}\label{sec: def of demi-bits}
\begin{definition}[Demi-Bits Generators]
    Let $n,m$ be length parameters such that $n<m$. A function $G:\{0,1\}^n\to \{0,1\}^m$ is an \emph{$(s,\eps)$-secure demi-bits generator} if there is no $\NP/_\poly$ adversary $\Adv$ of size $s$ such that 
    \begin{align*}
        \Pr_{y\gets\{0,1\}^m}[\Adv(y)\text{ accepts}]\ge\eps \quad\text{and}\quad\Pr_{x\gets\{0,1\}^n}[\Adv(G(x))\text{ rejects}]=1.
    \end{align*}
In this paper we assume the existence of demi-bits generators for any $\NP/_\poly$ adversary with polynomial size, so we will just write $\eps$-secure demi-bits generators. Equivalently, no efficient proof system can prove the statement ``$y\not\in\Range(G)$'' for $\geq \eps 2^m$ many strings $y$ (every such proof system must fail on at least $(1-\eps)2^m$ strings), since if there is such a proof then an $\NP/_\poly$ adversary can first guess such a proof and verify it. Notice that the notion of demi-bits generators is weaker when $\eps$ is larger.
\end{definition}

We will also study the following notion of \emph{weak} demi-bits generators. To break a $K$-weak demi-bits generator $G$, the adversary needs to prove ``$y\not\in\Range(G)$'' for \emph{all but $K$} many strings $y$.

\begin{definition}[Weak Demi-Bits Generators]
    Let $n, m$ be length parameters such that $n < m$. A function $G: \{0, 1\}^n \to \{0, 1\}^m$ is an \emph{$(s, K)$-weak demi-bits generator} if there is no $\NP/_\poly$ adversary $\Adv$ of size $s$ such that
    \[\mleft|\{y \in \{0, 1\}^m: \Adv(y)\text{ rejects}\}\mright| < K\quad\text{and}\quad\Pr_{x\gets \{0, 1\}^n}[\Adv(G(x))\text{ rejects}] = 1.\]
\end{definition}

Again, we assume the existence of weak demi-bits generators for any $\NP/_\poly$ adversary with polynomial size, so we will just write $K$-weak demi-bits generator. Note that a $K$-weak demi-bits generator is equivalent to a $(1-K/2^m)$-secure demi-bits generator. Also, any injective function $G:\{0, 1\}^n \to \{0, 1\}^m$ is trivially a $2^n$-weak demi-bits generator. When we refer to weak demi-bits generators, we usually work in the regime where $K$ is just slightly larger than $2^n$ (and we deduce hardness of $\Avoid$ even in this regime).

We will also consider demi-bits generators against uniform (i.e., $\NP$) adversaries (or uniform propositional proof systems, see \autoref{def: Cook-Reckhow proof systems}). When the demi-bits generator $G$ is non-uniform but the adversary is uniform, we assume that the adversary gets the circuit description of $G$ as input as well.

\subsection{Proof Complexity}

Recall that $\TAUT$ is the $\coNP$-complete language consisting of all tautological DNFs. In one sentence, a \emph{proof system} is a nondeterministic algorithm for $\TAUT$.

\begin{definition}[{\cite{CookR79}}]\label{def: Cook-Reckhow proof systems}
    A \emph{Cook--Reckhow} proof system for a language $L$ (usually $L = \TAUT$) is a deterministic algorithm $V(\varphi, \pi)$ that takes a purported tautology $\varphi$ and a purported proof $\pi$ as inputs and satisfies the following requirements:
    \begin{itemize}
        \item {\bf Completeness:} For every $\varphi \in \TAUT$, there exists a proof $\pi$ such that $V(\varphi, \pi)$ accepts.
        \item {\bf Soundness:} For every $\varphi \not\in\TAUT$ and every purported proof $\pi$, $V(\varphi, \pi)$ rejects.
        \item {\bf Efficiency:} $V(\varphi, \pi)$ runs in time $\poly(|\varphi| + |\pi|)$.
    \end{itemize}
\end{definition}

Note that it is possible that $|\pi|$ is super-polynomially longer than $\varphi$. Indeed, there exists a proof system such that every $\varphi \in \TAUT$ has a polynomially long proof $\pi$ if and only if $\NP = \coNP$ (such proof systems are called ``p-bounded''), and the goal of propositional proof complexity is to exhibit tautologies requiring super-polynomial (or even exponential) proof length in natural proof systems. We refer to~\cite{Krajicek_proof_complexity, Segerlind07} for a comprehensive overview of proof complexity.

\paragraph{Proof complexity generators.} Let $G: \{0, 1\}^n \to \{0, 1\}^N$ be a polynomial-size circuit and $y \in \{0, 1\}^N$. We fix a propositional encoding of the statement ``$y \not \in \Range(G)$'' into DNFs, such as the ones in~\cite{ABRW04}. Such an encoding is possible because this statement can be written as ``$\forall x \in \{0, 1\}^n$, $G(x) \ne y$'', hence the propositional variables of ``$y \not \in \Range(G)$'' consists of $x$ (and possibly extension variables representing the computational history of $G(x)$). As discussed in \autoref{sec: proof complexity generators from demi-bits generators}, most natural encodings would fit in our purpose.

Let $\calP$ be a proof system. We say $G$ is a \emph{proof complexity generator} against $\calP$ if for every $y \in \{0, 1\}^N$, the DNF encoding of ``$y \not \in \Range(G)$'' does not admit polynomial-size proofs in $\calP$. We emphasize that if $y$ is actually in $\Range(G)$, then by soundness of $\calP$, ``$y \not \in \Range(G)$'' does not admit proofs in $\calP$ (of any length). Hence this definition is equivalent to saying that for every string $y \in \{0, 1\}^N \setminus \Range(G)$, the \emph{tautology} ``$y \not \in \Range(G)$'' is hard to prove in $\calP$.

Again, when the proof system is uniform but the proof complexity generator $G$ is non-uniform, we assume that the proof system gets the circuit description of $G$ as well. This is reflected in the encoding of ``$y \not \in \Range(G)$'' (which contains a description of $G$).

\subsection{Nondeterministic Algorithms}

\paragraph{\texorpdfstring{$\FNP$}{FNP} v.s. \texorpdfstring{$\SearchNP$}{SearchNP}.}

In this paper, we will distinguish between the two notions $\FNP$ and $\SearchNP$.

\begin{definition}[$\SearchNP$~\cite{DBLP:conf/stoc/ChenL24}]
    Let $P$ be a search problem and $R$ be the binary relation defining $P$. We say $P$ can be solved by a nondeterministic polynomial-time algorithm if there is a nondeterministic Turing machine $M$ such that for every input $x$,
    \begin{itemize}
        \item If $x$ has a solution, then $M(x)$ has an accepting computation path, and every accepting path will output a valid solution $y$, i.e., $R(x,y)$ is true.
        \item If $x$ has no solution, then $M(x)$ has no accepting computation path.
    \end{itemize}
    The class of search problems solvable by nondeterministic polynomial-time algorithm is defined as $\SearchNP$.
\end{definition}

\begin{definition}[$\FNP$~\cite{DBLP:conf/stoc/ChenL24}]
    The class of search problems defined by a polynomial-time relation, i.e.,~$R\in \mathsf{P}$ is defined as $\FNP$.
\end{definition}

While it is clear that $\FNP\subseteq \SearchNP$, the following example suggests that this inclusion is strict.

\begin{proposition}[\cite{DBLP:conf/stoc/ChenL24}]
    If $\mathsf{P}\neq \mathsf{NP}$, then there is a total search problem in $\SearchNP\setminus \FNP$.
\end{proposition}

We define \emph{errorless $\SearchNP$ heuristics} for search problems (such as $\Avoid$ or $\PartialHard$):
\begin{definition}
    Let $\Pi$ be a search problem. We say a nondeterministic algorithm $\cal A$ is an \emph{errorless heuristic} for $\Pi$ if for every input instance $x$, every computational path of $\calA(x)$ either outputs $\bot$ or outputs a correct $\Pi$-solution of $x$.

    If $\calA$ is an errorless heuristic for $\Pi$ and $x$ is an instance, then we say that $\calA$ solves $\Pi$ on input $x$ if there exists at least one computational path on which $\calA$ does not output $\bot$ (which means that $\calA$ outputs a correct $\Pi$-solution of $x$). Otherwise, we say that $\calA(x) = \bot$.

    Let $\Pi$ be a search problem, $\calD$ be a distribution over inputs of $\Pi$, and $\eps > 0$. We say that an errorless heuristic $\calA$ \emph{solves $\Pi$ on $\calD$ with success probability $1-\eps$} if
    \[\Pr_{x\gets \calD}[\calA(x) = \bot] \le \eps.\]
\end{definition}

\subsection{Arthur--Merlin Protocols}

An \emph{Arthur--Merlin} protocol~\cite{Babai85} for a language $L$ is a constant-round public-coin interactive protocol between a computationally unbounded $\Prover$ (Merlin) and a randomized polynomial-time $\Verifier$ (Arthur) that satisfies the following properties for every input $x$:
\begin{itemize}
    \item {\it (Completeness)} If $x\in L$, then there is a $\Prover$ that makes the $\Verifier$ accept w.p.~$\ge 2/3$.
    \item {\it (Soundness)} If $x\not\in L$, then no $\Prover$ can make the $\Verifier$ accept w.p.~$>1/3$.
\end{itemize}

Let $\AM$ denote the set of languages with an Arthur--Merlin protocol.

In this paper, we do not require the full power of $\AM$; 
we only require an average-case variant where completeness 
holds over random instances rather than all yes-instances:
\begin{itemize}
    \item {\it (Soundness)} If $x \not\in L$, then no 
    $\Prover$ can make the $\Verifier$ accept w.p.~$> 1/3$.
    \item {\it (Completeness)} If $x \leftarrow L$ is a 
    random yes-instance, then there exists a $\Prover$ 
    that makes the $\Verifier$ accept w.p.~$\ge 2/3$ over the randomness of the instance and 
    the $\Verifier$'s coins.
\end{itemize}
For example, when refuting random $k$-CNFs: soundness 
requires that no $\Prover$ can make the $\Verifier$ 
accept a satisfiable CNF w.p.~$> 1/3$, 
while completeness only requires that the $\Verifier$ 
accepts w.p.~$\ge 2/3$ over a random CNF 
and its own coins. (See, e.g.,~\autoref{def: AM algorithm for random k-SAT}.)

\subsection{Nondeterministic CSP Refutation}

\begin{definition}[Nondeterministic Strong Refutation of 
Semi-random $k$-$\XOR$]\label{def:strong-xor}
A \emph{semi-random $k$-$\XOR$ instance} $\phi$ on $n$ Boolean 
variables consists of an adversarial $k$-uniform hypergraph 
$\mathcal{H}$ on $[n]$ with $m$ hyperedges and i.i.d.\ 
uniform random right-hand sides 
$b_C \sim \{\pm 1\}$ for $C \in \mathcal{H}$. Associated to 
$\phi$ is the polynomial
\[
    \phi(x) \;=\; \frac{1}{m} \sum_{C \in \mathcal{H}} 
    b_C \prod_{i \in C} x_i,
\]
so that 
$\val(\phi) = \frac{1}{2} + \frac{1}{2} 
\max_{x \in \{\pm 1\}^n} \phi(x)$.
 
A \emph{nondeterministic $\varepsilon$-tight refutation 
algorithm} is a polynomial-time $\Verifier$ $V$ that takes as 
input an instance $\phi$, a witness 
$w \in \{0,1\}^{\poly(n)}$, and a claimed upper bound 
$\alpha \in [0,1]$, such that:
\begin{itemize}
    \item \emph{(Soundness)} For every instance $\phi$, 
        every witness $w$, and every $\alpha$,
        \[
            V(\phi, w, \alpha) = 1 
            \;\Longrightarrow\; 
            \max_{x \in \{\pm 1\}^n} \phi(x) 
            \;\leq\; \alpha.
        \]
    \item \emph{(Completeness)} For every adversarial 
        $\mathcal{H}$ with $m$ sufficiently large,
        \[
            \Pr_{b \sim \{\pm 1\}^m}\!\bigl[
            \exists\, w \in \{0,1\}^{\poly(n)}:\;
            V(\phi, w, \varepsilon) = 1
            \bigr] \;\geq\; 0.99.
        \]
\end{itemize}
In particular, whenever $V$ accepts, the certified bound 
$\val(\phi) \leq \frac{1}{2} + \frac{\alpha}{2}$ holds. 
For $\varepsilon \to 0$, this approaches the random 
assignment threshold of $\frac{1}{2}$.
\end{definition}
\subsection{Useful Results}
\paragraph{Pajor's lemma.} Let $H\subseteq \{0, 1\}^m$, and $I\subseteq [m]$ denote a subset of indices. We say that $H$ \emph{shatters} $I$ if for every string $z \in \{0, 1\}^I$, there exists a string $y\in H$ such that $y|_I = z$. 

\begin{lemma}[Pajor's Lemma~{\cite{pajor1985sous}}]\label{lemma: Pajor}
    Let $H\subseteq \{0, 1\}^m$. Then there are at least $|H|$ many subsets of indices $I$ such that $H$ shatters $I$.
\end{lemma}

As a direct corollary of \hyperref[lemma: Pajor]{Pajor's lemma}, we have:
\begin{lemma}[Sauer--Shelah Lemma~{\cite{Sauer72, Shelah72, VC68}}]\label{lemma: Sauer--Shelah}
    Let $H\subseteq \{0, 1\}^m$ such that no subset $I\subseteq [m]$ with size $k$ shatters $H$. Then
    \[|H| \le \sum_{i=0}^{k-1}\binom{m}{i}.\]
\end{lemma}

\paragraph{Kruskal--Katona theorem.} Let $\calF$ be a collection of subsets of $[n]$, we say $\calF$ is \emph{downward closed} if for every two subsets $C\subseteq D\subseteq[n]$, if $D\in \calF$, then $C\in\calF$. For a positive real number $x > 0$ and an integer $i \le x$, define
\[\binom{x}{i} := \frac{x(x-1) \dots (x-i+1)}{i!}.\]

We need the following version of the Kruskal--Katona theorem~\cite{kruskal1963number,katona1968theorem, harper1966optimal,clements1969generalization}, formulated by Lovász~\cite[13.31b]{lovasz2007combinatorial}:
\begin{theorem}[{Kruskal--Katona Theorem, Lovász's Formulation}]\label{thm: Kruskal--Katona}
    Let $\calF$ be a collection of subsets of $[n]$ that is downward closed, and for each $k\in\N$ denote $\calF_k$ to be the collection of size-$k$ subsets in $\calF$. For every $i > j$, if $|\calF_i| = \binom{x}{i}$ where $x\ge i$ is real, then $|\calF_j| \ge \binom{x}{j}$.
\end{theorem}

\paragraph{Prefix sum of binomial coefficients.} Let $H(p) = -p\log p - (1-p)\log(1-p)$ be the binary entropy function, then:
\begin{proposition}\label{prop: prefix sum of binomial coefficients}
For every integers $n, k$ such that $k\le n/2$, we have
    \[\sum_{i=0}^k\binom{n}{i} \le 2^{n\cdot H(k/n)}.\]
\end{proposition}

\section{Disperser Families}\label{sec: disperser families}
\begin{definition}[$(k, \eps)$ Zero-Error Disperser Family]\label{def: disperser family}
    Let $\calF$ be a distribution over functions $f: \{0, 1\}^n \to \{0, 1\}^m$, $k\ge m$, and $\eps > 0$ be parameters. We say that $\calF$ is a \emph{$(k, \eps)$ zero-error disperser family}, if for every random variable $\calX$ over $\{0, 1\}^n$ with $H_\infty(\calX) \ge k$,
    \[\Pr_{f\gets \calF}[\Supp(f(\calX)) = \{0, 1\}^m] \ge 1-\eps.\]
    We say $\calF$ is efficient if every function $f \in \calF$ has a polynomial size description and is computable in polynomial time. 
\end{definition}

It is easy to see that in \autoref{def: disperser family}, it suffices to consider random variables $\calX$ that are uniform distributions over $2^k$ many strings.

\subsection{Random Projections}

Let $\Proj_{n \to m}$ be the distribution of a uniformly random projection from $n$ bits to $m$ bits. That is, to sample $f: \{0, 1\}^n \to \{0, 1\}^m$ from $\Proj_{n \to m}$, sample a random subset $I \subseteq [n]$ of size $m$, let $I = \{i_1, i_2, \dots, i_m\}$ where $1 \le i_1 < i_2 < \dots < i_m \le n$, and define $f(x) = x|_I := x_{i_1}x_{i_2} \dots x_{i_m}$.

We first observe that the \hyperref[lemma: Sauer--Shelah]{Sauer--Shelah Lemma} implies: 
\begin{theorem}\label{thm: disperser from Sauer--Shelah}
    For any $k \leq n$ and some $m := \Omega(k/\log(2n/k))$, $\Proj_{n\to m}$ is a $(k, 1-\delta)$ zero-error disperser family with some (exponentially-small) $\delta > 0$.
\end{theorem}
\begin{proof}
    Let $\calX$ be a random variable over $\{0, 1\}^n$ with min-entropy $k$. Without loss of generality, we may assume that $\calX$ is the uniform distribution over a subset of $2^k$ many strings, and abusing notation, we denote this size-$2^k$ subset as $\calX$ as well.

For some $m=\Omega(k/\log( 2n/k))$, we have \[\left(\frac{en}{m}\right)^m < 2^k.\]

    By \hyperref[lemma: Sauer--Shelah]{Sauer--Shelah Lemma}, since
    \[|\calX| = 2^k \ge \left(\frac{en}{m} \right)^m \ge \sum_{i=0}^{m-1}\binom{n}{i},\]
    there exists a subset $I \subseteq[n]$ with $|I| = m$ such that $I$ is shattered by $\calX$. It follows that the function $f(x) = x|_I$ in the support of $\Proj_{n\to m}$ satisfies that $\Supp(f(\calX)) = \{0, 1\}^m$.
\end{proof}

The success probability $\delta$ in \autoref{thm: disperser from Sauer--Shelah} is inherently exponentially small if $k \ll n$: Let $\calX$ be the uniform distribution over strings whose last $n-k$ bits are $0$ (but first $k$ bits can be arbitrary). Then, a random function $f$ from $\Proj_{n\to m}$ satisfies $\Supp(f(\calX)) = \{0, 1\}^m$ only when the underlying index set $I$ is a subset of $[k]$, which happens with probability $\approx (k/n)^m \ll o(1)$. However, when $k$ is very close to $n$, by combining \hyperref[lemma: Pajor]{Pajor's Lemma} and \hyperref[thm: Kruskal--Katona]{Kruskal--Katona Theorem}, we can indeed show that $\Proj_{n\to m}$ is a good disperser family with probability close to $1$ (a typical parameter setting of the following theorem is that $\Delta$ being a constant and $m = \eps n / \Delta$ for some small constant $\eps > 0$):

\begin{theorem}[{Random Projection Disperser Family}]\label{thm: disperser family by a projection}
    Let $m \le 0.4n$ and $\Delta \le o(n)$, then $\Proj_{n \to m}$ is an $(n-\Delta, 4m\Delta / n)$ zero-error disperser family.
\end{theorem}
\begin{proof}
    Let $\calX\subseteq \{0, 1\}^n$ be a subset of size $2^{n-\Delta}$. Let $\calF$ be the collection of subsets of indices $I \subseteq [n]$ that are shattered by $\calX$. Recall this means that for every $z \in \{0, 1\}^I$, there exists $y \in \calX$ such that $y|_I = z|_I$. By \hyperref[lemma: Pajor]{Pajor's Lemma}, $|\calF| \ge |\calX| \ge 2^{n-\Delta}$. Moreover, $\calF$ is downward closed: if $\calX$ shatters $I$, then $\calX$ shatters every subset of $I$.

    For every $k\in \N$, let $\calF_k$ denote the collection of size-$k$ sets in $\calF$. Suppose that $|\calF_m| = \binom{x}{m}$ for some real number $x\ge m$. Then by \hyperref[thm: Kruskal--Katona]{Kruskal--Katona Theorem}, for every integer $i \in [m, x]$, $|\calF_i| \le \binom{x}{i}$; while if $i > x$ then there is no size-$i$ subset in $\calF$. It follows that
    \[|\calF|\le \sum_{i=0}^{m-1}\binom{n}{i} + \sum_{i=m}^{\lfloor x \rfloor}\binom{x}{i} \le 2^{n\cdot H(m/n)} + 2^x \le 2^{0.9999n} + 2^x.\]
    Therefore $x \ge n-\Delta-1$. It follows that a random size-$m$ subset of $[n]$ is in $\calF$ with probability
    \[\binom{x}{m} / \binom{n}{m} \ge \mleft(\frac{x-m}{n-m}\mright)^m\ge \mleft(1-\frac{\Delta+1}{n-m}\mright)^m \ge 1-4m\Delta / n.\]

    Let $f:\{0, 1\}^n \to \{0, 1\}^m$ be a random projection sampled from $\Proj_{n\to m}$, and let $I \subseteq[n]$ be the subset of indices such that $f(x) = x|_I$. With probability at least $1-4m\Delta / n$, we have $I \in \calF$, and thus $\calX$ shatters $I$. This means that $\Supp(f(\calX)) = \{0, 1\}^m$.
\end{proof}

\subsection{\texorpdfstring{$t$}{t}-Wise Independent Hash Families}

We also show that $t$-wise independent hash functions are good disperser families in the suitable parameter regime. Although the circuit complexity overhead of these hash functions is much worse than projections, looking ahead, these disperser families will be used in \autoref{sec: improving RWZ26} to obtain extreme hardness results for (unrestricted) $\Avoid$ from mild demi-hardness assumptions and vastly improve the main results of~\cite{RenWZ26}. %

We start with pairwise independent hash families. We need the following easy corollary of Chebyshev Inequality.
\begin{lemma}\label{fact: Chebyshev}
    Let $x_1, x_2, \dots, x_n$ be pairwise independent variables where the marginal distribution of each $x_i$ is the Bernoulli distribution with success probability $p$. Then, the probability that every $x_i$ is equal to $0$ is at most $1/(pn)$.
\end{lemma}

\begin{theorem}\label{thm: disperser family from pairwise independence}
    Let $k := 2m + \log(1/\eps)$. Let $\calH$ be a family of pairwise independent hash functions where each function $h: \{0, 1\}^n \to \{0, 1\}^m$. Then $\calH$ is a $(k,\eps)$ zero-error disperser family.
\end{theorem}
\begin{proof}%
    Let $\calX\subseteq\{0, 1\}^n$ be a subset of size $2^k$, and let $h\gets\calH$ be a random hash function. For every fixed $z \in \{0, 1\}^m$, we will argue that $\Pr_{h\gets \calH}[z\not\in h(\calX)] \le \eps\cdot 2^{-m}$. Indeed, for every $x \in \calX$, let $I_x$ denote the indicator variable of $h(x) = z$, then the marginal distribution of each $I_x$ is the Bernoulli distribution with success probability $2^{-m}$, and the variables $\{I_x\}_{x\in \calX}$ are pairwise independent. The probability that $z\not\in h(\calX)$ is equal to the probability that every $I_x$ is equal to $0$, hence by \autoref{fact: Chebyshev} is at most $2^{m-k} \le \eps\cdot 2^{-m}$. Finally, the theorem then follows from a union bound over all $z \in \{0, 1\}^m$.
\end{proof}

For $t$-wise independence, we need the following tail bound, which follows from \cite[Lemma 2.3]{BellareR94} by substituting $\mu := np$ and $A := np$:
\begin{lemma}\label{lemma: concentration for t-wise independence}
    Let $t\ge 4$ be an even integer and suppose that $x_1, x_2, \dots, x_n$ are $t$-wise independent random variables where the marginal distribution of each $x_i$ is the Bernoulli distribution with success probability $p$. Then, the probability that every $x_i$ is $0$ is at most $8(tnp + t^2)^{t/2} / (np)^t$.
\end{lemma}
\begin{theorem}\label{thm: disperser family from twise independence}
    Let $t\ge 6$ be an even number and $k := m(1+2/t) + \log t + O(\log(1/\eps)/t) + 1$. Let $\calH$ be a family of $t$-wise independent hash functions where each function $h: \{0, 1\}^n \to \{0, 1\}^m$. Then $\calH$ is a $(k,\eps)$ zero-error disperser family.
\end{theorem}
\begin{proof}
    We follow the same proof as \autoref{thm: disperser family from pairwise independence}, but use \autoref{lemma: concentration for t-wise independence} instead of \autoref{fact: Chebyshev}.

    Let $\calX\subseteq \{0, 1\}^n$ be a subset of size $2^k$, and let $h\gets \calH$ be a random hash function. For every fixed $z \in \{0, 1\}^m$, we will argue that $\Pr_{h\gets \calH}[z\not\in h(\calX)] \le \eps\cdot 2^{-m}$. Indeed, for every $x \in \calX$, let $I_x$ denote the indicator variable of $h(x) = z$, then the marginal distribution of each $I_x$ is the Bernoulli distribution with success probability $2^{-m}$, and the variables $\{I_x\}_{x\in \calX}$ are $t$-wise independent. The probability that $z\not\in h(\calX)$ is equal to the probability that every $I_x$ is equal to $0$, hence by \autoref{lemma: concentration for t-wise independence} is at most
    \begin{align*}
        \frac{8(t\cdot 2^{k-m} + t^2)^{t/2}}{2^{(k-m)t}} \le&\, \frac{8\cdot (2t\cdot 2^{k-m})^{t/2}}{2^{(k-m)t}}&(\text{Since }2^{k-m}\ge t)\\
        \le&\,\frac{8\cdot (2t)^{t/2}\cdot 2^{mt/2}}{2^{kt/2}}\\
        \le&\,\frac{8\cdot (2t)^{t/2}\cdot 2^{mt/2}}{2^{mt/2}2^m(2t)^{t/2}\cdot \poly(1/\eps)}\\
        \le&\, \eps\cdot 2^{-m}.
    \end{align*}
    Finally, the theorem then follows from a union bound over all $z \in \{0, 1\}^m$.
\end{proof}

\subsection{Ilango's Seeded Disperser Family}\label{sec: ilango seeded}

We also remark that Ilango's proof~\cite{Ilango25} implicitly uses a \emph{seeded} disperser family as well:

\begin{theorem}
    Let $\Delta \ge 1$, $\eps > 0$, $d := \Delta + \log n + \log(1/\eps) + O(1)$, and $\calF$ be the following distribution of functions $f: \{0, 1\}^n \times \{0, 1\}^d \to \{0, 1\}^n$. To sample a function $f\gets \calF$, sample $t := 2^d$ random strings (``shifts'') $s_1, s_2, \dots, s_t \gets \{0, 1\}^n$ and define $f(x, i) = x \oplus s_i$.
    
    Then $\calF$ is an $(n-\Delta, \eps)$ zero-error \emph{seeded} disperser family in the following sense. For every random variable $\calX$ over $\{0, 1\}^n$ with $H_\infty(\calX) \ge n-\Delta$, let $\calZ$ be the uniform distribution over $\{0, 1\}^d$ independent from everything else, then
    \[\Pr_{f\gets \calF}[\Supp(f(\calX, \calZ)) = \{0, 1\}^n] \ge 1-\eps.\]
\end{theorem}
\begin{proof}
    This is basically the same proof as in \cite{Ilango25} and \cite{Lautemann83}. Again, without loss of generality, we may assume that $\calX$ is the uniform distribution over some $X\subseteq \{0, 1\}^n$ of size at least $2^{n-\Delta}$.

    Let $s_1, s_2, \dots, s_t \gets \{0, 1\}^n$ be sampled independently and uniformly at random. For each $z \in \{0, 1\}^n$, the probability over $s \gets \{0, 1\}^n$ that $z \not \in X\oplus s$ is at most $1-2^{-\Delta}$, hence the probability that $z\not\in X\oplus s_i$ for every $i\in [t]$ is at most $(1-2^{-\Delta})^t < \eps \cdot 2^{-n}$. It follows from a union bound over all $z \in \{0, 1\}^n$ that with probability at least $1-\eps$, the support of $f(\calX, \calZ)$ contains every length-$n$ string.
\end{proof}

\section{Hardness of Range Avoidance}\label{sec: main avoid}

Our main result is that composing zero-error disperser families with demi-bits gives us hard instances for the Range Avoidance problem. We have the following theorem.

\begin{theorem}\label{thm:avoidgeneral}
    Let $G: \{0, 1\}^n \to \{0, 1\}^N$ be a $2^k$-weak demi-bits generator, and $\calF = \{f: \{0, 1\}^N \to \{0, 1\}^m\}$ be an efficient $(k, \eps)$ zero-error disperser family. For any function $f \in \calF$ define the function $C_f: \{0, 1\}^n \to \{0, 1\}^m$ such that $\forall x \in  \{0, 1\}^n, C_f(x)=f(G(x))$. Let $\mathscr{C}$ be a circuit class where $C_f\in \mathscr{C}$ for all $f\in \calF$. Then, for any errorless $\SearchNP$ heuristic $\calA$ for $\mathscr{C}$-$\Avoid[n, m]$, $\Pr_{f \leftarrow \calF}\bigl[\calA(C_f) \neq \bot \bigr] \leq \eps$.
\end{theorem}

\begin{proof}
Suppose for contradiction that some errorless $\SearchNP$ heuristic $\calA$ satisfies
$\Pr_f[\calA(C_f) \neq \bot] > \eps$. We construct a nondeterministic polynomial-time adversary
$\calB$ breaking the demi-bits generator $G$. On input $y \in \{0,1\}^N$, the adversary $\calB$ accepts
if and only if there exists $f \in \calF$ such that some
nondeterministic branch of $\calA(C_f)$ outputs $f(y)$.
\begin{description}
    \item[$\calB$ rejects every $y \in \Range(G)$.] If $y = G(s)$, then for every $f \in \calF$ we
have $C_f(s) = f(G(s)) = f(y)$, so $f(y) \in \Range(C_f)$. Since $\calA$ is an errorless
solver for $\Avoid$, no accepting branch of $\calA(C_f)$ outputs an element of
$\Range(C_f)$, so $\calB(y)$ rejects.

\item [$\calB$ accepts at least $2^N-2^k$ strings $y\in \{0,1\}^N$.] Assume to the contrary that $\calB$ accepts $\le 2^N-2^k$ strings of length $N$. Let $X$ be the uniform distribution over the strings in $\{0,1\}^N$ rejected by $\calB$. Then we have $|\Supp(X)| \geq 2^k$ and  $H_\infty(X)\ge k$.  Since $\calF$ is an efficient $(k, \eps)$ zero-error disperser family, we have
\[
  \Pr_{f \leftarrow \calF} \bigl[\Supp(f(X)) = \{0,1\}^m\bigr] \;\geq\; 1 - \eps.
\]
Since $\Pr_f[\calA(C_f) \neq \bot] > \eps$, there must exist some $f \in \calF$ such that $\calA(C_f) \neq \bot$ and $\Supp(f(X)) = \{0,1\}^m$. This means that there exists some $x \in \Supp(X)$ such that some nondeterministic accepting branch of  $\calA(C_f)$ outputs exactly $f(x)$, hence $\calB$ accepts $x$, a contradiction.
\end{description}

Hence $\calB$ is a nondeterministic polynomial-time adversary that rejects all of $\Range(G)$
and accepts all but $2^k$ many strings in $\{0,1\}^N$, contradicting the demi-bits
security of $G$.
\end{proof}

\begin{remark}
    Recall that $G: \{0, 1\}^n \to \{0, 1\}^N$ is a $2^k$-weak demi-bits generator if any efficient proof system must fail to prove the statement ``$y \notin \Range(G)$'' for at least $2^k$ strings $y \in \{0, 1\}^N$. Applying any efficient function $f: \{0, 1\}^N \to \{0, 1\}^m$, one can ``convert'' a statement ``$y\not\in\Range(G)$'' that is hard to prove into another statement ``$f(y) \not\in\Range(C_f)$'' that is also hard to prove, where $C_f = f\circ G$. If the function $f$ is a zero-error disperser for the set of $2^k$ strings $y$ that are hard for $G$, then the image of these hard strings under $f$ covers the entire $\{0, 1\}^m$, and we get hard $\Avoid$ instances.%
\end{remark}

Similarly, composing zero-error disperser families with demi-bits generators gives us (non-uniform) proof complexity generators secure against every (uniform) proof system. This result can be proved by enumerating all such proof systems and using the Borel--Cantelli lemma; see~\cite[Theorem 3.6]{Ilango25}. Here we present an alternative proof due to \Krajicek~\cite{Krajicek26} using the existence of optimal proof systems with one bit of advice.

\begin{theorem}\label{thm: one proof complexity generator against all proof systems}
    Let $G: \{0, 1\}^n \to \{0, 1\}^N$ be a $2^k$-weak demi-bits generator, and $\calF = \{f: \{0, 1\}^N \to \{0, 1\}^m\}$ be an efficient $(k, \eps)$ zero-error disperser family. Then, with probability $\ge 1-\eps$ over the choice of $f\gets \calF$, the generator $C: \{0, 1\}^n \to \{0, 1\}^m$ defined as
    \[C(x) := f(G(x))\]
    is a proof complexity generator against all (uniform) proof systems.
\end{theorem}
\begin{proof}[Proof Sketch]
    Cook and \Krajicek proved that there exists an optimal proof system $\calP$ with one bit of advice~\cite[Theorem 6.6]{CookK07}. Since $G$ is a secure demi-bits generator against $\calP$, the same argument as in~\autoref{thm:avoidgeneral} shows that with probability at least $1-\eps$ over $f\gets \calF$, $f\circ G$ is a proof complexity generator against $\calP$. Since $\calP$ simulates every (uniform) proof system, $f\circ G$ is a proof complexity generator against every (uniform) proof system.
\end{proof}

In the remainder of this section, we instantiate \autoref{thm:avoidgeneral} and \autoref{thm: one proof complexity generator against all proof systems} with various (candidate) demi-bits and disperser families to obtain hardness of $\Avoid$ in various settings.

\subsection{Hardness of \texorpdfstring{$\Avoid$}{Avoid} for Restricted Circuits}
Thanks to the~\hyperref[thm: disperser family by a projection]{Random Projection Disperser Family}, we can convert demi-bits generators in any complexity class into hard $\Avoid$ instances within the same complexity class.

Recall that for a circuit $G: \{0, 1\}^n \to \{0, 1\}^N$ and a subset of output indices $I \subseteq [N]$, $I = \{i_1, i_2, \dots, i_m\}$ where $1 \le i_1 < i_2 < \dots < i_m \le N$, we define $G|_I$ to be the circuit $G|_I: \{0, 1\}^n \to \{0, 1\}^m$ where for every input $x \in \{0, 1\}^n$,
\[G|_I(x) = G(x)|_I = G(x)_{i_1} G(x)_{i_2} \dots G(x)_{i_m}.\]

\begin{theorem}\label{thm:avg-avoid}
Let $G: \{0, 1\}^n \to \{0, 1\}^N$ be a $(1-2^{-\Delta})$-secure demi-bits generator with $N \geq 3n$ and $\Delta \le o(N)$. Then:
\begin{itemize}
    \item For some universal constant $\alpha < 1$, let $m = \alpha N$ and assuming $m > n$, then any $\SearchNP$ algorithm $\calA$ fails to solve $\Avoid$ on the instance $G|_I$ for some subset $I\subseteq[N]$, $|I| = m$.

    \item For any $n < m \leq 0.4 N$ and any errorless $\SearchNP$ heuristic $\calA$ for $\Avoid$, $\Pr_I\bigl[\calA(G|_I) \neq \bot \bigr] \leq \frac{4m \Delta}{N}$, where $I\subseteq [N]$ is a uniformly random subset of size $|I| = m$.
\end{itemize}

\end{theorem}

\begin{proof}
    We combine \autoref{thm:avoidgeneral} with the Random Projection Disperser Family. For the first bullet, let $k := N-\Delta$, $m := \Omega(k/\log (N/k)) = \alpha N$, then by \autoref{thm: disperser from Sauer--Shelah}, $\Proj_{N \to m}$ is a $(k, 1-\delta)$ zero-error disperser family for some $\delta > 0$. Similarly, the second bullet follows from \autoref{thm: disperser family by a projection}.
\end{proof}

\begin{corollary}
    Let $k\in\N$, and $\ell(n) > n$ be a stretch function. The following holds for some universal constant $c\ge 1$: If there exists a demi-bits generator $G: \{0, 1\}^n \to \{0, 1\}^{c\ell(n)}$ computable in $\NC^0_k$, then $\NC^0_k$-$\Avoid[n, \ell(n)] \not\in\SearchNP$.
\end{corollary}

Combining the current best algorithm for $\NC^0$-$\Avoid$~\cite{KPI25, LZ25, GLY26}, we have:
\begin{corollary}
    For every $k\ge 3$ and $\eps > 0$:
    \begin{itemize}
        \item Every demi-bits generator computable in $\NC^0_k$ with stretch $\Omega_k(n^{(k-1)/2}\log n)$ can be broken in nondeterministic $n^{O(k)}$ time.
        \item Every demi-bits generator computable in $\NC^0_k$ with stretch $n^{1+\eps}$ can be broken in nondeterministic $2^{n^{1-\frac{2\eps}{k-3} + o(1)}}$ time.
    \end{itemize}
\end{corollary}
\begin{proof}
    The first bullet follows from a deterministic algorithm for $\NC^0_k$-$\Avoid[n, \Omega_k(n^{(k-1)/2}\log n)]$ in $n^{O(k)}$ time, while the second bullet follows from a deterministic algorithm for $\NC^0_k$-$\Avoid[n, n^{1+\eps}]$ in $2^{n^{1-\frac{2\eps}{k-3}+o(1)}}$ time~\cite{LZ25, GLY26}.
\end{proof}

\subsection{Hardness of \texorpdfstring{$\Avoid$}{Avoid} from Barely Non-trivial Demi-Bits}\label{sec: improving RWZ26}
In this section, we instantiate \autoref{thm:avoidgeneral} and \autoref{thm: one proof complexity generator against all proof systems} with the $t$-wise independent hash families as dispersers, and derive various hardness results for $\Avoid$.

The main result of~\cite{RenWZ26} is that for every demi-bits generator $G: \{0, 1\}^n \to \{0, 1\}^N$, pairwise independent hash family $\calH = \{h: \{0, 1\}^N \to \{0, 1\}^m\}$, and every errorless $\SearchNP$ heuristic $\calA$, there exists a hash function $h \in \calH$ such that $\calA$ fails to solve $\Avoid$ on the instance $h\circ G$. Our new framework allows us to improve this theorem in two aspects: \emph{with high probability} over $h\gets \calH$, $\calA$ fails to solve $\Avoid$ on $h\circ G$; and this only requires $G$ to be a \emph{weak} demi-bits generator.

\begin{theorem}[Strengthening of {\cite[Theorem 1.2]{RenWZ26}}]
    Let $G: \{0, 1\}^n \to \{0, 1\}^N$ be a demi-bits generator, $\calH = \{h: \{0, 1\}^N \to \{0, 1\}^m\}$ be a family of pairwise independent hash functions, and $\calA$ be an errorless $\SearchNP$ heuristic for $\Avoid$. If $N \geq 4m$ and $m > n$, then with probability $\ge 1-2^{-m}$ over $h\gets \calH$, $\calA$ fails to solve $\Avoid$ on the input instance $h\circ G$.

    Moreover, this is still true even if $G$ is only a $2^{3m}$-weak demi-bits generator.
\end{theorem}

\begin{proof}
    We combine \autoref{thm:avoidgeneral} with the disperser family in \autoref{thm: disperser family from pairwise independence}. By setting $\eps=2^{-m}$ in \autoref{thm: disperser family from pairwise independence}, the theorem follows.
\end{proof}

Similarly, we can also get the following corollaries.

\begin{corollary}
    If $\calH$ is an $m$-wise independent hash family, then the above theorem still holds even if $G$ is only a $cm2^m$-weak demi-bits generator for some universal constant $c\ge 1$.
\end{corollary}

\begin{proof}
    We combine \autoref{thm:avoidgeneral} with the disperser family in \autoref{thm: disperser family from twise independence} and set $\eps=2^{-m}$ there.
\end{proof}

Any injective function $G: \{0, 1\}^n \to \{0, 1\}^N$ is trivially a $2^n$-weak demi-bits generator, since for every $y\in\Range(G)$, it is impossible to prove that ``$y\not\in\Range(G)$'' in a sound proof system. In this regard, an $O(n2^n)$-weak demi-bits generator is \emph{barely non-trivial}. However, using an $n$-wise independent hash family, we show that such a barely non-trivially weak demi-bits generator implies proof complexity generators: 

\begin{corollary}
For some universal constant $c\ge 1$, if there exists a $cn2^n$-weak demi-bits generator $G: \{0, 1\}^n \to \{0, 1\}^N$ with $N \geq 2n$, then there exists a single (non-uniform family of) proof complexity generator secure against all (uniform) proof systems.
\end{corollary}

\begin{proof}
    We combine \autoref{thm: one proof complexity generator against all proof systems} with the disperser family in \autoref{thm: disperser family from twise independence}, where we use an $n$-wise independent hash family, and set $m=n+2$ and error $\eps=2^{-m}$ to get a disperser family $\calH = \{h: \{0, 1\}^N \to \{0, 1\}^m\}$. By \autoref{thm: disperser family from twise independence}, the entropy requirement is $k \geq n+\log n + O(1)$. Therefore we only need a $O(n2^n)$-weak demi-bits generator.
\end{proof}

\subsection{Many Proof Complexity Generators Inside One Demi-Bits Generator}\label{sec: proof complexity generators from demi-bits generators}

In this subsection, we show that for every well-behaved proof system $\calP$, demi-bits generators against $\calP$ imply proof complexity generators against $\calP$.

\begin{definition}\label{def: well-behaved proof systems}
    We say that a proof system $\calP$ is \emph{well-behaved} if there exists a polynomial $p(\cdot)$ such that for every two CNFs $C \subseteq C'$ (i.e., every clause appearing in $C$ also appears in $C'$), if $\calP$ admits a length-$\ell$ refutation of $C$, then $\calP$ admits a length-$\poly(\ell)$ refutation of $C'$.
\end{definition}

That is, a proof system is \emph{well-behaved} if adding more axioms cannot make the proof harder. This is a very natural assumption that holds for essentially every proof system of interest: to prove $C'$ is unsatisfiable, it suffices to only look at the clauses in $C$. We are not aware of any proof system studied in the literature that is not well-behaved, and we believe that any such proof system needs to be very contrived.

Let $G: \{0, 1\}^n \to \{0, 1\}^m$ and $y \in \{0, 1\}^m$, we use the notation ``$y\not\in\Range(G)$'' to represent the propositional formula asserting that there is no $z \in \{0, 1\}^n$ such that $G(z) = y$. (This formula is also denoted as $\tau_y(G)$ in the literature~\cite{Krajicek-generator-book}.) We do not specify how to encode ``$y\not\in\Range(G)$'' as a CNF here; any encoding of ``$y\not\in\Range(G)$'' satisfying the following fact would suffice, and we also remark that the encoding described in \cite[Section 3.1]{RenWZ26} satisfies the fact.
\begin{fact}\label{fact: good encoding}
    Let $G: \{0, 1\}^n \to \{0, 1\}^m$ be a generator and $I\subseteq[m]$. Then every clause appearing in ``$y|_I\not\in\Range(G|_I)$'' also appears in ``$y\not\in\Range(G)$''.
\end{fact}

\begin{theorem}\label{thm: many proof complexity generators inside one demi-bits generator}
    Let $\calP$ be a well-behaved proof system, then there is a polynomial $p(\cdot)$ such that the following holds. Let $2^{-o(N)} < \eps \le 1/2$, $n < m < N/(10\log(1/\eps))$, and $G: \{0, 1\}^n \to \{0, 1\}^N$ be a $(1-\eps)$-secure demi-bits generator against $\calP$-proofs of length $p(\ell)$. Then w.p.~at least $1-4m\log(1/\eps)/N$ over $I \subseteq [N]$, $|I| = m$, $G|_I$ is a proof complexity generator secure against $\calP$-proofs of length $\ell$.
\end{theorem}
\def\Hard{\textsc{Hard}}
\begin{proof}
    Since $\calP$ is well-behaved, there is a polynomial $p(\cdot)$ such that for every two CNFs $C\subseteq C'$, if $\calP$ admits a length-$\ell$ refutation of $C$, then $\calP$ admits a length-$p(\ell)$ refutation of $C'$.

    Let $\Hard\subseteq \{0, 1\}^N$ be the set of strings $y \in \{0, 1\}^N$ such that $\calP$ does not admit length-$p(\ell)$ proofs of ``$y\not\in\Range(G)$''. Since $G$ is a $(1-\eps)$-secure demi-bits generator, we have $|\Hard| \ge \eps\cdot 2^N$. It follows from the \hyperref[thm: disperser family by a projection]{Random Projection Disperser Family} that a random subset $I \subseteq [N]$ of size $|I| = m$ is shattered by $\Hard$ with probability $\ge 1-4m\log(1/\eps) / N$. By \autoref{fact: good encoding}, $G|_I$ is a proof complexity generator in the sense that for every $y \in \{0, 1\}^I$, $\calP$ does not admit a length-$\ell$ proof of ``$y\not\in\Range(G|_I)$''.
\end{proof}

\begin{remark}[{Comparison with~\cite{RenWZ26}}]
    Previously, \cite{RenWZ26} proved the same theorem for every proof system $\calP$ closed under \emph{simple parity reductions} (\cite[Definition 3.3]{RenWZ26}). Roughly speaking, this is because~\cite{RenWZ26} needs to compose the demi-bits generator with a hash function in some universal hash family; since such hash functions need unbounded-fanin $\XOR$ gates to compute, the results in~\cite{RenWZ26} only hold for proof systems that can reason about such $\XOR$ gates. In particular, when $\calP$ is the \emph{resolution} proof system, \cite{RenWZ26} fails to transform demi-bits against $\calP$ to proof complexity generators against $\calP$.\footnote{Resolution is \emph{provably not} closed under simple parity reductions. The Tseitin formula is exponentially hard for resolution~\cite{Urquhart87}, but it reduces to a trivial formula $x_1\land \bar{x}_1$ under simple parity reductions.} Thanks to the \hyperref[thm: disperser family by a projection]{Random Projection Disperser Family}, our results hold as long as the proof system is \emph{well-behaved}, which is a much weaker (in fact, nearly trivial) condition compared to being closed under simple parity reductions, and includes weak proof systems such as resolution.
\end{remark}

As a corollary, for every well-behaved proof system $\calP$, the ability for $\calP$ to break demi-bits generators in $\NC^0$ coincides with its ability to break proof complexity generators in $\NC^0$ up to a constant factor in the stretch.

\begin{corollary}
    Let $\calP$ be a well-behaved proof system, $k\ge 2$ be a constant, and $\ell(n) > n$ be a function, then the following are equivalent:
    \begin{enumerate}
        \item For every constant $c\ge 1$, there is a demi-bits generator $G: \{0, 1\}^n \to \{0, 1\}^{c\cdot \ell(n)}$ computable in $\NC^0_k$ that is $0.99$-secure against $\calP$.\label{item: NC0 demi-bits against P}
        \item For every constant $c\ge 1$, there is a proof complexity generator $G: \{0, 1\}^n \to \{0, 1\}^{c\cdot \ell(n)}$ computable in $\NC^0_k$ that is secure against $\calP$. \label{item: NC0 proof complexity generator against P}
    \end{enumerate}
\end{corollary}
\begin{proof}
    $(\ref{item: NC0 demi-bits against P}) \implies (\ref{item: NC0 proof complexity generator against P})$ follows from \autoref{thm: many proof complexity generators inside one demi-bits generator}: If $G: \{0, 1\}^n \to \{0, 1\}^{1000\ell(n)}$ is a demi-bits generator that is $0.99$-secure against $\calP$, then with probability $\ge 0.9$ over $I \subseteq [1000\ell(n)]$, $|I| = \ell(n)$, $G|_I$ is a proof complexity generator against $\calP$.

    $(\ref{item: NC0 proof complexity generator against P}) \implies (\ref{item: NC0 demi-bits against P})$ follows from definition: If $\calP$ can prove ``$y\not\in\Range(G)$'' for \emph{many} strings $y$, then $\calP$ can prove ``$y\not\in\Range(G)$'' for at least \emph{some} string $y$.
\end{proof}

A similar corollary concerns Goldreich's PRG~\cite{Goldreich11-PRG} where the underlying hypergraph is selected uniformly at random. Let $k\ge 3$, $P: \{0, 1\}^k \to \{0, 1\}$ be any function, and $\ell(n) > n$ be a stretch parameter. Let $\calG_{\ell(n)}$ denote the following distribution of candidate PRGs $G: \{0, 1\}^n \to \{0, 1\}^{\ell(n)}$: For each output bit $i \in [\ell(n)]$, we select $k$ random indices $j_1, j_2, \dots, j_k \gets [n]$ and set the $i$-th output bit as
\[G(x)_i := P(x_{j_1}, x_{j_2}, \dots, x_{j_k}).\]

\begin{corollary}
    Let $k\ge 3$ be a constant, $P: \{0, 1\}^k \to \{0, 1\}$ be any Boolean function, $\ell(n)$ be a stretch function, and $\calP$ be a well-behaved proof system. Then the following are equivalent:
    \begin{enumerate}
        \item For every constant $c\ge 1$ and $\eps > 0$, w.p.~at least $1-\eps$ over $G \gets \calG_{c\cdot \ell(n)}$, $G$ is a $(1-\eps)$-secure demi-bits generator against $\calP$.\label{item: Goldreich demi-bits}
        \item For every constant $c\ge 1$ and $\eps > 0$, w.p.~at least $1-\eps$ over $G\gets \calG_{c\cdot \ell(n)}$, $G$ is a proof complexity generator against $\calP$.\label{item: Goldreich proof complexity generators}
    \end{enumerate}
\end{corollary}
\begin{proof}
    $(\ref{item: Goldreich demi-bits}) \implies (\ref{item: Goldreich proof complexity generators})$ follows from \autoref{thm: many proof complexity generators inside one demi-bits generator}: Fix $c\ge 1$, $\eps > 0$, and let $c' := 1000c\eps^{-1}\log(\eps^{-1})$ and $\eps' := \eps / 3$. Suppose that $G\gets \calG_{c'\ell(n)}$ and $G' := G|_I$ for a random subset $I\subseteq [c'\ell(n)]$ with $|I| = c\ell(n)$, then the distribution of $G'$ is exactly $\calG_{c\ell(n)}$. Moreover, if $G$ is a $(1-\eps')$-secure demi-bits generator secure against $\calP$ (which happens w.p.~$\ge 1-\eps'$), then w.p.~$\ge 1-\eps/3$ we have that $G'$ is a proof complexity generator secure against $\calP$. Hence, a generator $G'\gets \calG_{c\ell(n)}$ is a proof complexity generator secure against $\calP$ w.p.~$\ge 1-\eps$.

    $(\ref{item: Goldreich proof complexity generators}) \implies (\ref{item: Goldreich demi-bits})$ is trivial, as every proof complexity generator against $\calP$ is a demi-bits generator secure against $\calP$.
\end{proof}

\section{Hard Partial Truth Tables is Hard}
Let $\ckt$ be a class of Boolean functions; the reader is encouraged to think of $\ckt$ as a circuit class with a certain resource bound, such as the class of DNFs of $n^{\log n}$ size. In this section, we first apply the \hyperref[thm: disperser family by a projection]{Random Projection Disperser Family} to establish the hardness of $\ckt$-$\PartialHard$ from the nondeterministic hardness of \emph{PAC-learning} $\ckt$. %
Then we apply a reduction in~\cite{DanielyS16} to connect the hardness of PAC-learning DNFs to the \hyperref[assumption: nondeterministic Feige]{Random $k$-SAT Hypothesis Against $\AM$}.

Recall that for input-output pairs $\{(x_i, b_i)\}_{i\in [L]}$, we use the notation $\{x_i \mapsto b_i\}_{i\in [L]}$ to denote the partial function with domain $\{x_i: i\in [L]\}$ that maps each $x_i$ to $b_i$. Without loss of generality, we assume that no two inputs $x_i, x_j$ are the same, as otherwise the input becomes trivial for both PAC-learning and $\PartialHard$. We also denote $\vec{x} = (x_1, \dots, x_L)$ and $\vec{b} = (b_1, \dots, b_L)$ for convenience of notation.

We now define nondeterministic algorithms for PAC-learning $\ckt$. The definition actually refers to a notion of \emph{random right-hand-side (RRHS) refutation} of the dual class of $\ckt$~\cite{Vadhan17}. With respect to randomized algorithms, it is known that PAC-learning $\ckt$ is equivalent to RRHS refutation of the dual class of $\ckt$~\cite{Vadhan17}. To simplify our terminology, we will also refer to RRHS refutation algorithms for the dual class of $\ckt$ as ``PAC-learning algorithms for $\ckt$'' in the nondeterministic setting.

\yan{We should be able to show that the following definition implies the one in Luming's thesis.}

\hanlin{About the name ``nondeterministic PAC-learning'': are there better names? What about simply ``nondeterministic $\ckt$-refutation''?}

\begin{definition}[Nondeterministic PAC-Learning]\label{def: nondet PAC learning}
    Let $\ckt$ be a class of Boolean functions, $L\ge 1$, $p\in (0, 1)$, and $\calD$ be a distribution over $\{0, 1\}^n$. A nondeterministic algorithm $\calA$ is said to \emph{PAC-learn} $\ckt$ over input distribution $\calD$ using $L$ samples with success probability $p$ if:
    \begin{itemize}
        \item {\it (Soundness)} for all inputs $x_1, \dots, x_L\in \{0, 1\}^n$ and labels $b_1, \dots, b_L \in \{0, 1\}$, if there exists a function $C \in \ckt$ such that $C(x_i) = b_i$ for every $i\in [L]$, then $\calA((x_1, b_1), \dots, (x_L, b_L))$ rejects. 
        \item {\it (Completeness)} If we sample i.i.d.~$x_i\gets \calD$ and $b_i \gets \{0, 1\}$, then
        \[\Pr[\calA((x_1, b_1), \dots, (x_L, b_L))\text{ accepts}] \ge p.\]
    \end{itemize}
    We also say that a nondeterministic algorithm $\calA$ \emph{distribution-freely} PAC-learns $\ckt$ with success probability $p$ if it also satisfies the following stronger (Completeness) condition:
    \begin{itemize}
        \item {\it (Completeness')} For every $x_1, \dots, x_L\in \{0, 1\}^n$, if we choose i.i.d.~labels $b_1, \dots, b_L \gets \{0, 1\}$, then
        \[\Pr[\calA((x_1, b_1), \dots, (x_L, b_L))\text{ accepts}]\ge p.\]
    \end{itemize}
\end{definition}

We also recall the definition of errorless heuristics for $\PartialHard$:

\begin{definition}[Errorless Heuristics for $\PartialHard$]
    Let $\ckt$ be a class of functions, $L\ge 1$, $p\in (0, 1)$, and $\calD$ be a distribution over $\{0, 1\}^n$. A $\SearchNP$ algorithm $\calA$ is an \emph{errorless heuristic} for $\ckt$-$\PartialHard$ over $\calD$ on domain size $L$ with success probability $p$ if:
    \begin{itemize}
        \item {\it (Soundness)} For every $x_1, \dots, x_L \in \{0, 1\}^n$, any accepting path of $\calA(\vec{x})$ prints a sequence of bits $b_1, \dots, b_L \in \{0, 1\}$ such that the partial function $\{x_i \mapsto b_i\}$ disagrees with every function in $\ckt$.
        \item {\it (Completeness)} Let $x_1, \dots, x_L\gets \calD$ be i.i.d.~samples, then with probability at least $p$, $\calA(\vec{x})$ has at least one accepting path.
    \end{itemize}
    
    We remark that the (Soundness) condition needs to hold for every $x_1, \dots, x_L \in \{0, 1\}^n$, hence the term ``\emph{errorless} heuristic''. If for \emph{every} $x_1, \dots, x_L \in \{0, 1\}^n$, $\calA(\vec{x})$ has at least one accepting path, then we say $\calA$ is a \emph{(worst-case) algorithm} for $\ckt$-$\PartialHard$.
\end{definition}

Our main theorem in this section is that hardness of nondeterministic PAC-learning implies hardness of $\PartialHard$:

\begin{theorem}\label{thm: PartialHard to PAC learning for NP}
    Let $\ckt$ be a class of functions and $L, L'\ge \omega(1)$ be parameters such that $L' \le 0.1L$. Then:
    \begin{itemize}
        \item Suppose there is a (worst-case) $\SearchNP$ algorithm $\calA$ for $\ckt$-$\PartialHard$ on domain size $L'$. Then there is a nondeterministic polynomial-time algorithm $\calB$ that distribution-freely PAC-learns $\ckt$ using $L$ samples with success probability $1-o(1)$.
        \item Let $\calD$ be a distribution over $\{0, 1\}^n$ and $\eps > 0$. Suppose there is an errorless $\SearchNP$ heuristic $\calA$ for $\ckt$-$\PartialHard$ over $\calD$ on domain size $L'$ with success probability $1-\eps$. Then there is a nondeterministic polynomial-time algorithm $\calB$ that PAC-learns $\ckt$ on input distribution $\calD$ using $L$ samples with success probability $1-2^{-0.1(L/L')} - 2\eps$.
    \end{itemize}
\end{theorem}
\def\Rej{\textsc{Rej}}
\begin{proof}
    Let $\calA$ be an errorless $\SearchNP$ heuristic (or worst-case algorithm) for $\ckt$-$\PartialHard$ on domain size $L'$. Our nondeterministic algorithm $\calB$, given inputs $(x_1, b_1), \dots, (x_L, b_L)$, guesses a subset $I \subseteq[L]$ of size $|I| = L'$, and verifies that there is a nondeterministic branch of $\calA(\{x_i\}_{i \in I})$ that outputs exactly $\{b_i\}_{i \in I}$. If this is true, then $\calB$ accepts the input, otherwise $\calB$ rejects the input.
    
    The soundness of $\calB$ is easy to see: Consider any inputs $x_1, \dots, x_L \in \{0, 1\}^n$ and labels $b_1, \dots, b_L \in \{0, 1\}$. Suppose there is a function $C \in \ckt$ such that $C(x_i) = b_i$ for every $i \in [L]$. Then for every subset $I\subseteq [L]$, the partial function $\{x_i\mapsto b_i\}_{i \in I}$ agrees with the function $C \in \ckt$, and the soundness of $\calA$ implies that $\calB$ will never accept $\{(x_i, b_i)\}_{i \in [L]}$.

    Now we argue the completeness of $\calB$. Let $x_1, \dots, x_L \in \{0, 1\}^n$ and define
    \[\Rej := \Rej(x_1, \dots, x_L) := \{\vec{b} \in \{0, 1\}^L: \calB(\{(x_i, b_i)\})\text{ rejects}\}.\]
    Suppose $I \subseteq [L]$, $|I| = L'$ is a subset shattered by $\Rej$. We argue that $\calA(\{x_i\}_{i\in I})$ does not have any accepting branch. Indeed, if $\calA(\{x_i\}_{i\in I})$ outputs any sequence of bits $\{b_i\}_{i \in I}$, letting $\vec{b}' \in \Rej$ be the element with $b'_i = b_i$ for every $i\in I$, then $\calB(\{(x_i, b_i)\})$ accepts, contradicting the definition of $\Rej$.

    \begin{itemize}
        \item Suppose that $\calA$ is a worst-case $\SearchNP$ algorithm for $\ckt$-$\PartialHard$ on domain size $L'$. Then for every $x_1, \dots, x_L$, there is no subset $I\subseteq[L]$ of size $L'$ that is shattered by $\Rej(x_1, \dots, x_L)$. By the \hyperref[lemma: Sauer--Shelah]{Sauer--Shelah Lemma}, $|\Rej(x_1, \dots, x_L)|\le \sum_{i=0}^{L'-1}\binom{L}{i} \le o(2^L)$. It follows that $\calB$ distributional-freely PAC-learns $\ckt$ using $L$ samples with success probability $1-o(1)$.
        \item Suppose that $\calA$ is an errorless $\SearchNP$ heuristic for $\ckt$-$\PartialHard$ over $\calD$ on domain size $L'$ with completeness $1-\eps$. Set $\Delta := 0.1(L/L')$ and $\delta := 2\eps$. Assume, towards a contradiction, that w.p.~at least $\delta$ over $x_1, \dots, x_L \gets \calD$ we have $|\Rej(x_1, \dots, x_L)| \ge 2^{L-\Delta}$. By \autoref{thm: disperser family by a projection}, $\Proj_{L\to L'}$ is an $(L-\Delta, 4\Delta\frac{L'}{L})$-zero error disperser family, which means that for such $(x_1, \dots, x_L)$, a random subset $I\subseteq[L]$ with size $L'$ shatters $\Rej(x_1, \dots, x_L)$ w.p.~$\ge 1-4\Delta\frac{L'}{L}$. Since the distribution of $\{x_i\}_{i\in I}$ (where the randomness is over samples $\{x_i\}_{i\in [L]}$ and subset $I$) is exactly $L'$ i.i.d.~samples from $\calD$, we have that $\calA(\{x_i\}_{i\in [L']})$ fails to produce an answer with probability at least $\delta(1-4\Delta\frac{L'}{L}) \ge \eps$, contradicting the completeness of $\calA$. Hence, $\calB$ PAC-learns $\ckt$ over input distribution $\calD$ using $L$ samples with success probability
        \[\ge1-2^{-\Delta} - \delta \ge 1-2^{-0.1(L/L')} - 2\eps.\qedhere\]
    \end{itemize}
\end{proof}

\subsection{Hardness of \texorpdfstring{$\DNF$-$\PartialHard$}{DNF-Partial-Hard}}
The nondeterministic hardness of $\PartialHard$ for a circuit class $\ckt$ follows from the nondeterministic hardness of PAC-learning $\ckt$. When $\ckt$ is the class of DNFs, hardness of PAC-learning $\ckt$ follows from the hardness of random $k$-SAT~\cite{DanielyS16}; we show here that the nondeterministic hardness of PAC-learning DNFs follows from the nondeterministic hardness of random $k$-SAT as well. Since the reduction in~\cite{DanielyS16} uses randomness, here we need a version of the random $k$-SAT hypothesis against $\AM$ refuting algorithms: 

\begin{definition}\label{def: AM algorithm for random k-SAT}
    Let $k\in\N$, $m(n) \ge n$, and $\eps \in (0, 1/2)$. An $\AM$ algorithm $\calA$ is said to \emph{refute random $k$-SAT} with $m(n)$ clauses with error $\eps$, if:
    \begin{itemize}
        \item {\it (Soundness)} for every $k$-CNF formula $C$ with $m(n)$ clauses that is satisfiable, $\Pr[\calA(C)\text{ accepts}] \le \eps$, where the probability is over the internal (Arthur's) randomness of $\calA$.
        \item {\it (Completeness)} let $C$ be a random $k$-CNF formula with $m(n)$ clauses, then $\Pr[\calA(C)\text{ accepts}] \ge 1-\eps$, where the probability is over the choice of $C$ and the internal (Arthur's) randomness of $\calA$.
    \end{itemize}
\end{definition}

Feige~\cite{Feige02} conjectured that random $3$-SAT with $O(n)$ clauses is hard to refute (against deterministic polynomial-time algorithms). O'Donnell made the same conjecture against nondeterministic polynomial-time algorithms (cf.~\cite{burgisser2016complexity, HiraharaS17}); this nondeterministic version has been useful in the context of meta-complexity~\cite{HiraharaS17, Hirahara18} and proof complexity~\cite{PichS19} as well.

The main result in~\cite{DanielyS16} requires hardness of random $k$-SAT in a somewhat different regime: for every constant $c\ge 1$, there exists a constant $k\ge 3$ such that random $k$-SAT with $n^c$ clauses is hard. We make the following assumption asserting that this is still true for $\AM$ algorithms:

\begin{assumption}[Random $k$-SAT Hypothesis Against $\AM$]\label{assumption: nondeterministic Feige}
    For every constant $c\ge 1$, there exists a constant $k\ge 3$ such that no polynomial-time $\AM$ algorithm refutes random $k$-SAT with $n^c$ clauses with error $0.01$.
\end{assumption}

We remark that hardness of random $k$-SAT against $\AM$ algorithms is implied by hardness of random $k$-SAT against \emph{non-uniform} nondeterministic algorithms ($\NP/_\poly$):
\begin{restatable}{fact}{FactAMtoNPpoly}\label{fact: AM adversaries to NPpoly adversaries}
    If there is an $\AM$ algorithm that refutes random $k$-SAT with $m(n)$ clauses with error $\eps$, then there is an $\NP/_\poly$ algorithm that refutes random $k$-SAT with $m(n)$ clauses with error $O(\eps)$.
\end{restatable}

The proof basically follows from the argument showing that $\AM \subseteq \NP/_\poly$; however, as the (Completeness) item in \autoref{def: AM algorithm for random k-SAT} is only average-case, some details are different, hence we provide a proof in \autoref{appendix: proof of AM to NPpoly}. By \autoref{fact: AM adversaries to NPpoly adversaries}, if one believes the nondeterministic hardness of random $k$-SAT even against $\NP/_\poly$ algorithms, then \autoref{assumption: nondeterministic Feige} should also be treated as a ``plausible'' assumption. In any case, refuting \autoref{assumption: nondeterministic Feige} would be a breakthrough in SAT algorithms.

We also need the following reduction:

\begin{restatable}[{\cite{DanielyS16}}]{theorem}{ThmDS}\label{thm: reduction in DS16}
    Let $n, k, m\in\N$, $\delta > 0$, and let $t := O(2^k \log (m/\delta))$. Then the random $k$-SAT problem with $n$ inputs and $m$ clauses reduces to PAC-learning CNFs of input length $n' := 2nt$, size $s := O(nt)$, and top fan-in $t$, using $L := \lfloor m/t\rfloor$ samples.

    More precisely, there exists a polynomial-time randomized reduction $\calR$ that takes a $k$-SAT instance $\varphi$ as input and outputs a partial function $\{(x_i \mapsto b_i)\}_{i\in [L]}$, such that the following holds:
    \begin{itemize}
        \item {\bf $\calR$ maps satisfiable instances to easy partial functions.}

        For every $\varphi$ that is satisfiable, with probability at least $1-\delta$ over the internal randomness of $\calR$, the partial function $\{(x_i \mapsto b_i)\}\gets \calR(\varphi)$ can be computed by CNFs of size $s$ and top fanin $t$.

        \item {\bf $\calR$ maps random instances to random partial functions.}

        The following holds for some input distribution $\calD$ over $\{0, 1\}^{n'}$. Let $\varphi$ be a random $k$-SAT instance and let $\{(x_i\mapsto b_i)\} \gets \calR(\varphi)$, then each $(x_i, b_i)$ is i.i.d.~distributed according to $\calD \times \{0, 1\}$.
    \end{itemize}
\end{restatable}

For completeness, we provide a proof of \autoref{thm: reduction in DS16} in \autoref{appendix: proof of DS16}, verifying that the parameters indeed follow from arguments in~\cite{DanielyS16}.

We also note that by De Morgan's law, learning CNFs and learning DNFs reduce to each other with no change in parameters (size, top fanin, sample complexity, input length).

\begin{theorem}\label{thm: random kSAT to learning DNF}
    Let $\ckt$ denote the class of DNFs over $n' := 2nt$ inputs with size $s := O(nt)$ and top fanin $t := O(2^k\log(m/\eps))$, and $\calD$ be the distribution in \autoref{thm: reduction in DS16}. Suppose there is an $\NP$ algorithm that PAC-learns $\ckt$ over $\calD$ using $L := \lfloor m/t\rfloor$ samples with success probability $\ge 1-\eps$. Then there is an $\AM$ algorithm for refuting random $k$-SAT instances with $m(n)$ clauses with error $\eps$.
\end{theorem}
\begin{proof}
    Let $\calA$ be an $\NP$ algorithm for PAC-learning DNF with success probability $\ge 1-\eps$. Let $\calR$ be the randomized reduction in \autoref{thm: reduction in DS16} with error probability $\delta := \eps$. Given an input $k$-SAT instance, we accept if and only if $\calA(\calR(\varphi))$ accepts. It is easy to see that this is an $\AM$ algorithm.

    If $\varphi$ is satisfiable, then with probability $\ge 1-\eps$, $\calR(\varphi)$ can be computed by CNFs of size $s$ and top fanin $t$. Therefore, our algorithm accepts with probability $\le \eps$.

    If $\varphi$ is random, then $\calR(\varphi)$ consists of $L$ i.i.d.~samples from $\calD\times \{0, 1\}$, where $\calD$ is the distribution in \autoref{thm: reduction in DS16}. It follows that our algorithm accepts with probability $\ge 1-\eps$.
\end{proof}

Combining \autoref{thm: random kSAT to learning DNF} with \autoref{thm: PartialHard to PAC learning for NP}, we have:

\begin{corollary}\label{cor: hardness of DNF-HPTT from random k-SAT}
    Let $\ckt$ denote the class of DNFs over $n' := 2nt$ inputs with size $s := O(nt)$ and top fanin $t := O(2^k\log(m/\eps))$, and $\calD$ be the distribution in \autoref{thm: reduction in DS16}. Suppose there is an errorless $\SearchNP$ heuristic for $\ckt$-$\PartialHard$ over $\calD$ on domain size $L' := \frac{m}{20t\log(1/\eps)}$ with success probability $1-\eps/4$. Then there is an $\AM$ algorithm for refuting random $k$-SAT instances with $m(n)$ clauses with error $\eps$.
\end{corollary}

Let $\DNF_{t(n)}$ denote the class of DNFs over $n$ inputs with linear size and top fanin $t(n)$. We show that $\DNF_{t(n)}$-$\PartialHard$ is hard for suitable choices of parameters $t(n)$. In fact, $\DNF$-$\PartialHard$ does not even admit errorless $\SearchNP$ heuristics over the distribution $\calD$ in \autoref{thm: reduction in DS16}:

\begin{corollary}[$\DNF$-$\PartialHard$ Is Hard]
    The \hyperref[assumption: nondeterministic Feige]{Random $k$-SAT Hypothesis Against $\AM$} implies that, for every integer $c\ge 1$, there is an integer $q\ge 1$ such that there is no errorless $\SearchNP$ heuristic for $\DNF_{q\log n}$-$\PartialHard$ over some distribution $\calD$ on domain size $n^c$ with success probability $0.999$.
\end{corollary}

\section{The Remote Point Problem}
Under suitable demi-hardness assumptions, our techniques directly imply the hardness of the Remote Point problem for $\GF(2)$-linear circuits and the hardness of Average-Case Hard Partial Truth Tables for extremely simple circuit classes such as conjunctions and narrow parities. %

We begin with a standard notion of neighborhood in Hamming space.

\begin{definition}[Hamming Ball Around a Set]
    For a set $S \subseteq \{0,1\}^m$ and a remoteness parameter $\tau \in [0,1]$, we define the \emph{$\tau$-ball} around $S$ as
    \[
        \calB(S, \tau) \coloneqq \left\{ z \in \{0,1\}^m : \delta(z, S) \le \tau \right\},
    \]
    where $\delta(z, S) \coloneqq \min_{x \in S} \delta(z, x)$ denotes the (normalized/relative) Hamming distance from $z$ to the nearest point in $S$.
\end{definition}

Using this, we define a generalization of demi-bits generators 
that requires the adversary to certify remoteness from the range, 
rather than merely being a non-output.

\begin{definition}[Remote-Point Demi-Bits Generators]\label{def: remote-demi-bits}
    Let $n,m$ be length parameters such that $n<m$. Let $\tau$ be a remoteness parameter. A function $G:\{0,1\}^n\to \{0,1\}^m$ is an \emph{$(s,\eps,\tau)$-secure remote-point demi-bits generator} if there is no $\NP/\poly$ adversary $\Adv$ of size $s$ such that 
    \begin{align*}
        \Pr_{y\leftarrow\{0,1\}^m}[\Adv(y)=1]\ge\eps \quad\text{and}\quad\Pr_{z\leftarrow \calB(\Range(G),\tau)}[\Adv(z)=1]=0.
    \end{align*}
    In this paper, we consider security against all polynomial-size $\NP/\poly$ adversaries, so we simply write $(\eps,\tau)$-secure remote-point demi-bits generator. Equivalently, no efficient proof system can prove the (suitably encoded) statement ``$\delta(y,\Range(G))>\tau$'' for at least $\eps \cdot 2^m$ many strings $y$.%
\end{definition}

Our main construction also requires the disperser family to 
satisfy a Lipschitz property, ensuring that nearby inputs map 
to nearby outputs.

\begin{definition}[$(k, \eps, \alpha)$-Lipschitz Zero-Error Disperser Family]\label{def: Lipschitz disperser family}
    Let $\calF$ be a distribution over functions 
    $f: \{0,1\}^n \to \{0,1\}^m$, $k \ge m$, $\eps > 0$, 
    and $\alpha > 0$ be parameters. We say that $\calF$ is a 
    \emph{$(k, \eps, \alpha)$-Lipschitz zero-error disperser 
    family} if:
    \begin{enumerate}
        \item For every random variable $\calX$ over 
            $\{0,1\}^n$ with $H_\infty(\calX) \ge k$,
            \[\Pr_{f \gets \calF}[\Supp(f(\calX)) 
            = \{0,1\}^m] \ge 1 - \eps.\]
        \item For every $f \in \calF$ and every 
            $x_1, x_2 \in \{0,1\}^n$,
            \[\delta(f(x_1), f(x_2)) 
            \;\leq\; \alpha \cdot \delta(x_1, x_2).\]
    \end{enumerate}
\end{definition}

We note that the \hyperref[thm: disperser family by a projection]{Random Projection Disperser Family} $\Proj_{n \to m}$ is $(n-\Delta, 4m\Delta / n, n/m)$-Lipschitz: If $\delta(x_1, x_2) = \eps$, then $x_1$ and $x_2$ differs in only $\eps\cdot n$ many coordinates, which means that for any $f \in \Proj_{n\to m}$, $f(x_1)$ and $f(x_2)$ differs in only $\eps\cdot n$ many coordinates as well, therefore $\delta(f(x_1), f(x_2)) < \eps\cdot \frac{n}{m}$. %

We show that composing a remote-point demi-bits generator with a Lipschitz zero-error disperser family yields hard instances for the Remote Point problem.

\begin{theorem}\label{thm:remote-point general}
Let $G: \{0, 1\}^n \to \{0, 1\}^N$ be a $(1-2^{k-N},\tau/\alpha)$-secure remote-point demi-bits generator, and $\calF: \{0,1\}^N \to \{0,1\}^m$ be an efficient $(k,\eps,\alpha)$-Lipschitz zero-error disperser family. For any function $f\in \calF$ define the function $C_f:\{0,1\}^n \to \{0,1\}^m$ such that $\forall x\in \{0,1\}^n, C_f(x) = f(G(x))$. Let $\mathscr{C}$ be a circuit class where $C_f\in \mathscr{C}$ for all $f\in \calF$. 
Then for any errorless $\SearchNP$ heuristic $\calA$ for $\mathscr{C}$-$\RemotePoint[n,m,\tau]$, it holds that $\Pr_{f \leftarrow \calF}[\calA(C_f) \neq \bot] \le \eps$.

\end{theorem}
\begin{proof}
Suppose for contradiction that some nondeterministic polynomial-time algorithm $\calA$ satisfies $\Pr_{f \leftarrow \calF}[\calA(C_f) \neq \bot] >\eps$. We construct a nondeterministic polynomial-time adversary $\calB$ breaking the remote-point demi-bits generator $G$. On input $y\in \{0,1\}^N$, the adversary $\calB$ accepts if and only if there exists $f\in \calF$ such that some nondeterministic branch of $\calA(C_f)$ outputs $f(y)$.
\begin{description}
    \item[$\calB$ rejects every $y$ such that $\delta(y,\Range(G))\le \tau/\alpha$.] If there exists some seed $s\in \{0,1\}^{n}$ such that $\delta(y,G(s)) \le \tau/\alpha$, then for every $f \in \calF$ we have $C_f(s) = f(G(s))$ and $\delta(C_f(s),f(y))\le \tau$, hence $\delta(f(y),\Range(C_f))\le \tau$. Since $\calA$ is an errorless solver for $\mathscr{C}$-$\RemotePoint$, no accepting branch of $\calA(C_f)$ outputs an element $\tau$-close to $\Range(C_f)$, so $\calB(y)$ rejects.
    \item [$\calB$ accepts at least $2^N-2^k$ strings in $\{0,1\}^N$.] Assume to the contrary that $\calB$ accepts $<2^N-2^k$ strings of $\{0,1\}^N$. Let $X$ be the uniform distribution over the strings in $\{0,1\}^N$ rejected by $\calB$. Then we have $|\Supp(X)| \geq 2^k$ and  $H_\infty(X)\ge k$.  Since $\calF$ is an efficient $(k, \eps, \alpha)$-Lipschitz zero-error disperser family, we have
    \[
      \Pr_{f \leftarrow \calF} \bigl[\Supp(f(X)) = \{0,1\}^m\bigr] \;\geq\; 1 - \eps.
    \]
    Since $\Pr_f[\calA(C_f) \neq \bot] > \eps$, there must exist some $f \in \calF$ such that $\calA(C_f) \neq \bot$ and $\Supp(f(X)) = \{0,1\}^m$. This means that there exists some $x \in \Supp(X)$ such that some non-deterministic accepting branch of $\calA(C_f)$ outputs exactly $f(x)$, hence $\calB$ accepts $x$, a contradiction.

\end{description}

Hence $\calB$ is a nondeterministic polynomial-time adversary that rejects all strings $\tau/\alpha$-close to $\Range(G)$
and accepts $\ge 1-2^{k-N}$ fraction of $\{0,1\}^N$, contradicting the demi-bits
security of $G$.
\end{proof}

\subsection{Hardness of \texorpdfstring{$\XOR$-$\RemotePoint$}{XOR-Remote-Point}}

\autoref{thm:remote-point general} allows us to base the hardness of $\XOR$-$\RemotePoint$ on (demi-hardness variants of) the Learning Parity with Noise (LPN) assumption. LPN is a standard assumption in cryptography dating back to~\cite{BFKL93}: Roughly speaking, given a set of $m$ random linear equations over $n$ variables in $\GF(2)$ where $m \gg n$, the LPN assumption states that it is infeasible to distinguish between the case that the equations are randomly generated and that there exists a solution satisfying a $(1-\mu)$ fraction of equations simultaneously. We also consider a sparse variant of LPN formulated in~\cite{Feige02, ale03, ApplebaumBW10}, where each equation only mentions $k = O(1)$ many variables.

The LPN-style demi-bits assumption can be viewed as a nondeterministic analogue of the standard LPN assumption, with the mild relaxation that the noise vector is required to be sparse (bounded Hamming weight) rather than i.i.d.~Bernoulli. %

\begin{assumption}[LPN-Style Demi-Bits~\cite{BFKL93,DBLP:conf/stoc/ChenL24,RenWZ26}]\label{assumption: demi-hardness of LPN}
    For some (public) matrix $A\in \F_2^{m\times n}$, there is no polynomial-size non-uniform nondeterministic circuit $\calB: \{0, 1\}^m \to \{0, 1\}$ such that $\calB$ accepts a constant fraction of random strings but rejects every string of the form $A\vec{s} + \vec{e}$, where $\vec{e} \in \F_2^m$ is $(\mu\cdot m)$-sparse and $\vec{s} \in \F_2^n$.
\end{assumption}

It is easy to see that~\hyperref[assumption: demi-hardness of LPN]{LPN-Style Demi-Bits} Assumption implies a remote-point demi-bits generator computable by $\GF(2)$-linear circuits (circuits with only $\XOR$ gates).

Analogously, we define the \emph{Sparse LPN-Style Demi-Bits} assumption by further restricting each row of $A$ to be 
$k$-sparse, mirroring the relationship between standard LPN and Sparse LPN~\cite{ale03}. This can be viewed as a nondeterministic analogue of the Sparse LPN assumption.%
 
\begin{assumption}[Sparse LPN-Style Demi-Bits~\cite{ale03}]\label{assumption: sparse-demi-hardness of LPN}
    For some (public) matrix $A\in \F_2^{m\times n}$ where each row of $A$ has Hamming weight $\le k$, there is no polynomial-size non-uniform nondeterministic circuit $\calB: \{0, 1\}^m \to \{0, 1\}$ such that $\calB$ accepts a constant fraction of random strings but rejects every string of the form $A\vec{s} + \vec{e}$, where $\vec{e} \in \F_2^m$ is $(\mu\cdot m)$-sparse and $\vec{s} \in \F_2^n$.
\end{assumption}

Similarly, the \hyperref[assumption: sparse-demi-hardness of LPN]{Sparse LPN-Style Demi-Bits} Assumption implies a $\GF(2)$-linear circuit computable in $\NC^0_k$ that is a remote-point demi-bits generator.

Thanks to~\hyperref[assumption: demi-hardness of LPN]{LPN-Style Demi-Bits} Assumption, and the~\hyperref[thm: disperser family by a projection]{Random Projection Disperser Family}, we are able to show the hardness of $\XOR$-$\RemotePoint$ based on~\autoref{thm:remote-point general}.

\begin{theorem}[Hardness of $\XOR$-$\RemotePoint$]\label{thm:hardness-of-xor-RPP}
    Let $c \ge 10$, $\Delta \ge 1$, $\tau > 0$ be constants and let $N > cn$. If there exists a $(1-2^{-\Delta}, \tau/c)$-secure remote-point demi-bits generator $G: \{0, 1\}^n \to \{0, 1\}^N$ computable using only $\XOR$ gates, then there is a distribution $\calD$ of $\GF(2)$-linear circuits $C: \{0, 1\}^n \to \{0, 1\}^{N/c}$ (efficiently samplable given $G$) such that no errorless $\SearchNP$ heuristic for $\XOR$-$\RemotePoint[n, N/c, \tau]$ has success probability $> 4\Delta/c$ on $\calD$.
\end{theorem}
\begin{proof}
    This follows from instantiating~\autoref{thm:remote-point general} with the remote-point demi-bits generator $G$ and the \hyperref[thm: disperser family by a projection]{Random Projection Disperser Family} $\Proj_{N \to (N/c)}$ which is a $(N-\Delta, 4\Delta/c, c)$-Lipschitz zero-error disperser family. %
\end{proof}

\begin{corollary}\label{cor:hardness-of-xor-RPP-lpn}
    Fix constants $c \ge 10$, parameter $\tau > 0$, and let $N > cn$. If~\autoref{assumption: demi-hardness of LPN} holds with $\mu = \tau/c$ and $m = N$, then there is no errorless $\SearchNP$ heuristic for $\XOR$-$\RemotePoint[n, N/c, \tau]$ with success probability $\eps$ over the distribution $\calD$ from~\autoref{thm:hardness-of-xor-RPP}.
\end{corollary}

Based on~\autoref{thm:hardness-of-xor-RPP}, we comment on the optimality of the state-of-the-art algorithm for $\XOR$-$\RemotePoint$.

\begin{mdframed}[innertopmargin=-0.1em]
\small
\begin{remark}[Is the~{\cite{AlonPY09}} Algorithm Optimal?]
    Given a $\GF(2)$-linear circuit $C: \{0, 1\}^n \to \{0, 1\}^m$ with $m \ge 2n$, the $\XOR$-$\RemotePoint$ algorithm in~\cite{AlonPY09} finds a string $y \in \{0, 1\}^m$ that is $\tau$-far from $\Range(C)$, where $\tau = \Omega((\log n) / n)$. On the other hand, LPN with noise rate $\mu = O((\log n)/n)$ and $m \ge n/(1-\mu)$ is solvable in polynomial time via Gaussian elimination~\cite{Prange62}. LPN with noise rate $\mu = O((\log^2 n)/n)$ or even $\mu = O((\log^{1+\eps}n)/n)$ are not known to be polynomial-time solvable~\cite{BLSV18, AbramMR25, BCLV26}. However, LPN with noise rate $\mu = O((\log^2 n)/n)$ is in $\BPP^\SZK \subseteq \AM\cap\coAM$~\cite{BLVW19}, hence it is likely insecure against nondeterministic adversaries.
    
    It still seems possible that LPN with noise rate $\mu = \polylog(n) / n$ is secure against nondeterministic adversaries. If this is true, it would imply that the distance parameter in the~\cite{AlonPY09} algorithm for $\XOR$-$\RemotePoint$ cannot be improved to $\log^{\omega(1)} n$.
\end{remark}
\end{mdframed}

The same framework extends to 
$k$-$\XOR$-$\RemotePoint$ by replacing the LPN-style demi-bits assumption with its sparse variant~(\autoref{assumption: sparse-demi-hardness of LPN}).
\begin{corollary}\label{cor:hardness-of-k-xor-RPP}
    Fix constants $c \ge 10$, $k \ge 3$, parameter $\tau > 0$, and let $N > cn$. If~\autoref{assumption: sparse-demi-hardness of LPN} holds with $\mu = \tau/c$ and $m = N$, then there is no errorless $\SearchNP$ heuristic for $k$-$\XOR$-$\RemotePoint[n, N/c, \tau]$ with success probability $\eps$ over the distribution $\calD$ from~\autoref{thm:hardness-of-xor-RPP}.
\end{corollary}

\begin{mdframed}[innertopmargin=-0.2em, skipabove=-0.2em]
\small
\begin{remark}\label{rem: k-xor-rpp}
\cite{GLY26} gives a polynomial-time algorithm for
$k$-$\XOR$-$\RemotePoint[n, \tilde{\Theta}(n^{k/2}), \eps]$.
This may be essentially optimal: note that nondeterministically strongly refuting $k$-$\XOR$ (in the sense of~\autoref{def:strong-xor} where ``strong'' means certifying the system is far from satisfiable) is equivalent to breaking $k$-$\XOR$ Remote-Point Demi-Bits against $\NP$. By~\autoref{cor:hardness-of-k-xor-RPP},
any improvement in stretch would yield improved
nondeterministic algorithms for strongly refuting
semi-random $k$-$\XOR$ instances, which in turn imply better such algorithms for refuting semi-random CSPs~\cite{Feige02,COCF10}. The current
$\widetilde{O}(n^{k/2})$ threshold for strong refutation
is believed to be tight for polynomial-time
algorithms~\cite{KMOW17}, and no nondeterministic approach
is known to improve upon it. (See~\autoref{sec:csp-literature} for a table of the
literature on refuting random, semi-random, and smoothed
CSPs, covering both weak and strong refutation as well as
deterministic and nondeterministic algorithms.)

While nondeterministic
algorithms that work with polynomially fewer constraints
do exist~\cite{FKO06, GKM22}, they only achieve
\emph{weak} refutation (i.e., they only certify the system is \emph{unsatisfiable}, not \emph{far from satisfiable}). These algorithms are based on
even-cover witnesses, where each disjoint even cover
certifies that at most one constraint must be violated.
Strong refutation by this counting argument would require $\Omega(m)$ disjoint even
covers, hence covers of constant length, but by Feige's
conjecture~\cite{fei08} (resolved in~\cite{GKM22}),
constant-length even covers already require
$m \geq \widetilde{\Omega}(n^{k/2})$ hyperedges ---
precisely the spectral threshold at which deterministic
algorithms suffice.  
\end{remark}
\end{mdframed}

\subsection{Average-Case Hard Partial Truth Tables is Hard}

With respect to randomized algorithms, agnostic PAC-learning 
$\ckt$ is equivalent to strong RRHS-refutation of the dual 
class $\ckt^*$. This extends~\cite{Vadhan17}, who established 
the analogous equivalence in the realizable setting; 
see~\autoref{sec:agnostic} for a self-contained proof. 
We note that~\cite{KL18} obtained a similar 
result (in different terminology) in the agnostic setting, 
but their equivalence is distribution-specific, whereas ours 
is distribution-free. The proof techniques also differ: we use 
Yao's next-bit prediction combined with agnostic boosting, 
while~\cite{KL18} use a hybrid argument combined with 
Kalai--Kanade boosting~\cite{KalaiK09}.

To simplify our terminology, we will also refer to strong RRHS refutation algorithms for the dual class of $\ckt$ as ``Agnostic PAC-learning algorithms for $\ckt$'' in the nondeterministic setting.
\begin{definition}[Nondeterministic $\tau$-Agnostic PAC-Learning]\label{def: nd-agnostic-pac-learning}
    Let $\ckt$ be a class of Boolean functions, $L\ge 1$, $p\in (0, 1)$, and $\calD$ be a distribution over $\{0, 1\}^n$. A $\SearchNP$ algorithm $\calA$ is said to \emph{$\tau$-agnostic PAC-learn} $\ckt$ over input distribution $\calD$ using $L$ samples with success probability $p$ if:
    \begin{itemize}
        \item {\it (Soundness)} for every inputs $x_1, \dots, x_L\in \{0, 1\}^n$ and labels $b_1, \dots, b_L \in \{0, 1\}$, if there exists a function $C \in \ckt$ such that $C(x_i) = b_i$ for $1-\tau$ fraction of $i\in [L]$, then $\calA((x_1, b_1), \dots, (x_L, b_L))$ rejects. 
        \item {\it (Completeness)} If we sample i.i.d.~$x_i\gets \calD$ and $b_i \gets \{0, 1\}$, then
        \[\Pr[\calA((x_1, b_1), \dots, (x_L, b_L))\text{ accepts}] \ge p.\]
    \end{itemize}
    Given $q\in(0,1]$ and $\eta\in[0,1)$, we say that
$\calA$ has $(q,\eta)$-typical completeness over $\calD$ if
    \begin{itemize}
        \item {\it ($(q,\eta)$-Typical Completeness)} 
        \[\Pr_{\vec{x}\leftarrow \calD^L}\sbra{\Pr_{\vec{b}\leftarrow\{0,1\}^L}\sbra{\calA((x_1, b_1), \dots, (x_L, b_L))\text{ accepts}}\ge q}\ge 1-\eta.\]
    \end{itemize}
    Notice that $(q,\eta)$-typical completeness implies ordinary
completeness with success probability at least $(1-\eta)q$.
    
    We also say that a $\SearchNP$ algorithm \emph{distribution-freely} agnostic PAC-learns $\ckt$ with success probability $p$ if it also satisfies the following stronger (Completeness) condition:
    \begin{itemize}
        \item {\it (Completeness')} For every $x_1, \dots, x_L\in \{0, 1\}^n$, if we sample i.i.d.~labels $b_1, \dots, b_L \gets \{0, 1\}$, then
        \[\Pr[\calA((x_1, b_1), \dots, (x_L, b_L))\text{ accepts}]\ge p.\]
    \end{itemize}
\end{definition}

Before connecting~\autoref{def: nd-agnostic-pac-learning} with the hardness of $\PartialAvgHard$, we recall the definition of errorless heuristics for $\PartialAvgHard$.
\begin{definition}[Errorless Heuristics for $\PartialAvgHard$]
    Let $\ckt$ be a class of functions, $L \ge 1$, $\tau \in (0,1)$, $p \in (0,1)$, and $\calD$ be a distribution over $\{0,1\}^n$. A $\SearchNP$ algorithm $\calA$ is an \emph{errorless heuristic} for $\ckt$-$\PartialAvgHard$ with remoteness $\tau$, success probability $p$, and domain size $L$ over $\calD$ if:
    \begin{itemize}
        \item {\it (Soundness)} For every $x_1, \dots, x_L \in \{0,1\}^n$, every accepting path of $\calA(\vec{x})$ outputs bits $b_1, \dots, b_L \in \{0,1\}$ such that the partial function $\{x_i \mapsto b_i\}_{i \in [L]}$ has relative distance $\ge \tau$ from every function in $\ckt$.
        \item {\it (Completeness)} When $x_1, \dots, x_L \leftarrow \calD$ are i.i.d.~samples, the algorithm $\calA(\vec{x})$ has at least one accepting path with probability at least $p$.
    \end{itemize}

    If $\calA(\vec{x})$ has at least one accepting path for \emph{every} $x_1, \dots, x_L \in \{0,1\}^n$, we say $\calA$ is a \emph{worst-case algorithm} for $\ckt$-$\PartialAvgHard$.
\end{definition}

The main result of this section shows that nondeterministic hardness of agnostic PAC-learning yields nondeterministic hardness of $\PartialAvgHard$:

\begin{theorem}\label{thm: PartialAvgHard to PAC learning for NP}
    Let $\ckt$ be a class of functions and $L, L'\ge \omega(1)$ be parameters such that $L' \le 0.1L$. Then:
    \begin{itemize}
        \item Suppose there is a (worst-case) $\SearchNP$ algorithm $\calA$ for $\ckt$-$\PartialAvgHard$ with remoteness $\tau$ on domain size $L'$. Then there is a $\SearchNP$ algorithm $\calB$ that distribution-freely $\frac{\tau L'}{2L}$-agnostic PAC-learns $\ckt$ using $L$ samples with success probability $1-o(1)$.
        \item Let $\calD$ be a distribution over $\{0, 1\}^n$ and $\eps\in(0,3/5)$. Suppose there is an errorless $\SearchNP$ heuristic $\calA$ for $\ckt$-$\PartialAvgHard$ with remoteness $\tau$ over $\calD$ on domain size $L'$ with success probability $1-\eps$. Then there is a $\SearchNP$ algorithm $\calB$ that $\frac{\tau L'}{2L}$-agnostic PAC-learns $\ckt$ on input distribution $\calD$ using $L$ samples with success probability at least $(1-5\eps/3)(1-2^{-0.1(L/L')}) \ge 1 - 2^{-0.1(L/L')} - 2\eps$. Moreover, $\calB$ has $(1-2^{-0.1(L/L')},5\eps/3)$-typical completeness. %
    \end{itemize}
\end{theorem}
\begin{proof}
Let $\calA$ be an errorless $\SearchNP$ heuristic (or a worst-case algorithm) for $\ckt$-$\PartialAvgHard$ with remoteness $\tau$ on domain size $L'$. Our $\SearchNP$ algorithm $\calB$, given inputs $(x_1,b_1),\dots,(x_L,b_L)$, guesses a subset $I\subseteq [L]$ of size $|I| = L'$, and verifies that there is a nondeterministic branch of $\calA(\cbra{x_i}_{i\in I})$ that outputs a sequence $\cbra{b'_i}_{i\in I}$ that is $(\tau/2)$-close to $\cbra{b_i}_{i\in I}$. If this is true, then $\calB$ accepts the input, otherwise $\calB$ rejects the input.

\begin{description}
    \item[Soundness of $\calB$:] Consider any inputs $x_1,\dots,x_L \in \{0,1\}^n$ and labels $b_1,\dots,b_L\in \{0,1\}$. Suppose there is a function $C\in\ckt$ such that $C(x_i) = b_i$ for $\ge 1-\frac{\tau L'}{2L}$ fraction of $i\in [L]$. Then for every subset $I\subseteq [L]$ of size $L'$, the partial function $\cbra{(x_i \mapsto b_i)}_{i\in I}$ agrees with $C$ for $\ge 1-\tau/2$ fraction of $i\in I$. Now suppose $\calA(\cbra{x_i}_{i\in I})$ outputs bits $\cbra{b'_i}_{i\in I}$ with $\delta(\vec{b}', \vec{b}|_I) < \tau/2$. By the triangle inequality,
    \[\delta(\vec{b}', C|_I) \le \delta(\vec{b}', \vec{b}|_I) + \delta(\vec{b}|_I, C|_I) < \tau/2 + \tau/2 = \tau,\]
    contradicting the soundness of $\calA$ (which guarantees $\delta(\vec{b}', C|_I) \ge \tau$). Hence $\calB$ never accepts $\cbra{(x_i,b_i)}_{i\in [L]}$.

    \item [Completeness of $\calB$:] Let $x_1,\dots,x_L \in \{0,1\}^n$ and define 
    \begin{align*}
        \Rej := \Rej(x_1, \dots, x_L) := \{\vec{b} \in \{0, 1\}^L: \calB(\{(x_i, b_i)\})\text{ rejects}\}.
    \end{align*}
    Suppose $I \subseteq [L]$, $|I| = L'$ is a subset shattered by $\Rej$. We argue that $\calA(\cbra{x_i}_{i\in I})$ does not have any accepting branch. Indeed, if $\calA(\cbra{x_i}_{i\in I})$ outputs any sequence of bits $\cbra{b'_i}_{i\in I}$, there exists $\vec{b} \in \Rej$ such that $b_i=b'_i$ for every $i\in I$. Then $\delta(\vec{b}', \vec{b}|_I) = 0 < \tau/2$, so $\calB(\cbra{(x_i,b_i)})$ accepts, contradicting $\vec{b} \in \Rej$.
        \begin{itemize}
        \item Suppose that $\calA$ is a worst-case $\SearchNP$ algorithm for $\ckt$-$\PartialAvgHard[n,s,L',\tau]$. Then for every $x_1, \dots, x_L$, there is no subset $I\subseteq[L]$ of size $L'$ shattered by $\Rej(x_1, \dots, x_L)$. By the \hyperref[lemma: Sauer--Shelah]{Sauer--Shelah Lemma}, $|\Rej(x_1, \dots, x_L)|\le \sum_{i=0}^{L'-1}\binom{L}{i} \le o(2^L)$. It follows that $\calB$ distribution-freely PAC-learns $\ckt$ using $L$ samples with success probability $1-o(1)$.
        \item Suppose that $\calA$ is an errorless $\SearchNP$ heuristic for $\ckt$-$\PartialAvgHard[n,s,L',\tau]$ over $\calD$ with completeness $1-\eps$. Set $\Delta := 0.1(L/L')$ and define
        \[q:=\Pr_{(x_1, \dots, x_L)\gets \calD^L}\sbra{|\Rej|\ge 2^{L-\Delta}}.\]
        Fix any $\vec{x} = (x_1, \dots, x_L)$ such that $|\Rej|\ge 2^{L-\Delta}$.
        By~\autoref{thm: disperser family by a projection}, a random subset $I \subseteq [L]$ of size $L'$ is 
        shattered by $\Rej$ with probability at least~$1-4\Delta\frac{L'}{L}=0.6$. 
    
    Since the marginal distribution of $(x_i)_{i\in I}$ is exactly $\calD^{L'}$, we obtain 
    \[q\mleft(1-4\Delta\frac{L'}{L}\mright)\le\Pr_{\vec{z}\leftarrow\calD^{L'}}[
    \calA(\vec{z})\text{ has no accepting path}]\le\eps.\] Hence, $q\le 5\eps/3 $.
    For every $\vec{x}$ such that $|\Rej(\vec{x})|<2^{L-\Delta}$, we have 
    \[\Pr_{\vec{b}\leftarrow\{0,1\}^L}
\left[
    \calB\bigl((x_i,b_i)_{i\in[L]}\bigr)
    \text{ accepts}
\right]=
1-\frac{|\Rej(\vec{x})|}{2^L}
>
1-2^{-\Delta}.\]
     Hence, it holds that 
        \[\Pr_{\vec{x}\leftarrow \calD^L}\sbra{\Pr_{\vec{b}\leftarrow\{0,1\}^L}\sbra{\calB((x_1, b_1), \dots, (x_L, b_L))\text{ accepts}}\ge 1-2^{-\Delta}}\ge 1-q \ge 1-5\eps/3.\]
        Thus $\calB$ has $\left(1-2^{-\Delta},5\eps/3\right)$-typical completeness.
        Finally, 
        \[\Pr_{\overset{\vec{x}\leftarrow \calD^L}{\vec{b}\leftarrow\{0,1\}^L}}\sbra{\calB((x_1, b_1), \dots, (x_L, b_L))\text{ accepts}}\ge (1-q)(1-2^{-\Delta})\ge 1-2^{-0.1(L/L')} - 2\eps.\]
        Hence, $\calB$ PAC-learns $\ckt$ over input distribution $\calD$ using $L$ samples with success probability
        \[1-2^{-0.1(L/L')} - 2\eps.\qedhere\]
    \end{itemize}
\end{description}
\end{proof}

\subsubsection{Hardness of \texorpdfstring{$k$-$\XOR$-$\PartialAvgHard$}{k-XOR-Partial-AvgHard}}

To prove the hardness of $k$-$\XOR$-$\PartialAvgHard$, we introduce 
an assumption closely related to~\autoref{assumption: sparse-demi-hardness of LPN}, that essentially states the hardness of agnostic learning $k$-$\XOR$.

\begin{assumption}[Sparse-Secret LPN-Style Demi-Bits for Random Matrices~\cite{ale03}]\label{assumption: demi-hardness of random-Sparse-LPN}
    Let $k=k(n)$ satisfy $\omega(1)\le k \le n$, $\mu  = \mu(n) \in (0,1/2)$, $\delta \in (0,1/2)$, and $m=m(n)=\poly(n)$.
    Fix a distribution $\calD$ (say, the uniform distribution) over $\F_2^n$. We additionally assume
        \begin{equation}\label{eq:sparse-secret-LPN-parameters}
            m(1-H(\mu)) \ge \log_2\pbra{\sum_{j=0}^k {n \choose j}}+1,
        \end{equation}
    where $H(\cdot)$ is the binary entropy function. 
    For a random matrix $A\in \F_2^{m\times n}$ with rows drawn 
    i.i.d.~from $\calD$, with probability $\ge 1/2$ over $A$, 
    there is no polynomial-size non-uniform nondeterministic 
    circuit $\calB: \{0, 1\}^m \to \{0, 1\}$ such that $\calB$ 
    accepts at least a $1/2-\delta$ fraction of random strings but rejects 
    every string of the form $A\vec{s} + \vec{e}$, where 
    $\vec{e} \in \F_2^m$ is $(\mu\cdot m)$-sparse and 
    $\vec{s} \in \F_2^n$ is $k$-sparse.
\end{assumption}

The condition \eqref{eq:sparse-secret-LPN-parameters} guarantees that the assumption is not information-theoretically vacuous: the number of strings of the form $A\vec{s} + \vec{e}$ is at most $2^{mH(\mu)}\cdot \mleft(\sum_{j=0}^k\binom{n}{j}\mright) \le 2^{m-1}$. If $\mu < 1/2$ is a constant and $k\le n/2$, it suffices that $m = \Omega(k\log(en/k))$. Also note that the assumption is meaningful only in the regime $k=\omega(1)$. If $k = O(1)$, one can always break the assumption in deterministic $O(n^k)$ time.%

The two assumptions are not directly comparable:~\autoref{assumption: sparse-demi-hardness of LPN} restricts the rows of $A$ to be $k$-sparse and allows an unrestricted secret, whereas the present assumption allows unrestricted rows but requires the secret to be $k$-sparse. In addition, the present assumption requires hardness for at least half of the random matrices rather than for one fixed matrix.
Concretely, by~\autoref{thm: PartialAvgHard to PAC learning for NP}, nondeterministic hardness of $\PartialAvgHard$ for a circuit class $\ckt$ follows from nondeterministic hardness of agnostically PAC-learning $\ckt$. When $\ckt$ is the class of $k$-$\XOR$s, the relevant LPN-style assumption must restrict the secret $\vec{s}$ to be $k$-sparse, since $a \mapsto \langle a, \vec{s}\rangle$ is a $k$-$\XOR$ function under precisely this restriction. No sparsity condition on the rows of $A$ is required for this implication.

\begin{corollary}\label{cor: sparse secret LPN to PartialAvgHard}
    Assuming the~\hyperref[assumption: demi-hardness of random-Sparse-LPN]{Sparse-Secret LPN Assumption for Random Matrices} with distribution $\calD$, let $m=L$, and assume that $10\mid L$.
    For any constant $\tau \in (20\mu, 1/2)$, there is no errorless 
    $\SearchNP$ heuristic for 
    $k$-$\XOR$-$\PartialAvgHard[n, s, L/10, \tau]$ over $\calD$ 
    with success probability at least $3/4$.
\end{corollary}
\begin{proof}
    Suppose for contradiction that such a heuristic $\calA$ exists.
    Set $L':=L/10$ and $\eps:=1/4$.
    Then $\Delta:=0.1L/L'=1$.
    By \autoref{thm: PartialAvgHard to PAC learning for NP},
    there is a $\SearchNP$ algorithm $\calB$ that
    $\tau/20$-agnostically PAC-learns $k$-$\XOR$ using $L$
    samples and has $(1/2,5/12)$-typical completeness. Equivalently,
    \[
    \Pr_{A\leftarrow\calD^L}
    \left[
        \Pr_{y\leftarrow\{0,1\}^L}
        \left[
            \calB\bigl((a_i,y_i)_{i\in[L]}\bigr)
            \text{ accepts}
        \right]
        \ge \frac12
    \right]
    \ge \frac{7}{12},
    \]
    where $a_1,\ldots,a_L$ are the rows of $A$.

    For each fixed matrix $A$, define
    \[
        \calB_A(y)
        :=
        \calB\bigl((a_i,y_i)_{i\in[L]}\bigr).
    \]
    Since $A$ is fixed, it can be hardwired into $\calB$, so
    $\calB_A$ is a polynomial-size non-uniform nondeterministic
    circuit.

    We claim that $\calB_A$ rejects every sparse-secret LPN string.
    Indeed, suppose $y=A\vec{s}+\vec{e}$, $|\vec{s}|\le k$, and $|\vec{e}|\le\mu L$.
    The function $f_{\vec{s}}(a):=\abra{a,\vec{s}}$ is a $k$-$\XOR$, and $\delta(y,(f_{\vec{s}}(a_i))_{i\in[L]})\le\mu<\tau/20$.
    Therefore, the soundness of the agnostic learner implies that
    $\calB_A(y)$ rejects.

    Thus, for at least a $7/12$ fraction of
    $A\leftarrow\calD^L$, there is a sound nondeterministic circuit
    $\calB_A$ accepting at least half of all uniformly random
    right-hand sides.

    On the other hand,
    \hyperref[assumption: demi-hardness of random-Sparse-LPN]{Sparse-Secret LPN Assumption for Random Matrices}
    states that a set of matrices of probability at least $1/2$
    admits no sound circuit accepting even a
    $1/2-\delta$ fraction of uniformly random right-hand sides.
    Since $7/12+1/2>1$, these two sets of matrices intersect. For a matrix in the
    intersection, $\calB_A$ accepts at least $1/2\ge1/2-\delta$ of all random strings while rejecting every sparse-secret LPN
    string, contradicting the assumption.
\end{proof}

\subsubsection{Hardness from Random \texorpdfstring{$k$}{k}-SAT}
Finally, we observe that in suitable parameter regimes, $\DNF$-$\PartialHard$ reduces to both $\AND$-$\PartialAvgHard$ and $\XOR$-$\PartialAvgHard$. This is because both conjunctions and parities can weakly approximate DNFs:
\begin{fact}\label{fact: conjunctions and parities weakly approximate DNFs}
    Let $F = T_1\lor T_2 \lor\dots \lor T_t$ be a DNF with $t\ge1$ terms, each of width at most $w$. Then, for every distribution $\calD$ over the input space, there exist a conjunction $C$ of width at most $w$ and a possibly complemented parity function $P$ on at most $w$ variables that both have $\Omega(1/t)$ correlation with $F$ over $\calD$:
    \begin{align*}
        \Pr_{x\gets \calD}[C(x) = F(x)]\ge&\, \frac{1}{2} + \frac{1}{O(t)}; & \textnormal{(\cite[Section 3.5]{schapire90})}\\
        \text{ and }\Pr_{x\gets \calD}[P(x) = F(x)]\ge&\, \frac{1}{2} + \frac{1}{O(t)}. & \textnormal{(\cite[Fact 8]{Jackson97})}
    \end{align*}
    Here constant functions are allowed.
\end{fact}

This allows us to derive hardness of $\AND$-$\PartialAvgHard$ and $\XOR$-$\PartialAvgHard$ directly from \autoref{cor: hardness of DNF-HPTT from random k-SAT}.

\begin{corollary}\label{cor: hardness of AND-AvgHPTT from random k-SAT}
    Under the \hyperref[assumption: nondeterministic Feige]{Random $k$-SAT Hypothesis Against $\AM$}, for every constant $\eta\in(0,1)$ and integer $c\ge1$, there are an integer $q\ge1$ and an input distribution $\calD$ over $\{0,1\}^n$ such that neither $\AND_{\lceil n^\eta\rceil}$-$\PartialAvgHard$ nor $\XOR_{\lceil n^\eta\rceil}$-$\PartialAvgHard$ has an errorless $\SearchNP$ heuristic over $\calD$ on domain size $n^c$, with remoteness $\frac12-\frac{1}{q\log n}$ and success probability $0.999$.
\end{corollary}
\begin{proof}
    We first claim that every errorless $\SearchNP$ heuristic producing labels at relative distance at least $\frac12-\frac1{Kt}$ from all width-$w$ conjunctions also gives an errorless $\SearchNP$ heuristic for $\PartialHard$ against size-$s$ DNFs with top fanin at most $t$ and term width at most $w$, for a sufficiently large absolute constant $K$. Here conjunctions may use negated literals and constants. Indeed, let $\{x_i \mapsto b_i\}_{i\in[L]}$ be a partial function computable by such a DNF $F$, and let $\calD$ be the uniform distribution on the indexed samples. Then by \autoref{fact: conjunctions and parities weakly approximate DNFs}, there is a width-$w$ conjunction $C$ such that
    \[
        \Pr_{i\gets[L]}[C(x_i)=b_i]\ge\frac12+\frac{1}{Kt}.
    \]
    Hence, any partial function at distance at least $\frac12-\frac1{Kt}$ from all width-$w$ conjunctions cannot be computed by such a DNF.

    Fix $\eta,c$, choose integers $a>1/\eta$ and $d>ac$, and choose $k$ from the hypothesis for density exponent $d$. Let $t:=O(d2^k\log n)$. By \autoref{cor: hardness of DNF-HPTT from random k-SAT}, to refute random $k$-SAT instances on $n$ variables with $n^d$ clauses with error $0.01$ by an $\AM$ algorithm, it suffices to design an errorless $\SearchNP$ heuristic for $\ckt$-$\PartialHard$ over $\calD$ on domain size
    \[
        L':=\left\lfloor\frac{n^d}{20t\log 100}\right\rfloor
    \]
    with success probability at least $0.9975$, where $\ckt$ is the class of DNFs over $n':=2nt=O_{k,d}(n\log n)$ inputs with size $s:=O(nt)$ and top fanin $t$, and $\calD$ is the distribution in \autoref{thm: reduction in DS16}.

    Set $\widetilde{n}=n^a$ and pad each input $x$ to $(x,\neg x,1,0,\ldots,0)\in\{0,1\}^{\widetilde{n}}$; let $\widetilde{\calD}$ be the resulting distribution. For all sufficiently large $n$, $2n'+2\le \widetilde{n}$, $n'+1\le\lceil \widetilde{n}^\eta\rceil$, $\widetilde{n}^c\le L'$.
    Each DNF term has width at most $n'$ after simplification, and the padding realizes negated literals and constants. Choose an integer $q$ large enough that $Kt\le q\log \widetilde{n}$. If the stated $\AND_{\lceil \widetilde{n}^\eta\rceil}$-$\PartialAvgHard$ heuristic exists, apply it to the padded first $\widetilde{n}^c$ of the $L'$ inputs and extend its output by arbitrary labels. By the first claim, no DNF in $\ckt$ computes these labels, already on the first $\widetilde{n}^c$ inputs. Its success probability remains $0.999\ge0.9975$, giving the required $\ckt$-$\PartialHard$ heuristic and contradicting the hypothesis. For other input lengths, define $\widetilde{\calD}$ using $n=\lfloor \widetilde{n}^{1/a}\rfloor$; the contradiction uses $\widetilde{n}=n^a$.

    The proof for $\XOR_{\lceil \widetilde{n}^\eta\rceil}$-$\PartialAvgHard$ is analogous: use the parity guarantee in~\autoref{fact: conjunctions and parities weakly approximate DNFs} and the constant-one coordinate to realize complemented parities with at most one additional input.
\end{proof}

However, compared to~\autoref{cor: sparse secret LPN to PartialAvgHard}, one weakness of \autoref{cor: hardness of AND-AvgHPTT from random k-SAT} is that it only derives hardness of $\PartialAvgHard$ with an inapproximability parameter very close to $1/2$.

\section*{AI Disclosure}
During the preparation of this work, large language models (Gemini and Claude) were employed to aid in literature review, parameter computation, and the proofs in~\autoref{sec: disperser families}. The authors verified the correctness and originality of all content including references, and take full responsibility for the contents in this paper.

\ifnum\Anonymity=0
\section*{Acknowledgments}
    We thank Yunqi Li for discussions regarding~\cite{AlonPY09}, Jan \Krajicek for pointing us to a proof of~\autoref{thm: one proof complexity generator against all proof systems} using an optimal proof system with one bit of advice, and anonymous reviewers for helpful comments.
\fi

\printbibliography[heading=bibintoc]

\appendix

\section{Omitted Proofs}
\subsection{Proof of \autoref{fact: AM adversaries to NPpoly adversaries}}\label{appendix: proof of AM to NPpoly}
\FactAMtoNPpoly*
\begin{proof}

    Let $\calA(C, r, w)$ denote an $\AM$ algorithm that refutes random $k$-SAT with $m(n)$ clauses with error $\eps$, where $r$ denotes Arthur's randomness and $w$ denotes Merlin's proof. We repeat this algorithm $t := \poly(m(n), \log (1/\eps))$ times: let $\calB(C, r_{1 \dots t}, w_{1\dots t})$ denote another $\AM$ algorithm that accepts if the majority of $\calA(C, r_1, w_1), \dots, \calA(C, r_t, w_t)$ accepts. As usual, we use ``$\calB(C)$ accepts'' to denote the event (over the choice of $r_{1\dots t}$) that there exists $w_{1\dots t}$ such that $\calB(C, r_{1\dots t}, w_{1\dots t})$ accepts. We can see that:
    \begin{itemize}
        \item If $C$ is a satisfiable $k$-CNF with $m(n)$ clauses, then by a Chernoff bound, $\Pr[\calB(C)\text{ accepts}]\le 2^{-\Omega(t)}$.
        \item We say that a $k$-CNF is \emph{good} if $\Pr[\calA(C)\text{ accepts}] \ge 3/4$, then by a Markov bound, a random $k$-CNF is good w.p.~$\ge 1-4\eps$. If $C$ is good, then by a Chernoff bound, $\Pr[\calB(C)\text{ accepts}] \ge 1-2^{-\Omega(t)}\ge 1-\eps$. It follows from a union bound that a random $k$-CNF is accepted by $\calB$ with probability $\ge 1-5\eps$.
    \end{itemize}
    Now choose $r_{1\dots t}$ uniformly at random and hardwire it into the algorithm $\calB$. Call the resulting algorithm $\calB'$; clearly, $\calB'$ is an $\NP/_\poly$ algorithm. We claim that with nonzero probability, $\calB'$ refutes random $k$-SAT with $m(n)$ clauses.
    \begin{itemize}
        \item Since there are only $2^{\tilde{O}(m(n))}$ many $k$-CNFs, it follows from a union bound that with probability at least $1 - 2^{-\Omega(t)}\cdot 2^{m(n)}\ge 0.99$, $\calB'$ rejects every satisfiable $k$-CNF.
        \item By a Markov bound, at least half of random strings $r_{1\dots t}$ satisfy that $\calB'$ accepts a random $k$-CNF with probability at least $1 - 10\eps$.
    \end{itemize}
    Hence, with probability at least $0.49$, the algorithm $\calB'$ is an $\NP/_\poly$ algorithm that refutes random $k$-SAT with $m(n)$ clauses with error $10\eps$.
\end{proof}

\subsection{Proof of \autoref{thm: reduction in DS16}}\label{appendix: proof of DS16}
\ThmDS*
\begin{table}[h]
    \centering
    \begin{tabular}{|c|c|c||c||c|c|c|c|}
        \hline
        \multicolumn{3}{|c||}{\bf Random $k$-SAT} &  & \multicolumn{4}{|c|}{\bf Learning CNFs} \\
        \hline
        $n$ & $k$ & $m$ & $\delta$ & $s$ & $t$ & $L$ & $n'$\\
        \hline 
        \#inputs & arity & \#clauses & fail.~prob. & CNF size & top fanin & domain size & input length\\
        \hline
        & & & & $O(nt)$ & $O(2^k\log(m/\delta))$ & $\lfloor m/t\rfloor$ & $2nt$\\
        \hline
    \end{tabular}
    \caption{Parameters in \autoref{thm: reduction in DS16}}
\end{table}

\begin{proof}
    Let $\varphi(x) = C_1(x) \land C_2(x) \land\dots \land C_m(x)$ denote the input CNF, where $C_i$ denotes the $i$-th clause of $\varphi$. Let $t := O(2^k \log(m/\delta))$, we first group the clauses into CNFs of top fanin $t$. That is, we define
    \[C'_i(x) = C_{(i-1)t + 1}(x) \land C_{(i-1)t + 2}(x) \land \dots \land C_{it}(x).\]
    Consider the list $\{(C'_i, 1)\}_{i\in [\lfloor m/t\rfloor]}$. We say that an input $x$ \emph{satisfies} the list if for every $i\in [\lfloor m/t\rfloor]$, $C'_i(x) = 1$. Observe that if $\varphi$ is satisfiable, then the same assignment satisfies the list; if $\varphi$ is a random $k$-CNF, then every $C'_i$ is a random $k$-CNF of top fanin $t$.

    Then, we replace each item in the list $\{(C'_i, 1)\}$ with $(\text{random $k$-CNF}, 0)$ with probability $1/2$. More precisely, we obtain a new list $\{(C''_i, b_i)\}$ where for each $i\in [\lfloor m/t\rfloor]$:
    \begin{itemize}
        \item with probability $1/2$, $C''_i = C'_i$ and $b_i = 1$;
        \item otherwise, $C''_i$ is a random $k$-CNF of top fanin $t$, and $b_i = 0$.
    \end{itemize}
    Note that if $\varphi$ is satisfied by an assignment $x$, then with probability $\ge 1-\delta$, $x$ also satisfies the new list in the sense that $C''_i(x) = b_i$ for every $i \in [\lfloor m/t\rfloor]$. This is because the probability that $x$ satisfies a random $k$-CNF is at most $(1-2^{-k})^t \le \delta/m$; hence by a union bound, with probability at least $1-\delta$, for every $i \in [\lfloor m/t\rfloor]$, if $b_i = 0$ then $C''_i(x) = 0$. On the other hand, if $\varphi$ is a random $k$-CNF, then every $C''_i$ is a random $k$-CNF of top fanin $t$ and every $b_i$ is a random bit.

    Finally, it is well known that one can encode a $k$-CNF $C$ over $n$ input bits with top fanin $t$ into $2n\cdot t$ bits $\langle C\rangle \in \{0, 1\}^{2n\cdot t}$, and transform an input $x \in \{0, 1\}^n$ into a CNF $\phi_x$ over $2n\cdot t$ input bits with size $s = O(nt)$ and top fanin $t$, such that $C(x) = \phi_x(\langle C\rangle)$. (See, e.g., \cite[Proposition 10]{Vadhan17}.) Our reduction $\calR(\varphi)$ outputs the partial function
    \[\{(\langle C''_i\rangle, b_i)\}_{i \in [\lfloor m/t\rfloor]}.\]

    Suppose that $\varphi$ is satisfied by an assignment $x$. Then with probability $\ge 1-\delta$, the CNF $\phi_x$ computes the above partial function. On the other hand, let $\calD$ denote the encoding of a random $k$-CNF with top fanin $t$. If $\varphi$ is a random $k$-CNF, then every $(\langle C''_i\rangle, b_i)$ is an i.i.d.~sample from $\calD \times \{0, 1\}$.
\end{proof}

\hanlin{Perhaps we should also incorporate ``agnostic PAC learning = strong RRHS refutation'' and ``refuter implies learning'' into this section.}

\section{Literature on Refuting CSP}\label{sec:csp-literature}

The table below summarizes known results on CSP refutation,
following the presentation in Table~1 of~\cite{AOW15}.
Throughout, the arity $k$ is fixed and $m=\poly(n)$.
Unless otherwise specified, the instances are fully random.

The ``CSP'' column specifies the predicate and instance model.
The ``Strength'' column describes the certified upper bound on
the fraction of simultaneously satisfiable constraints:
a ``refutation'' certifies unsatisfiability;
a ``$1/n^{\Omega(1)}$-refutation'' certifies an upper bound of
$1-1/n^{\Omega(1)}$;
and a ``strong'' refutation, as used here, certifies an upper bound of
$\rho(P)+\eta$ for any fixed $\eta>0$, where
$\rho(P):=|P^{-1}(1)|/2^k$.
The last guarantee is near-optimal, since a uniformly random assignment
satisfies exactly a $\rho(P)$ fraction of the constraints in expectation.

The ``Eff./Exist.''~column distinguishes \emph{efficient} refutation
by a deterministic polynomial-time algorithm from
\emph{existential} refutation by polynomial-size certificates
verifiable in deterministic polynomial time (i.e.,~a \emph{nondeterministic} polynomial-time
refutation algorithm). %
All certificates are sound on every input instance, while their
existence or efficient discovery is guaranteed with high probability
over the randomness of the specified instance model.

In the bounds below, implicit constants may depend on $k$
and, for strong refutation, on $\eta$.
The notation $\widetilde{\Omega}$ suppresses polylogarithmic
factors in $n$.
For the smoothed model, all literal-perturbation rates are fixed
positive constants.
\begin{table}[ht]
\centering \small
\renewcommand{\arraystretch}{1.4}
\begin{tabular}{|c|c|c|c|c|}
\hline
\textbf{CSP}
& \textbf{Regime for $m$}
& \textbf{Strength}
& \textbf{Eff./Exist.}
& \textbf{Reference} \\
\hline
3-SAT
& $\Omega(n^{7/5})$
& Refutation
& Existential
& \cite{FKO06} \\
\hline
3-$\XOR$
& $m\ge 2n$
& $1/n^{\Omega(1)}$-refutation
& Efficient
& Gaussian elimination \\
\hline
3-SAT
& $\Omega(n^{3/2})$
& Refutation
& Efficient
& Claimed in \cite{FKO06} \\
\hline
Boolean $k$-CSP
& $\widetilde{\Omega}(n^{\lceil k/2\rceil})$
& Strong
& Efficient
& \cite{COCF10} \\
\hline
$k$-$\XOR$
& $\widetilde{\Omega}(n^{k/2})$
& Strong
& Efficient
& \cite{BM16,AOW15} \\
\hline
Random $k$-CSP
& $\widetilde{\Omega}(n^{k/2})$
& Strong
& Efficient
& \cite{AOW15,RSS17} \\
\hline
Semi-random $k$-$\XOR$/CSP
& $\widetilde{\Omega}(n^{k/2})$
& Strong
& Efficient
& \cite{AGK21} \\
\hline
Smoothed $k$-CSP
& $\widetilde{\Omega}
   (n^{\frac{k}{2}-\frac{k-2}{2(k+8)}})$
& Refutation
& Existential
& \cite{GKM22} \\
\hline
\end{tabular}
\caption{Known results on deterministic and nondeterministic
CSP refutation.
For smoothed $3$-CSPs with constant smoothing rates,
\cite{GKM22} improves the general bound to
$\widetilde{O}(n^{7/5})$ constraints. Note that as instance models,
random $\subset$ semi-random $\subset$ smoothed, so
results for more general models subsume those for
more restricted ones.}
\label{tab:results}
\end{table}

The lower bounds of~\cite{KMOW17} show that polynomial-time
sum-of-squares refutation of fully random $k$-$\XOR$ requires
$\widetilde{\Omega}(n^{k/2})$ constraints, essentially matching
the corresponding upper bound.
These results provide evidence for the optimality of the spectral
threshold for strong refutation, but do not establish a lower bound
for unrestricted polynomial-time or nondeterministic algorithms.

\section{Equivalence between Agnostic PAC-learning \texorpdfstring{$\ckt$}{C} and Strong RRHS-refutation of \texorpdfstring{$\ckt^*$}{C*}}\label{sec:agnostic}

We prove the equivalence between agnostic PAC-learning and strong RRHS-refutation, extending~\cite[Theorem~5]{Vadhan17} (which establishes PAC-learning $\iff$ RRHS-refutation in the realizable setting) to the agnostic setting. Throughout, we use the notation in~\cite{Vadhan17}: an evaluation function $\mathsf{Eval}\colon \{0,1\}^s \times \{0,1\}^t \to \{0,1\}$ gives rise to a circuit class $\mathscr{C} = \{C_x\}_{x \in \{0,1\}^s}$ and its dual $\mathscr{C}^* = \{C^*_y\}_{y \in \{0,1\}^t}$, where $C_x(y) = C^*_y(x) = \mathsf{Eval}(x,y)$. We suppress unary encodings of the accuracy parameters from the displayed algorithm inputs, and all probabilities include the algorithms' internal randomness.

We begin by recalling the relevant definitions.

\begin{definition}[Agnostic PAC-learning, prediction version]\label{def:agnostic-pac}
  $\mathscr{C}$ is \emph{agnostically PAC-learnable} with sample complexity $m = m(s, t, 1/\eps)$ if there is a polynomial-time algorithm~$\calA$ such that for every $0<\eps\le 1/2$, every $s, t \in \mathbb{N}$, and every distribution~$\calD$ on $\{0,1\}^t \times \{0,1\}$, given $m+1$ i.i.d.~samples $(y_1, b_1), \ldots, (y_{m+1}, b_{m+1})$ from $\calD$, it holds that
  \[
    \Pr\!\Big[\calA\!\big(1^s, 1^t, (y_1, b_1), \ldots, (y_m, b_m),\; y_{m+1}\big) = b_{m+1}\Big]
    \;\ge\;
    \max_{x \in \{0,1\}^s} \Pr_{(y,b) \leftarrow \calD}\!\big[C_x(y) = b\big] - \eps.
  \]
  We say $\mathscr{C}$ is agnostically PAC-learnable if $m(s,t,1/\eps) \leq \mathrm{poly}(s, t, 1/\eps)$.
\end{definition}

\begin{remark}\label{rem:realizable-reduction}
  In the realizable case (i.e., $b = C_{x^*}(y)$ for some $x^*$), the maximum equals~$1$ and the definition reduces to requiring prediction probability $\ge 1 - \eps$, recovering~\cite[Definition~1]{Vadhan17} up to the parametrization of the constant.
\end{remark}

\begin{definition}[Strong RRHS-refutation]\label{def:strong-rrhs}
  For every $0<\gamma\le 1/2$, let $N=N(s,t,1/\gamma)$ satisfy $N \ge \lceil 2/\gamma^2 \rceil$.
  We say that $\mathscr{C}^*$ is \emph{strongly RRHS-refutable} with $N$ equations if there is a polynomial-time algorithm~$\calB$ such that for every such~$\gamma$ and every $y_1, \ldots, y_N \in \{0,1\}^t$:
  \begin{itemize}
    \item \textit{(Soundness)}
      For every $b_1, \ldots, b_N \in \{0,1\}$, if there exists $x \in \{0,1\}^s$ with
      \[
        \frac{1}{N}\big|\{i \in [N] : C^*_{y_i}(x) = b_i\}\big|
        \;\ge\;
        \tfrac{1}{2} + \gamma,
      \]
      then $\Pr\!\big[\calB(1^s, 1^t, (y_1, b_1), \ldots, (y_N, b_N)) = 1\big] \leq 1/3$.

    \item \textit{(Completeness)}
      If we sample i.i.d.~$b_1, \ldots, b_N \leftarrow \{0,1\}$, then
      $\Pr\!\big[\calB(1^s, 1^t, (y_1, b_1), \ldots, (y_N, b_N)) = 1\big] \ge 2/3$.
  \end{itemize}
  We say $\mathscr{C}^*$ is strongly RRHS-refutable if $N(s, t, 1/\gamma) \leq \mathrm{poly}(s, t, 1/\gamma)$.
\end{definition}

\begin{remark}\label{rem:standard-rrhs}
  At $\gamma=1/2$, the soundness condition reduces to that of standard RRHS-refutation~\cite[Definition~4]{Vadhan17}: it applies only when the system is \emph{fully satisfiable}. Strong refutation extends this condition to partial satisfiability.
\end{remark}

The main result of this appendix shows that these two notions are equivalent.

\begin{theorem}\label{thm:agnostic-rrhs}
  Let $\mathscr{C} = (\mathscr{C}^*)^*$ be a circuit class given by $\mathsf{Eval}\colon \{0,1\}^s \times \{0,1\}^t \to \{0,1\}$. Then:
  \begin{itemize}
    \item If $\mathscr{C}$ is agnostically PAC-learnable with sample complexity $m(s,t,1/\eps)$, then $\mathscr{C}^*$ is strongly RRHS-refutable using $N=O(m(s,t,3/\gamma)/\gamma^3)$ equations.

    \item If $\mathscr{C}^*$ is strongly RRHS-refutable using $N(s,t,1/\gamma)$ equations, then $\mathscr{C}$ is agnostically PAC-learnable with sample complexity $m=\mathrm{poly}(N(s,t,3/\eps))$.

  \end{itemize}
  In particular, $\mathscr{C}$ is agnostically PAC-learnable if and only if $\mathscr{C}^*$ is strongly RRHS-refutable.
\end{theorem}

\begin{proof}
We prove each direction in turn.

\paragraph{Agnostic PAC-learning $\implies$ Strong RRHS-refutation.}
Let $\calA$ be the agnostic PAC-learner with sample complexity $m(s, t, 1/\eps)$. We may assume that $m(s,t,1/\eps)\ge1$, since the learner can ignore an additional sample. Given the refutation threshold $1/2\ge \gamma > 0$, set $\eps = \gamma/3$, $m_0 = m(s, t, 3/\gamma)$, and $q=\lceil 10m_0/\gamma\rceil$.

Define the following one-block procedure~$\calB_0$ on input $(1^s, 1^t, (y_1, b_1), \ldots, (y_q, b_q))$ as follows:
\begin{itemize}
  \item Sample i.i.d.~$i_1, \ldots, i_{m_0+1} \leftarrow [q]$.
  \item Output~$1$ if and only if $\calA(1^s, 1^t, (y_{i_1}, b_{i_1}), \ldots, (y_{i_{m_0}}, b_{i_{m_0}}), y_{i_{m_0+1}}) \neq b_{i_{m_0+1}}$ or $i_{m_0+1}\in\{i_1,\ldots,i_{m_0}\}$.
\end{itemize}

\begin{claim}
    Let $E$ be the event of collision: $E:=\{i_{m_0+1}\in \{i_1,\dots,i_{m_0}\}\}$, then it holds that $\gamma/11\le \Pr[E] \le \gamma/10$.
\end{claim}
\begin{proof}
    Let $p:=\Pr[E]=1-(1-1/q)^{m_0}$. By the union bound,
    \[
        p\le\frac{m_0}{q}\le\frac{\gamma}{10}.
    \]
    For the other direction, writing $u=m_0/q$, we have
    \[
        u
        \ge
        \frac{\gamma m_0}{10m_0+\gamma}
        \ge
        \frac{\gamma}{10+\gamma}.
    \]
    Hence
    \[
        p
        \ge 1-e^{-u}
        \ge\frac{u}{1+u}
        \ge\frac{\gamma}{10+2\gamma}
        \ge\frac{\gamma}{11},
    \]
    where the last inequality uses $\gamma\le1/2$.
\end{proof}

\begin{description}
    \item[Soundness of $\calB_0$:] Suppose there exists $x^* \in \{0,1\}^s$ with $\frac{1}{q}|\{i : C^*_{y_i}(x^*) = b_i\}| \ge \frac{1}{2} + \gamma$. Define the empirical distribution~$\calD$ as uniform on the multiset $\{(y_i, b_i) : i \in [q]\}$. Then
    \[
      \max_x \Pr_{(y,b) \leftarrow \calD}\!\big[C_x(y) = b\big]
      \;\ge\;
      \Pr_\calD\!\big[C_{x^*}(y) = b\big]
      = \frac{1}{q}\big|\{i : C^*_{y_i}(x^*) = b_i\}\big|
      \;\ge\; \frac{1}{2} + \gamma.
    \]
    Since $i_1,\dots, i_{m_0+1}$ are sampled independently with replacement, the samples fed to~$\calA$ are i.i.d.~from~$\calD$, and the agnostic learning guarantee gives
    \[
      \Pr[\calA \text{ correct}]
      \;\ge\;
      \Big(\frac{1}{2} + \gamma\Big) - \eps
      = \frac{1}{2} + \gamma - \frac{\gamma}{3}
      = \frac{1}{2} + \frac{2\gamma}{3}.
    \]
    Accounting for collisions:
    \[
      \Pr[\calB_0 = 1]
      = \Pr[\calA \text{ incorrect} \lor E]
      \;\leq\;
      \Big(\frac{1}{2} - \frac{2\gamma}{3}\Big)  + \frac{\gamma}{10}
      = \frac{1}{2} - \frac{17\gamma}{30}.
    \]
    \item[Completeness of $\calB_0$:] Suppose $b_1, \ldots, b_q \leftarrow \{0,1\}$ are i.i.d.~samples. Conditioned on $i_{m_0+1} \notin \{i_1, \ldots, i_{m_0}\}$, the bit $b_{i_{m_0+1}}$ is uniform and independent of the view of~$\calA$, so $\calA$ predicts it correctly with probability exactly~$1/2$. Therefore
    \[
      \Pr[\calB_0 = 1]
      = \frac{1}{2}\Pr[i_{m_0+1} \notin \{i_1, \ldots, i_{m_0}\}]+ \Pr[i_{m_0+1} \in \{i_1, \ldots, i_{m_0}\}]
      \;\ge\; \frac{1}{2}+\frac{\gamma}{22}.
    \]   
\end{description}

\paragraph{Amplification.}
Let $R$ be the smallest odd integer satisfying $R\ge242\ln 3/\gamma^2$, and set $N=Rq$. Define the final refuter~$\calB$ on $N$ equations by
partitioning them into $R$ disjoint blocks of size~$q$, independently
running $\calB_0$ on each block to obtain bits
$X_1,\ldots,X_R$, and outputting their majority.

For soundness, suppose some $x^*\in\{0,1\}^s$ satisfies at least a
$1/2+\gamma$ fraction of all $N$ equations. Let $a_r$ be the fraction
of equations in block~$r$ satisfied by $x^*$. Then
\[
    \frac1R\sum_{r=1}^R a_r\ge\frac12+\gamma.
\]
Writing $p_r=\Pr[X_r=1]$, the preceding argument gives
\[
    p_r
    \le
    1-a_r+\frac{\gamma}{3}+\frac{\gamma}{10}.
\]
Therefore
\[
\begin{aligned}
    \frac1R\sum_{r=1}^R p_r
    \le
    1-\frac1R\sum_{r=1}^R a_r
      +\frac{\gamma}{3}+\frac{\gamma}{10} \le\frac12-\frac{17\gamma}{30}.
\end{aligned}
\]
For a fixed soundness instance, the variables
$X_1,\ldots,X_R$ are independent because independent internal
randomness is used across the blocks. Hence Hoeffding's inequality
gives
\[
\begin{aligned}
    \Pr[\calB=1]
    \le
    \exp\left(-2R\left(\frac{17\gamma}{30}\right)^2\right) \le
    \exp\left(-\frac{R\gamma^2}{242}\right)\le\frac{1}{3}.
\end{aligned}
\]

For completeness, suppose all right-hand sides are independent
uniform bits. Since the blocks are disjoint and use independent
internal randomness, $X_1,\ldots,X_R$ are independent. Moreover,
for every $r\in[R]$,
\[
    \Pr[X_r=1]\ge\frac12+\frac{\gamma}{22}.
\]
Thus Hoeffding's inequality gives
\[
    \Pr[\calB=0]
    \le
    \exp\left(-2R\left(\frac{\gamma}{22}\right)^2\right)
    =
    \exp\left(-\frac{R\gamma^2}{242}\right)
    \le\frac13.
\]
Therefore $\calB$ has soundness at most~$1/3$ and completeness at
least~$2/3$. The total number of equations is
\[
    N=Rq
    =
    O\left(\frac{m_0}{\gamma^3}\right)
    =
    \mathrm{poly}(s,t,1/\gamma).
\]
Moreover, $N=Rq\ge R\ge 2/\gamma^2$, as required by
\autoref{def:strong-rrhs}.
\paragraph{Strong RRHS-refutation $\implies$ Agnostic PAC-learning.}
This is the more interesting direction. We follow the strategy in~\cite{Vadhan17} --- applying Yao's next-bit prediction \cite[Lemma~6]{Vadhan17} --- but the analysis must account for the fact that the labels~$b_i$ need not come from any concept~$C_x$.

Let $\calB$ be the strong RRHS-refuter for~$\mathscr{C}^*$ using $N = N(s, t, 1/\gamma)$ equations with threshold~$\gamma$. It suffices to consider $0<\gamma<1/4$. We construct an agnostic PAC-learner for~$\mathscr{C}$ with error $3\gamma$.

Fix any distribution~$\calD$ on $\{0,1\}^t \times \{0,1\}$ with
$\max_x \Pr_{(y,b)\leftarrow\calD}[C_x(y) = b] \ge 1/2 + 2\gamma$, and let $(Y_1, B_1), \ldots, (Y_N, B_N) \leftarrow \calD$ be i.i.d.~samples. We first observe that the refuter~$\calB$ can be turned into a distinguisher between real and random labels. Define
\[
  T\!\big((y_1, \ldots, y_N),\, (b_1, \ldots, b_N)\big)
  \;:=\;
  1 - \calB\!\big(1^s, 1^t, (y_1, b_1), \ldots, (y_N, b_N)\big),
\]
so that $T$ outputs~$1$ when $\calB$ does not refute. Write $\mathbf{Y} = (Y_1, \ldots, Y_N)$, $\mathbf{B} = (B_1, \ldots, B_N)$, and let $\mathbf{R} = (R_1, \ldots, R_N)$ be independent uniform bits, also independent of $(\mathbf{Y},\mathbf{B})$.

\begin{claim}\label{claim:distinguishing}
  $T$ distinguishes $(\mathbf{Y}, \mathbf{B})$ from $(\mathbf{Y}, \mathbf{R})$ with advantage $\ge 1/4$.
\end{claim}

\begin{proof}[Proof of \autoref{claim:distinguishing}]
  \emph{When bits are real (i.e., $\mathbf{B}$):}
  Let $x^*$ achieve $\Pr_{(y,b)\leftarrow\calD}[C_{x^*}(y) = b] \ge 1/2 + 2\gamma$. By Hoeffding's inequality and the lower bound on $N$ in~\autoref{def:strong-rrhs},
  \[
    \Pr\!\bigg[\frac{1}{N}\big|\{i : C^*_{Y_i}(x^*) = B_i\}\big| \ge \frac{1}{2} + \gamma\bigg]
    \;\ge\; 1-e^{-2N\gamma^2}
    \;\ge\; \frac{11}{12}.
  \]
  When this event holds, the soundness of the strong RRHS-refuter gives $\Pr[\calB = 1] \leq 1/3$, hence $\Pr[T = 1] \ge 2/3$. Therefore
  \[
    \Pr\!\big[T(\mathbf{Y}, \mathbf{B}) = 1\big]
    \;\ge\; \frac{11}{12} \cdot \frac{2}{3}
    \;=\; \frac{11}{18}.
  \]

  \emph{When bits are random (i.e., $\mathbf{R}$):}
  By completeness, $\Pr[\calB = 1] \ge 2/3$, so $\Pr[T(\mathbf{Y}, \mathbf{R}) = 1] \leq 1/3$.

  \medskip\noindent
  Thus the distinguishing advantage is
  \[
    \Pr\!\big[T(\mathbf{Y}, \mathbf{B}) = 1\big] - \Pr\!\big[T(\mathbf{Y}, \mathbf{R}) = 1\big]
    \;\ge\; \frac{11}{18} - \frac{1}{3}
    = \frac{5}{18}
    > \frac{1}{4}. \qedhere
  \]
\end{proof}

By Yao's lemma \cite{Yao82} (see \cite[Lemma~6]{Vadhan17}), this distinguisher with advantage $\alpha \ge 1/4$ yields a next-bit predictor~$T'$ satisfying
\[
  \Pr\!\big[T'(\mathbf{Y}, B_1, \ldots, B_{I-1}) = B_I\big]
  \;\ge\; \frac{1 + \alpha/N}{2}
  \;\ge\; \frac{1}{2} + \frac{1}{8N},
\]
where $I \leftarrow [N]$ is uniform. Concretely, $T'$ on input $((y_1, \ldots, y_N), b_1, \ldots, b_{i-1})$ draws independent uniform bits $r_i, r_{i+1}, \ldots, r_N$, runs $T\!\big((y_1, \ldots, y_N),\, (b_1, \ldots, b_{i-1}, r_i, \ldots, r_N)\big)$, and outputs~$r_i$ if $T$ accepts and $1-r_i$ otherwise. Crucially, $T'$ is constructed purely from~$\calB$ and does not depend on the unknown distribution~$\calD$ or any particular concept~$C_{x^*}$.

We now reinterpret $T'$ as a weak agnostic prediction algorithm $\calA_{\mathrm{weak}}$ using $N-1$ labeled examples. Given $(y_1,b_1),\ldots,(y_{N-1},b_{N-1})$ and a challenge point~$y$, choose $I\leftarrow[N]$ uniformly, insert $y$ as the $I$th point in $(y_1,\ldots,y_{N-1})$, and give $T'$ the labels of the points preceding~$y$. The labels of the remaining points are simply ignored. By exchangeability, the resulting sequence of $N$ points is distributed exactly as $\mathbf{Y}$. Thus
\[
  \Pr\!\big[\calA_{\mathrm{weak}} \text{ predicts correctly}\big]
  \;\ge\; \frac{1}{2} + \frac{1}{8N}
\]
whenever
\[
  \min_x\Pr_{(y,b)\leftarrow\calD}[C_x(y)\ne b]
  \le \frac12-2\gamma.
\]
Set $\eta=1/(8N)$; note that $0<\eta/4<2\gamma<1/2$ by the lower bound on~$N$. In the terminology of~\cite[Definition~4]{KMV08}, $\calA_{\mathrm{weak}}$ is an $(N-1,2\gamma,\eta)$-expected weak agnostic guesser. By~\cite[Lemma~4]{KMV08}, it yields an $(M,2\gamma,\eta/4)$-weak agnostic learner (with the guesser's internal randomness fixed in the resulting hypotheses), where
\[
  M
  =\frac{N-1}{\eta}
   +\frac{32}{\eta^2}\log\frac{20}{\eta}
  =O(N^2\log N).
\]
Applying the agnostic boosting theorem~\cite[Theorem~3]{KMV08} with accuracy parameter~$\gamma/2$ and failure probability~$\gamma/2$, we obtain a hypothesis~$h$ such that, except with probability at most~$\gamma/2$,
\[
  \Pr_{(y,b)\leftarrow\calD}[h(y)\ne b]
  \le
  \min_x\Pr_{(y,b)\leftarrow\calD}[C_x(y)\ne b]
  +\frac{5\gamma}{2}.
\]
Using the trivial error bound~$1$ on the failure event, the resulting prediction algorithm has expected error at most the optimum plus~$3\gamma$. The sample complexity is polynomial in $N$; here we use again that $N=\Omega(1/\gamma^2)$. Taking $\gamma=\eps/3$ proves the result. \qedhere
\end{proof}

\section{Refuter Implies Learner}\label{appendix: refuter-implies-learner}

In this appendix, we justify our definition of nondeterministic PAC learning
(\autoref{def: nondet PAC learning}) by showing that our nondeterministic refuter notion~(\autoref{def: nondet PAC learning})
yields PAC learning with nondeterministic hypotheses in the sense
of~\cite{Zhang2022NondeterministicVariants}. The learner construction and
hybrid argument below are the same as in
\cite[Lemma~4.2.1]{Zhang2022NondeterministicVariants}, which reformulates
\cite[Theorem~7]{BFKL93}. We include the proof to record the parameters in
our setting; averaging directly over the random hybrid index gives confidence
$p/L$ and accuracy $1/2+p/(2L)$. We begin by restating
the two definitions for the reader's convenience.

\begin{definition}[PAC Learning with Nondeterministic Hypotheses~{\cite[Definition~4.1.1]{Zhang2022NondeterministicVariants}}]
Let $\ckt$ be a class of Boolean functions and let
$\calD=(\calD_n)_{n\ge1}$ be a distribution ensemble, where $\calD_n$ is
supported on $\{0,1\}^n$.
A randomized $\P/\poly$ algorithm $\mathfrak{L}$ with oracle access to $f$ and sample access to $\calD_n$ is a \emph{learner with
nondeterministic hypothesis} for $(\ckt,\calD)$ at confidence $\delta$ and
accuracy $1-\eps$ if for every $f\in\ckt$ and all sufficiently large $n$:
\begin{itemize}
    \item $\mathfrak{L}^f(1^n,w)$ outputs a string $\abra{h}$ which is a
    description of a nondeterministic (or co-nondeterministic) circuit
    $h_w:\{0,1\}^n \to \{0,1\}$ of size $\poly(n)$;
    \item with probability $\ge \delta$ over the random tape $w$, the
    produced hypothesis satisfies
    $\Pr_{x\gets \calD_n}[h_w(x)=f(x)]\ge 1-\eps$.
\end{itemize}
We require $h_w = \mathfrak{L}^f(1^n,w)$ to be
evaluable\footnote{For every $f$ and $w$,
$h_w=\mathfrak{L}^f(1^n,w)$ is $\P/\poly$-evaluable if there is another
circuit family $\Eval\in\P/\poly$ such that for every possible output
$h_w=\mathfrak{L}^f(1^n,w)$ of $\mathfrak{L}$ and every
$x\in\{0,1\}^n$ given as input to $\Eval$, $\Eval$ computes $h_w(x)$.
When $h_w$ is (co-)nondeterministic or randomized, $\Eval$ is allowed to be
(co-)nondeterministic or randomized, respectively.} by an algorithm of
matching nondeterministic type (this is automatic when $h_w$ is itself a
circuit of polynomial size).
\end{definition}

\begin{definition}[Nondeterministic PAC Learning; restatement of~\autoref{def: nondet PAC learning}]
\label{def:nondet-pac-restated}
Let $\ckt$ be a class of Boolean functions, $L\ge 1$, $p\in (0, 1)$, and
let $\calD=(\calD_n)_{n\ge1}$ be a distribution ensemble, where $\calD_n$
is a distribution over $\{0, 1\}^n$. A nondeterministic algorithm $\calA$
is said to \emph{PAC-learn} $\ckt$ over input distribution $\calD_n$ using
$L$ samples with success probability $p$ if:
    \begin{itemize}
        \item {\it (Soundness)} for every $x_1, \dots, x_L\in \{0, 1\}^n$
        and labels $b_1, \dots, b_L \in \{0, 1\}$, if there exists a
        function $C \in \ckt$ such that $C(x_i) = b_i$ for every $i\in [L]$,
        then $\calA((x_1, b_1), \dots, (x_L, b_L))$ rejects.
        \item {\it (Completeness)} If we sample i.i.d.~$x_i\gets \calD_n$
        and $b_i \gets \{0, 1\}$, then
        \[\Pr[\calA((x_1, b_1), \dots, (x_L, b_L))\text{ accepts}] \ge p.\]
    \end{itemize}
    We also say that a nondeterministic algorithm $\calA$
    \emph{distribution-freely} PAC-learns $\ckt$ with success probability
    $p$ if it also satisfies the following stronger (Completeness) condition:
    \begin{itemize}
        \item {\it (Completeness')} For every
        $x_1, \dots, x_L\in \{0, 1\}^n$, if we choose i.i.d.~labels
        $b_1, \dots, b_L \gets \{0, 1\}$, then
        \[\Pr[\calA((x_1, b_1), \dots, (x_L, b_L))\text{ accepts}]\ge p.\]
    \end{itemize}
\end{definition}

\begin{theorem}[Refuter Implies Learner, adapted from~{\cite[Theorem~7]{BFKL93}}
and~{\cite[Lemma~4.2.1]{Zhang2022NondeterministicVariants}}]
Let $L=L(n)$ be polynomially bounded and let $p=p(n)\in(0,1)$. Let $\calA$
be a nondeterministic polynomial-time refuter parameterized by
$(\ckt,\calD_n,L,p)$ as in~\autoref{def:nondet-pac-restated}. Then there is a
randomized $\P/\poly$ learner $\mathfrak{L}$ for $(\ckt,\calD)$ with
confidence $\delta \ge p/L$ and accuracy
$1-\eps\ge 1/2+p/(2L)$, where each output hypothesis is either a
nondeterministic or co-nondeterministic circuit, and $\mathfrak{L}$ makes at
most $L-1$ queries to $f$.
\end{theorem}
\begin{proof}
We will construct the learner with access to the refuter.

\begin{algorithm2e}[H]
\DontPrintSemicolon
\SetAlgoLined
\caption{The randomized oracle learner $\mathfrak{L}^f(1^n)$}
\label{alg:learner}

\KwIn{$1^n$ and oracle access to $f$ and sample access to $\calD_n$}
\KwOut{A hypothesis $h$}

Sample $i \gets [L]$ uniformly at random\;
Sample $r_1,\ldots,r_L \gets \{0,1\}$ independently and uniformly\;
Sample $x_1,\ldots,x_{i-1},x_{i+1},\ldots,x_L \gets \mathcal{D}_n$ independently\;
Query the oracle to obtain $f(x_1),\ldots,f(x_{i-1})$\;

\Return the hypothesis $h$ defined, for each $x \in \{0,1\}^n$, by
\[
\begin{aligned}
h(x)
=
r_i \oplus
\mathcal{A}\bigl(
  &(x_1,f(x_1)),\ldots,(x_{i-1},f(x_{i-1})), (x,r_i),\\
  &(x_{i+1},r_{i+1}),\ldots,(x_L,r_L)
\bigr).
\end{aligned}
\]
\end{algorithm2e}
Note that for nondeterministic $\calA$, ``$\calA(\cdot) = 1$'' is an
$\exists$-statement over the witness.
\begin{itemize}
    \item If $r_i=0$, then $h(x) = \calA(\cdots)$, so $h(x)=1$ iff $\calA$
    has an accepting path, i.e.,~$h$ is nondeterministic.
    \item If $r_i=1$, then $h(x)= 1\oplus \calA(\cdots)$, so $h(x)=1$ iff
    $\calA$ has no accepting path, i.e.,~$h$ is co-nondeterministic.
\end{itemize}
Since $L$ is polynomially bounded and $\calA$ runs in nondeterministic
polynomial time, hardwiring the sampled points, labels, and oracle answers
into a circuit for $\calA$ gives a polynomial-size hypothesis, and its
description can be produced by a randomized $\P/\poly$ learner.

Now we define $(L+1)$ hybrid distributions. For
$j\in \cbra{0,1,\dots, L}$, define the random variable
\begin{align*}
    H_j:=\pbra{(x_1,f(x_1)),\dots,(x_j,f(x_j)),(x_{j+1},r_{j+1}),\dots,(x_L,r_L)},
\end{align*}
where $x_k\gets \calD_n$ i.i.d.~and $r_k\gets \{0,1\}$ uniform i.i.d.

Let
\begin{align*}
    \alpha_j:=\Pr[\calA(H_j)\text{ accepts}].
\end{align*}
Note that for $H_0$, all labels are uniformly at random and are independent
of $x_k$'s, so $\alpha_0\ge p$ by completeness. Also, since $H_L$ is
realizable by $f\in \ckt$, $\alpha_L = 0$ by soundness. Hence, by
telescoping, we have
\begin{align*}
    \sum^{L-1}_{i=0}(\alpha_i-\alpha_{i+1}) = \alpha_0 -\alpha_L \ge p.
\end{align*}
We now proceed to bound the accuracy of the hypothesis at a fixed $i$.

Fix $i\in [L]$ chosen by the learner. Consider the experiment in which
$x\gets \calD_n$ is the evaluation point and $r_i$, $x_k$ $(k \neq i)$,
$r_k$ $(k>i)$ are sampled as in the construction. The string fed to $\calA$
when computing $h(x)$ is
\begin{align*}
    Z:= \pbra{(x_1,f(x_1)),\dots,(x_{i-1},f(x_{i-1})),(x,r_i),(x_{i+1},r_{i+1}),\dots,(x_L,r_L)}.
\end{align*}
Conditioning on $r_i = f(x)$, $Z$ is distributed as $H_i$ (recall that
$x\gets \calD_n$ in our experiment). Hence
\[\Pr[\calA(Z) = 1 \mid r_i = f(x)] = \alpha_i.\]
On the other hand, let $\widetilde{H}_i$ denote the distribution of $Z$
conditioning on $r_i \neq f(x)$. That is, $\widetilde{H}_i$ is $H_i$ with
the $i$-th label flipped. The distribution of $H_{i-1}$ is the equal mixture
of $H_i$ and $\widetilde{H}_i$. Hence, we have
\begin{align*}
    \alpha_{i-1} = \alpha_i/2 + \Pr_{\widetilde{H}_i}[\calA = 1]/2,
\end{align*}
or in other words,
\begin{align*}
    \Pr_{\widetilde{H}_i}[\calA = 1] = 2\alpha_{i-1} - \alpha_i.
\end{align*}
Since $h(x) = f(x)$ iff $\calA(Z) = r_i \oplus f(x)$, we have
\begin{align*}
    \Pr[h(x) = f(x)] &= \Pr[\calA(Z) = 0 \mid r_i = f(x)]/2 + \Pr[\calA(Z)=1 \mid r_i\neq f(x)]/2 \\
    & = (1-\alpha_i)/2 + (2\alpha_{i-1}-\alpha_i)/2 \\
    & = 1/2 + (\alpha_{i-1}-\alpha_i).
\end{align*}

Let $w$ denote all the randomness used by the learner, including the choice
of $i$. Averaging also over an independent $x\gets\calD_n$, we obtain
\begin{align*}
    \Ex_{w,x}\sbra{\mathbf{1}[h(x) = f(x)]}
    = \frac{1}{L}\sum^L_{i=1}\pbra{\frac{1}{2}+\alpha_{i-1}-\alpha_i}
    = \frac{1}{2} + \frac{\alpha_0 - \alpha_L}{L}
    \ge \frac{1}{2} + \frac{p}{L}.
\end{align*}

Let $X = X(w):= \Pr_{x\gets \calD_n}[h_w(x) = f(x)]$ be a random variable
in the learner's randomness $w$. We have $\Ex[X] \ge 1/2 + p/L$. Writing
\[
\delta := \Pr\left[X\ge \frac12+\frac{p}{2L}\right],
\]
since $0\le X\le1$, we have
\begin{align*}
\frac12+\frac pL
\le \Ex[X]
\le \delta\cdot1+(1-\delta)\left(\frac12+\frac{p}{2L}\right).
\end{align*}
Therefore,
\begin{align*}
\delta \ge \frac{p/(2L)}{1/2-p/(2L)} \ge \frac{p}{L}.
\end{align*}

Hence, with probability $\ge p/L$ over the learner's randomness, the
produced hypothesis agrees with $f$ on a $(1/2+p/(2L))$ fraction of inputs.
\end{proof}

\newpage
\listoffixmes
\end{document}